\documentclass[12pt]{article}
\newcommand{\blind}{0}
\usepackage{easyReview}
\usepackage{amsmath,amsthm}
\usepackage{amsfonts}
\usepackage{amssymb}
\usepackage{bbm}
\usepackage{mathtools}
\usepackage{comment}

\usepackage{url}

\usepackage[version=4]{mhchem}
\usepackage{stmaryrd}

\usepackage{subcaption} % For sub-figures\
\usepackage{graphicx}
\usepackage{pgfplots}
\usepackage[all]{nowidow}
\usepackage[utf8]{inputenc}
\usepackage{tikz}
\usepackage{multicol}
\usepackage{algpseudocode,algorithm,algorithmicx}
\usepackage{scalerel}
\usepackage{natbib}
\usepackage[normalem]{ulem}

\usepackage[english]{babel}

\usepackage{mathdots}
\usepackage{yhmath}
\usepackage{cancel}
\usepackage{color}
\usepackage{array}
\usepackage{multirow}
\usepackage{gensymb}
\usepackage{tabularx}
\usepackage{booktabs}
\usetikzlibrary{fadings}
\usetikzlibrary{patterns}
\usetikzlibrary{shadows.blur}
\usepackage{enumerate}
\usepackage[algo2e]{algorithm2e}

\newtheorem{theorem}{Theorem}
\newtheorem{lemma}{Lemma}

\newtheorem{assumption}{Assumption}

\newtheorem{remark}{Remark}
\newtheorem{definition}{Definition}

\newcommand{\Z}{\mathbb{Z}}

\newcommand{\R}{\mathbb{R}}

\newcommand{\X}{\mathcal{X}}

\newcommand{\A}{\mathcal{A}}
\newcommand{\M}{\mathcal{M}}
\newcommand{\N}{\mathcal{N}}

\newcommand{\cP}{\mathcal{P}}

\newcommand{\cH}{\mathcal{H}}

\newcommand{\D}{\mathcal{D}}
\newcommand{\F}{\mathcal{F}}

\newcommand{\li}{\langle}
\newcommand{\ri}{\rangle}
\newcommand{\fr}{{\lfloor nr\rfloor}}

\newcommand{\fb}{{\lfloor nb \rfloor}}

\DeclareMathOperator*{\argmax}{arg\,max}

\DeclareMathOperator*{\var}{Var}
\DeclareMathOperator*{\cov}{Cov}

\DeclareMathOperator*{\tr}{tr}
\DeclareMathOperator{\E}{\mathbb{E}}

\usepackage{bbm}
\newcommand{\id}{\mathbbm{1}}
\usepackage[nodisplayskipstretch]{setspace}
\let\svthefootnote\thefootnote
\newcommand\freefootnote[1]{%
	\let\thefootnote\relax%
	\footnotetext{#1}%
	\let\thefootnote\svthefootnote%
}

\begin{document}

	\def\spacingset#1{\renewcommand{\baselinestretch}%
		{#1}\small\normalsize} \spacingset{1}

	%%%%%%%%%%%%%%%%%%%%%%%%%%%%%%%%%%%%%%%%%%%%%%%%%%%%%%%%%%%%%%%%%%%%%%%%%%%%%%
	
	\if0\blind
	{
	\title{\bf Change-Point Detection for Object-valued Time Series}

	\author{Yi Zhang$^\text{ a}$, \hspace{.2cm}Changbo Zhu$^{\text{ b}}$\thanks{Corresponding author}\hspace{.4cm}and Xiaofeng Shao$^\text{ c}$
\\
	}
\date{%
	$^\text{ a}$\small Department of Statistics, University of Illinois at Urbana-Champaign, 605 E. Springfield Ave. Champaign, IL, USA. Email: {\tt yiz19@illinois.edu} \normalsize\\%
	$^\text{ b}$\small Department of Applied and Computational Mathematics and Statistics, University of Notre Dame, 101H Crowley Hall, Notre Dame, IN, USA. Email: {\tt czhu4@nd.edu} \normalsize\\%
	$^\text{ c}$\small Department of Statistics and Data Science, and Department of Economics, Washington University in St. Louis, One  Brookings Dr, St. Louis, MO, USA. Email: {\tt shaox@wustl.edu} \normalsize\\%\\[2ex]%
	%\today	
}
	\maketitle
} \fi
	
	\if1\blind
	{
    
		\bigskip
		\bigskip
		\bigskip
		\begin{center}
			{\LARGE\bf Change-Point Detection for Object-valued Time Series}
		\end{center}
		\medskip
	} \fi
	
	\bigskip
	\begin{abstract}
 This article is concerned with change point detection for object-valued data that reside
in a metric space, which has attracted some recent interests in statistics and econometrics
literature. The existing methods either focus on independent data or can only detect change
in the Fr\'echet mean or variance. In this paper, we propose a self-normalization (SN, hereafter) based statistic for detecting a shift in the marginal distribution of object-valued time
series. Our test is universally applicable to a wide range of object-valued data, such as distributional and network data, and can accommodate weak serial dependence. In addition the
proposed test statistic is almost tuning parameter free, has pivotal limiting null distribution
and only uses the pairwise distances. When combined with the Wild Binary Segmentation
algorithm (WBS, hereafter), our statistic can be used to estimate the number and locations
of multiple change points. Asymptotic results for our SN based statistic are derived under
both null and local alternatives in the single change point setting. For the first time, the
WBS estimation consistency is shown for a broad class of object-valued time series and in
a nonparametric setting, which requires new non-standard theoretical arguments. Extensive
numerical experiments and real data analysis are conducted to illustrate the effectiveness
and broad applicability of our proposed method.
	\end{abstract}
	
	\noindent%
	\vfill
			%\freefootnote{We would like to thank the Editor, the Associate Editor, and the anonymous reviewers for their valuable comments and constructive suggestions, which lead to significant improvements in the paper. } 
			%\freefootnote{ Shao's research is partially supported by NSF grants DMS-2526477 and  DMS-2524356; Zhu's research is partially supported by NSF grant DMS-2412832.} 
	\newpage
	
	\spacingset{1.65} % DON'T change the spacing!
	
	\section{Introduction}
	
	Statistical analysis of non-Euclidean data (object-valued data or random objects) that reside in a metric space is gradually emerging as an important  branch of functional data analysis in econometrics and statistics, motivated by increasing encounter of such data in many modern applications. Examples include 
the analysis of sequences of age-at-death distributions over calendar years \citep{mazzuco2015,shang2017},
covariance matrices in the analysis of diffusion tensors in medical imaging \citep{dryden2009non}, compositional brain volumes \citep{https://doi.org/10.1002/hbm.26271, zhou2023network} and  graph Laplacians of networks \citep{ginestet2017,daren2021}.
One of the main challenges in dealing with such data is that the usual vector/Hilbert space operations, such as
projection and inner product may not be well defined and only the distance between two non-Euclidean
data objects is available. Despite the challenge, the list of papers that propose new statistical techniques to analyze 
non-Euclidean data has been growing. See \cite{dubey2019frechet, dubey2020frechet}, \cite{petersen2019}, \cite{tucker2022}, \cite{zhou2022network}, \cite{zhang2022}, \cite{zhu2023geodesic, zhumuller2024JoE}, \cite{Jiang2024Biometrika,jiang2024two}, among others. 

In this paper, we shall focus on offline change-point detection for object-valued time series.
Such temporally dependent objects are prevalent in practice and well-known examples include yearly age-at-death distributions for countries in Europe \citep{dubey2020frechet}, annual observations of compositions of energy sources for power generation in the US \citep{zhumuller2024JoE},  and daily Pearson correlation matrices for five cryptocurrencies \citep{jiang2024two}. Also see Section~\ref{sec5_real} for additional data examples. The literature for offline change-point detection of object-valued  data is quite recent. For example,  \cite{chenzhang2015} and \cite{chuchen2019} developed 
graph-based scan statistics for both the single change-point and the changed interval alternatives to detect the distributional change for independent object data; see \cite{chen2023graph} for a recent review on graph-based change point detection methods. Recently,	\cite{dubey2020frechet} proposed a novel test for a change point in the Fr\'{e}chet mean and variance \citep{frechet1948elements} for a sequence of independent random objects and adopted binary segmentation for change point estimation. \cite{dubey2022depth} proposed the concept of distance profiles, which fully characterize the distribution of random objects under conditions that basically require the metric space to be of strong negative type (see proposition 1 therein). Building on distance profile, \cite{dubey2023change} developed a test which can detect distributional change for serially independent random objects. Note that   theoretical results developed for the above-mentioned methods are all for single change point detection. To accommodate serial  dependence, 
\cite{jiang2024two} developed an self-normalization (SN, hereafter) based test by extending the SN idea \citep{shao2010,shaozhang2010,shao2015} from vector time series to object-valued time series. However, the method in \cite{jiang2024two} is also limited to capture change point in Fr\'{e}chet mean or variance only and their test may not be powerful when the change is in high-order characteristics of distributions; see numerical evidence in Section~\ref{sec4simu}.

In contrast to the limited literature for change-point detection of object-valued data, the same problem for functional time series, which resides in a Hilbert space, has been extensively studied. For example, single and multiple change point detection probem for functional data has been studied in \cite{zhang2011,aston2012detecting,shar2016,aue2018detecting,jiao2023break,chiou2019identifying,rice2022consistency,harris2022scalable}; see Chapter 8 in \cite{horvath2024change} for a comprehensive review of the main developments in this area. It is worth noting that these methods can not be easily extended to the object-valued time series setting, due to the lack of algebraic operations for metric space valued data.

% {\color{blue}
% In contrast to the limited literature for change-point detection of object-valued data, the same problem for functional time series, which resides in a Hilbert space, has been extensively studied. For example, a functional principle component analysis (fPCA) based detection method is proposed in \cite{aston2012detecting} and fully functional approaches are proposed in \cite{shar2016,aue2018detecting}. Recently, \cite{jiao2023break} proposed a CUSUM based procedure to test for a change point in the marginal covariance function. For multiple change point detection, \cite{chiou2019identifying} used a two-stage procedure to estimate multiple changes in the mean functions, and the estimation consistency of the traditional binary segmentation procedure in detecting and localizing change points in the mean functions, based on the norm of a functional analog of a CUSUM process, is proved in \cite{rice2022consistency}. To detect changes in the mean and covariance functions simultaneously, \cite{harris2022scalable} proposed to project the original functional data into two univariate time series and use augmented fused lasso and univariate CUSUM statistics to detect the change points locally; see Chapter 8 in \cite{horvath2024change} for a comprehensive review of the main theoretical developments in this area. }

The main goal of this paper is to fill in the gap and develop new testing and estimation procedures that target the shift in marginal distributions of object-valued time series with serial dependence. On the testing front, we propose a non-parametric method that can detect distributional changes for metric space valued time series based on sample splitting and self-normalization. 
To construct our test statistic, the object-valued data are first mapped into a Hilbert space, and then projected along the direction of change. After that, an SN-based test statistic is constructed using the projected one-dimensional sequence. In the single change point setting, the statistic we developed is nearly tuning parameter free (except for a cut-off parameter which restricts the true change point from being too close to the end-points) and has pivotal limiting null distribution.
On the estimation front,  we utilize the WBS algorithm \citep{wbs} for multiple change point estimation and provide theoretical guarantee in terms of estimation consistency, which is a new contribution to the literature. 

%To our best knowledge,  this is the first time the consistency of WBS is obtained in a nonparametric setting and for a very broad class of time series. 

%as the existing consistency has only been obtained in very simple Gaussian or parametric time series models. 

As one of the state-of-the-art segmentation algorithm with promising practical performance, WBS has been combined with various single change point detection methods for Euclidean or metric space valued time series to enable multiple change point estimation (see e.g.,  \cite{chak2021}, \cite{wzvs2022}, \cite{zhangwangshao2022}, \cite{wangsam2017} and \cite{jiang2024two}) without theoretical justification. The WBS algorithm was originally developed to detect change points in the mean of univariate i.i.d data, and its performance on temporally dependent data could depend on the estimation of the long run variance (LRV), which was known to be a difficult problem due to the difficulty of choosing the bandwidth in the presence of change-points \citep{Vogelsang2007,shaozhang2010,robbins2022}. By using self-normalization, our method avoids direct estimation of LRV. At the same time, the adoption of a self-normalizer and the usage of projection in our test statistic significantly complicate our proof for estimation consistency, which requires substantially different theoretical arguments from those in \cite{wbs}. To our knowledge, except for \cite{korkas2017multiple} where the WBS estimation consistency was proved for a parametric time series model, we are the first to show the WBS estimation consistency for a broad class of time series (which includes Euclidean valued time series as a special case) and in a nonparametric setting. In our theory, the convergence rates for each change point estimator are adaptive to the heterogeneous change magnitudes (see Remark \ref{rmk_sec3} in Section \ref{sec_3}).   A more detailed discussion on WBS and WBS2 \citep{fryzlewicz2020detecting} is provided at the beginning of Section \ref{sec_3}.

The rest of this paper is organized as follows. In Section \ref{sec2single}, we first provide preliminary result on Functional CLT for Hilbert space valued time series in Section \ref{sec2_1}, and introduce our SN based method for single change point testing and estimation in Section \ref{sec2_2}. In Section \ref{sec2_3}, we we provide some (counter) examples of metric spaces of strong negative type, and show that the metric spaces where some frequently encountered object-valued data live in are of strong negative type, which may be of independent interest. The WBS-SN algorithm for multiple change point estimation, as well as the estimation consistency result, are provided in Section \ref{sec_3}.  Section \ref{sec4simu}  demonstrates the promising performance of our proposed testing and estimation methods in finite sample via simulations. In Section \ref{sec5_real}, two real data applications are provided, including stock return data in Section \ref{sec5_1} and TLC Trip Record Data in Section \ref{sec5_2}, respectively. Section \ref{sec6conclu} concludes and all the proofs are gathered in the supplement.

The following notations will be used throughout the paper. 	Let $\|\cdot\|$ be the Euclidean norm on $\R^p$ and $\|\cdot\|_\cH= \sqrt{\li \cdot,\cdot\ri_{\cH}}$ be the norm on the Hilbert space $\cH$ (all the Hilbert spaces considered in this paper are separable (real) Hilbert spaces). For $\R$- or $\cH$-valued random variables $\{X_t\}_{t\in \Z}$, let $\F_a^b$ denote the $\sigma$-field generated by $\{X_a,\dots,X_b\}$ for any $-\infty \leq a\leq b\leq\infty$. We use ``$\rightsquigarrow$" to denote the weak convergence in $D_\cH[0,1]$: the space of $\cH$-valued cadlag functions defined on $[0,1]$ endowed with the Skorokhod metric.

	\section{Single Change Point Estimation}\label{sec2single}
	\subsection{Preliminary Result}\label{sec2_1}

	Given $\R$- or $\cH$-valued random variables $\{X_t\}_{t\in \Z}$ and positive integer $n$, define its $\rho$-mixing coefficient \citep{bradley2005} as $\rho(n)=\sup_{j\in\Z}\rho(\F_{-\infty}^j,\F_{j+n}^\infty)$, where $$\rho(\mathcal{A},\mathcal{B})=\sup \left\{ \frac{|\cov(X,Y)|}{ \sqrt{\var(X) \var(Y)} }: X\in L_2(\mathcal{A}),Y\in L_2(\mathcal{B})  \right\}$$ and $L_2(\mathcal{A})$ is the set of $\R$-valued $\mathcal{A}$-measurable random variables with finite second moment for any $\sigma$-algebra $\mathcal{A}$. $\{X_t\}_{t\in \Z}$ is said to be $\rho$-mixing if $\lim_{n\to\infty}\rho(n)=0$. As in \cite{BUCC2017}, define the coefficient $\alpha_{1,1}(n)=\sup_{j\in\Z,k\geq n}\alpha(\F_j^j,\F_{j+k}^{j+k})$, where $\alpha(\mathcal{A},\mathcal{B})=\sup\big\{ |P(A\cap B)-P(A)P(B) |: A\in \mathcal{A},B\in \mathcal{B}    \big\}$ for any $\sigma$-fields $\mathcal{A}$ and $\mathcal{B}$.
	An operator $Q:\cH\to\cH$ is said to be a positive, self-adjoint, trace class operator if (i) $\li Qh,h\ri_\cH\geq0$ for any $h\in\cH$;
	(ii) $\li Qh,x\ri_\cH=\li h, Qx\ri_\cH$ for any $h,x\in\cH$; (iii) $\sum_{k=1}^{\infty}\li Qe_k,e_k \ri_\cH <\infty$ for any orthonormal base $\{e_k\}$ of $\cH$.
	
	Based on these concepts, we can define $\cH$-valued Gaussian random element as well as $\cH$-valued Brownian motion.
	\begin{definition}[\cite{andresen2013}]
		An $\cH$-valued random element $X$ is Gaussian with covariance operator $Q$ and mean $\mu\in H$~(denoted as $X\sim \N(\mu,Q)$) if for any $h\in \cH$, $\li X,h\ri_\cH \sim \N(\li \mu,h\ri_\cH, \li Qh,h\ri_\cH)$. 
	\end{definition}
	
	\begin{definition}[\cite{andresen2013}] \label{def2}
		Let $Q$ be a positive, self-adjoint, trace class operator on $\cH$. An $\cH$-valued Brownian motion with covariance operator $Q$ is defined as $B_Q(t),t\in[0,1]$ such that (i) $B_Q(0)=0$; (ii)  $B_Q(t)$ has continuous sample path; (iii) $B_Q(t)$ has independent increment; (iv) $B_Q(t)-B_Q(s)\sim \N(0,(t-s)Q)$ for any $0 \leq s<t\leq 1$.
		
	\end{definition}

	The following FCLT for stationary $\cH$-valued random variables is from \cite{BUCC2017}; see \cite{shar2016} for similar results under different mixing conditions.
	
	\begin{theorem}\label{th1}
		Let $\{X_t\}_{t\in \Z}$ be a sequence of  strictly stationary $\cH$-valued random variables with mean $\mu$. Assume $\{X_t\}_{t\in \Z}$ is $\rho$-mixing and the following condition holds for some $\delta>0$: (1) $\E\|X_1\|_\cH^{2+\delta}<\infty$; (2) $\sum_{n=1}^{\infty}[\alpha_{1,1}(n)]^{\delta/(2+\delta)}<\infty$. Then $\frac{1}{\sqrt{n}}\sum_{t=1}^{\fr}\big(X_t-\mu\big)\rightsquigarrow B_Q(t),$
		where $\{B_Q(t)\}_{t\in[0,1]}$ is an $\cH$-valued Brownian motion with 
		\begin{align}\label{eq11}
			\li Qx, y\ri_\cH = \sum_{t\in\Z}\E\big(\li X_0-\mu,x \ri_\cH\li X_t-\mu,y \ri_\cH\big) \mbox{ for any } x,y\in\cH.
		\end{align}
		Furthermore, the series in Equation (\ref{eq11}) converges absolutely.
	\end{theorem}

\subsection{ Single Change Point Testing and Estimation}\label{sec2_2}
	
From now on, we assume the time series $\{X_t\}_{t\in\Z}$ takes value in a separable metric space $(\X,d)$ of strong negative type (see Definition 2.2 in \cite{chak2021} or \cite{lyons2013}). Let $P_i$ be the marginal distribution of $X_i$. We aim to test 
$H_0:P_1=\cdots=P_n$
versus
$H_1:P_1=\cdots=P_{k^\ast}\not=P_{k^\ast+1}=\cdots=P_n$ for some $k^\ast=\lfloor n r_0\rfloor\in (nb,(1-b)n)$ with $b\in (0,1/2)$.
Here $b$ is a cutoff parameter and it is typically set to be $0.1$ or $0.15$; see \cite{andrews1993}. This is a common assumption as the detection of a change point is very challenging when it is close to the boundary; see \cite{horvath2020} for some recent work.  

By Proposition 3.1 in \cite{chak2021}, for a metric space $(\X,d)$ of strong negative type, there exists a Hilbert space $\cH$ and an embedding map $\phi:\X\rightarrow \cH$ such that $d(x,x')=\|\phi(x)-\phi(x')\|_{\cH}^2$. For $t\in \Z$, let $Y_t=\phi(X_t)$. If $\int\|Y_t\|_\cH dP_t<\infty$, define $\mu_t = \int Y_tdP_t$, which is the unique element in $\cH$ corresponding to the bounded linear functional $\cH\to\R: h\to \int\li h,Y_t \ri_\cH dP_t$. Then by Proposition 3.1 in \cite{lyons2013}, the null and alternative  hypothesis are equivalent to $H_0:\mu_1=\cdots=\mu_n$
versus
$H_1:\mu_1=\cdots=\mu_{k^\ast}\not=\mu_{k^\ast+1}=\cdots=\mu_n$.

\begin{remark}
    Some possible embedding maps for specific metric spaces are given in Section 3 of \cite{lyons2013}. For example, when $(\X,d)$ is the one dimensional Euclidean space, we can define $\phi(x) = \id_{[0,\infty)}(y)-\id_{[x,\infty)}(y)\in L^2(\R)$ where $ L^2(\R)$ is the Hilbert space of square integrable functions on $\R$. A similar embedding into $L^2(\R^p)$ can also be defined when $(\X,d)$ is the $p$-dimensional Euclidean space for $p\geq 2$. Additionally, define $k:\R^p\times\R^p\to \R$ as 
	\begin{align}
	k(x,y) = \frac{1}{2}\big[d(x,x_0)+d(y,x_0)-d(x,y)\big]
	\end{align}
	for some fixed $x_0\in\R^p$. Then we can also define the embedding of the $p$-dimensional Euclidean space as $\phi(x) = k(x,\cdot)$ where $k(x,\cdot)$ is an element in the reproducing kernel Hilbert space (RKHS) induced by the kernel function $k(\cdot,\cdot)$; see Section 4 of \cite{sej2013} for more details.
\end{remark}

To detect a change in $\mu_t$, we first split the data into three parts: $\X_1=\{X_1,\dots,X_{\fb}\}$, $\X_2 =\{X_{\fb+1},\dots,X_{n-\fb}\} $ and $\X_3 = \{X_{n-\fb+1},\dots,X_n\}$. Based on $\X_1$ and $\X_3$, we define 
\begin{align}\label{eq_proj}
	\bar{Y}_1-\bar{Y}_3=\frac{1}{\sqrt{n}} \sum_{t=1}^{\fb} Y_t-\frac{1}{\sqrt{n}} \sum_{t=n-\fb+1}^{n} Y_t,
\end{align}
which estimates the mean change $\mu_1{-}\mu_n$ (subject to a rescaling factor) under $H_1$. Then we project the transformed observations in $\X_2$ along the direction $\bar{Y}_1-\bar{Y}_3$, that is, we define univariate  random variables 
$Z_t=\li\bar{Y}_1-\bar{Y}_3,Y_t\ri_\cH$ for $t=\fb{+}1,\cdots,n{-}\fb$. Note that the temporal dependence in the projected sequence $\{Z_t\}$ is very complex due to projection. Nevertheless, we shall apply the SN test statistic in \cite{shaozhang2010}, which is developed for weakly dependent vector time series, to $\{Z_t\}$. To be specific, let $S_{a,b}{=}\sum_{t=a}^{b}Z_t$, our test statistic is defined as
\begin{align}
	G_n = \max_{k=\fb+1,\fb+2,\dots,n-\fb-1}\frac{T_n(k)}{\sqrt{V_n(k)}},\nonumber
\end{align}
where $T_n, V_n$ are defined as
\begin{align}
	T_n(k)=&\frac{1}{\sqrt{n-2\fb}}\sum_{t=\fb+1}^{k}\left(Z_{t}-\frac {S_{\fb+1,n-\fb}}{n-2\fb}\right), \label{sn_num} \\
	V_n(k) =& \frac{1}{(n-2\fb)^{2}}\Bigg\{\sum_{t=\fb+1}^{k}\left\{S_{\fb+1,t}-\frac{t-\fb}{k-\fb}S_{\fb+1,k}\right\}^2  \nonumber\\
	& \qquad\qquad\qquad+\sum_{t=k+1}^{n-\fb}\left\{S_{t,n-\fb}-\frac{n-\fb-t+1}{n-\fb-k}S_{k+1,n-\fb} \right\}^2 \Bigg\}. \label{sn_den}
\end{align}

\begin{remark}
	It is worth noting that the sequence $\{Y_t\}$ depends on the embedding map $\phi(\cdot)$ which may be unknown or difficult to derive. However, due to the use of sample splitting and self-normalization, our statistic $G_n$ does not depend on  $\phi(\cdot)$ in any explicit way and only depends on pairwise distances. To see this, note that 
	\begin{align}
		Z_t=&\frac{1}{\sqrt{n}}\sum_{j=1}^{\fb}\li Y_j,Y_t\ri_\cH-\frac{1}{\sqrt{n}}\sum_{j=n-\fb+1}^{n}\li Y_j,Y_t\ri_\cH\nonumber\\
		=&\frac{1}{2\sqrt{n}}\sum_{j=1}^{\fb}\{\|Y_j\|_\cH^2+\|Y_t\|_\cH^2-d(X_j,X_t)\} \nonumber  -\frac{1}{2\sqrt{n}}\sum_{j=n-\fb+1}^{n}\{\|Y_j\|_\cH^2+\|Y_t\|_\cH^2-d(X_j,X_t)\}\nonumber\\
		=&-\frac{1}{2\sqrt{n}}\sum_{j=1}^{\fb}d(X_j,X_t)+\frac{1}{2\sqrt{n}}\sum_{j=n-\fb+1}^{n}d(X_j,X_t)+R_n,\nonumber
	\end{align}
	where $R_n=\frac{1}{2\sqrt{n}}\sum_{j=1}^{\fb} \|Y_j\|_\cH^2 + \frac{1}{2\sqrt{n}}\sum_{j=n-\fb+1}^{n} \|Y_j\|_\cH^2$ is an unknown quantity that depends on $\phi(\cdot)$. The key is that $R_n$ does not depend on $t$, so it gets canceled out when forming the SN test statistic. To be specific, define $\hat Z_t = \sum_{j=n-\fb+1}^{n}d(X_j,X_t)-\sum_{j=1}^{\fb}d(X_j,X_t)$ and $\hat S_{a,b}=\sum_{t=a}^{b}\hat Z_t$, then the statistic $G_n$ can be calculated by replacing $Z_t,S_{s,b}$ in the definition of $T_n(k)$ and $V_n(k)$ with $\hat Z_t$ and $\hat S_{a,b}$. In other words, our test statistic $G_n$ only depends on the pairwise distances $\{d(X_t,X_{t'})\}_{t,t'=1}^n$. 
\end{remark}

\begin{remark}
	Note that our statistic $G_n$ only requires one trimming parameter $b$, whereas the SN-based statistics proposed in Section 4 of \cite{jiang2024two} require two tuning parameters $\eta_1$ and $\eta_2$, and target the changes in Fr\'{e}chet mean and variance only. In their paper,   $\eta_1$ plays the same role as our trimming parameter $b$ and $\eta_2$ ensures their subsample estimators of Fr\'{e}chet mean and variance are calculated with large enough subsample size. Since the statistics in \cite{jiang2024two} require repeated estimation  of Fr\'{e}chet mean and variance on different subsamples, they are much more expensive than $G_n$ in computation. 
\end{remark}

As shown below, the limiting null distribution of $G_n$ is $G$, which is pivotal and its upper critical values have been tabulated in \cite{gws2023}. Given the level $\gamma\in (0,1)$ (say, 0.05), our corresponding test is $\boldsymbol{1}(G_n>G_\gamma)$, where $G_\gamma$ is the $100(1-\gamma)$th upper percentile of $G$. The following assumptions are needed to derive the asymptotic distributions of $G_n$ under the null and alternative.

\begin{assumption} \label{assump1}
	Let $\{Y_t{-}\mu_t\}_{t\in \Z}$ be a strictly stationary $\cH$-valued time series. Assume $\{Y_t\}_{t\in \Z}$ is $\rho$-mixing and the following condition hold for some $\delta>0$:
	\begin{enumerate}
		\item\label{assump1_p1} $\E d(X_1,x_0)^{1+\delta/2}<\infty$ for some $x_0\in\X$;
		\item \label{assump1_p2} $\sum_{n=1}^{\infty}[\alpha_{1,1}(n)]^{\delta/(2+\delta)}<\infty$.
	\end{enumerate}
\end{assumption}

\begin{remark}
	By the definition of $\phi(\cdot)$, we can see it is a continuous bijection between $(\X,d)$ and $(\phi(\X),\|\cdot\|_\cH)$ with continuous inverse. For any positive integers $a\leq b$, the $\sigma$-algebra generated by $\{Y_a,\dots,Y_b\}$ is the same as the $\sigma$-algebra generated by $\{X_a,\dots,X_b\}$. So part \ref{assump1_p2} in Assumption \ref{assump1} holds for $\{Y_t\}$ if and only if it holds for $\{X_t\}$.
\end{remark}
By replacing $\phi(\cdot)$ with $\phi(\cdot)-\phi(x_0)$, we can assume, without loss of generality, that $\phi(x_0)=0\in\cH$. Then $\E||Y_1||_\cH^{2+\delta}=\E\ d(X_1,x_0)^{1+\delta/2}$. Note that
% \begin{align}
% 	& \E\li Y_t-\mu_t,Y_t-\mu_t\ri_\cH \\
% 	= & \E\|Y_t\|^2_\cH+\|\mu_t\|^2_\cH-2\E\li Y_t,\mu_t \ri_\cH \nonumber \\
% 	= & \E d(X_t,x_0)+\int\int\li \phi(x),\phi(y)\ri_\cH dP_t(x)dP_t(y)-2\E\int \li Y_t,\phi(x)\ri_\cH dP_t(x) \nonumber \\
% 	\leq & 4\E d(X_t,x_0), \nonumber
% \end{align}

\begin{align}
\E\li Y_t-\mu_t,Y_t-\mu_t\ri_\cH^{1+\delta/2}  = 	& \E\li Y_1-\mu_1,Y_1-\mu_1\ri_\cH^{1+\delta/2} \\
	\leq & 2^{1+\delta} \big(\E\|Y_1\|^{2+\delta}_\cH+\|\mu_1\|^{2+\delta}_\cH\big) \nonumber \\
	= & 2^{1+\delta} \E d(X_1,x_0)^{1+\delta/2}+2^{1+\delta}\Big[\int\int\li \phi(x),\phi(y)\ri_\cH dP_1(x)dP_1(y)\Big]^{1+\delta/2} \nonumber \\
	\leq & 2^{2+\delta}\E d(X_1,x_0)^{1+\delta/2}, \nonumber
\end{align}
so part \ref{assump1_p1} in Assumption \ref{assump1} implies that $\E\|Y_t-\mu_t\|^{2+\delta}_\cH<\infty$, which further implies the existence of $\mu_t$. By Theorem \ref{th1} we have $\sum_{t=1}^\fr (Y_t-\mu_t)/ \sqrt{n} \rightsquigarrow B_Q(t)$, where $\{B_Q(t)\}_{t\in[0,1]}$ is an $\cH$-valued Brownian motion with $	\li Qx, y\ri_\cH = \sum_{t\in\Z}\E\big(\li Y_0-\mu_0,x \ri_\cH\li Y_t-\mu_t,y \ri_\cH\big)$ for any $x,y\in\cH$.
The following theorem shows the asymptotic properties of $G_n$, which is proved in the appendix of the supplement.
\begin{theorem}\label{theorem_cp}
	Suppose Assumption \ref{assump1} holds, then: (i) under $H_0$, we have $G_{n} \stackrel{\D}{\to} G$, where
	\begin{equation}
		G\stackrel{d}{=}\sup_{r \in [0,1]}  \frac{B(r)- rB(1) }  {\big\{\int_{0}^{r} [B(s)-\frac{s}{r}B(r)]^2 ds + \int_{r}^{1}[ B(1)-B(s)-\frac{1-s}{1-r}(B(1)-B(r)) ]^2 ds \big\}^{1/2}},\nonumber
	\end{equation}
	$\{B(t)\}_{t\in[0,1]}$ is an $\R$-valued standard Brownian motion;
	(ii) under $H_A$, denote ${\Delta}_n=\mu_n{-}\mu_1$, we have 
	\begin{enumerate}%[label=(\roman*)]
		\item If $\sqrt{n}\|{\Delta}_n\|_\cH\to\infty$, then $G_{n} \stackrel{p}{\to} \infty$.
		\item If $\sqrt{n}{\Delta}_n\to c\in \cH$ and $\|{c}\|_\cH \neq {0}$, then we have 
		$$\bar{Y}_1-\bar{Y}_3 \stackrel{\D}{\to}B_Q(b)-\big[B_Q(1)-B_Q(1-b)\big]-bc = Y^\ast \mbox{, } G_{n} \stackrel{\D}{\to}  G^\ast, $$
		where the conditional distribution of $G^\ast$ given $Y^\ast=y$ is equal to the distribution of
		\begin{equation}
			%G^\ast\big|_{\phi^\ast=\phi} \stackrel{d}{=} 
			\sup_{r \in [0,1]}  \frac{B'(r)- rB'(1) }  {\big\{\int_{0}^{r} [B'(s)-\frac{s}{r}B'(r)]^2 ds + \int_{r}^{1}[ B'(1)-B'(s)-\frac{1-s}{1-r}(B'(1)-B'(r)) ]^2 ds \big\}^{1/2}},\nonumber
		\end{equation}
		$B'(r){=}B(r){+}  H_y((1{-}2b)r{+}b)/ \sqrt{1-2b}$, $H_y=\li y,(r{-}r_0){c}\mathbf{1}_{r\geq r_0}\ri_\cH / \sqrt{\li Q y,y \ri_\cH}$ 
		%and ${H}(r){=}(r{-}r_0){c}\mathbf{1}_{r\geq r_0}$ 
		with $\mathbf{1}_{r\geq r_0}=1$ if $r\geq r_0$ and $0$ otherwise. 
		\item If $\sqrt{n}\|{\Delta}_n\|_\cH\to 0$, then $G_n \stackrel{\D}{\to} G,$ so our test have trivial power asymptotically.
	\end{enumerate}
\end{theorem}

When a single change point is detected, it is of interest to estimate its location. Building upon the test statistic $G_n$, it is natural to estimate  $k^\ast$ by  
$$\hat k =\argmax_{k=\fb+1,\fb+2,\dots,n-\fb-1}\frac{T_n(k)}{\sqrt{V_n(k)}}.$$
The following theorem shows the consistency of $\hat k$, which is proved in the appendix of the supplement. 

\begin{theorem}\label{cp_est}
	Suppose Assumption \ref{assump1} holds and $\sqrt{n}\|{\Delta}_n\|_\cH \to \infty$. For a given integer sequence $\epsilon_n$ which satisfies that $\frac{\epsilon_n}{n}\to 0$ and $n(\epsilon_n\|{\Delta}_n\|_\cH)^{-2}\to 0$, we have 
	\begin{equation}
		P(|\hat k -k^\ast|<\epsilon_n)\to 1\label{eqth1}. 
	\end{equation}
\end{theorem}

The proof of Theorem \ref{cp_est} basically follows the same idea as in the proof of Theorem 1 in \cite{zjs2022}. Note that the convergence rate $\epsilon_n$ of the change point estimator $\hat k$ is the same as the rate derived in Theorem 1 of \cite{zjs2022}, which first worked out the consistency and the rate for the SN-based change point location estimator for Euclidean time series. The latter paper  also stated that under the fixed alternative setting (i.e., $\|{\Delta}_n\|_\cH=C>0$), the convergence rate $\epsilon_n/n$ is at best $1/\sqrt{n}$. In comparison, as shown in \cite{bai1994least}, the optimal convergence rate for detecting change point in the mean of a multivariate time series is $1/n$. It seems difficult to improve convergence rate $\epsilon_n$ due to the complex form of $V_n(k)$ in the definition of $G_n$. We refer to the discussion following Theorem \ref{cp_est_wbs} in \cite{zjs2022} for more details of a local refinement procedure  to achieve the optimal rate.

		\subsection{Strong negativity of some metric spaces}\label{sec2_3}
	
 According to Theorem 3.16 in \cite{lyons2013}, every separable Hilbert space, such as the Euclidean space or the space of square integrable functions, is a metric space of strong negative type. Besides that, it is in general not easy to verify that a metric space is of strong negative type. According to part (iv) in \cite{lyonserrata2}, if $(\X,d)$ is a separable semimetric space of negative type, then $(\X,d^r)$ is a metric space of strong negative type for any $r\in(0,1)$. So in order to show $(\X,d)$ is a metric space of strong negative type, it suffices to show, by verifying the definition, $(\X,d^r)$ is of negative type for some $r>1$. For example, by Theorem 3.6 in \cite{meckes2013positive}, it is immediately clear that the space of measurable functions supported on $[0,1]$ with metric $d(f,g) = \{\int_{[0,1]} |f(x)-g(x)|^p dx\}^{1/p}$ is a metric space of strong negative type for any $p\in(1,2)$. 

    As for metric spaces that are not of strong negative type, it is proved in \cite{dor1976potentials} that $\R^q$ with the $\ell^p$ metric is not a metric space of negative type for any $3\leq  q\leq \infty$ and $2<p\leq \infty$. In addition, as stated in \cite{lyons2013}, $\R^2$ with the $\ell^1$ metric is a metric space of negative type, but not of strong negative type.
    
    We now show the following metric spaces where object-valued data resides in are of strong negative type. 
	\begin{enumerate}
		\item The space $\Omega_1$ of probability measures on $\R$, equipped with the 2-Wasserstein metric $d_W(x,y)= \{\int_0^1 (F^{-1}_x(u)-F^{-1}_y(u))^2du\}^{1/2}$, where $F^{-1}_x,F^{-1}_y$ are the inverse cdf functions of $x$ and $y$.
		
		\item The space $\Omega_2$ of Graph Laplacians of weighted graphs (see Section 3 in \cite{ginestet2017}) equipped with Frobenius metric $d_{F}(A,B) =\{ \tr [(A-B)^\top(A-B)]\}^{1/2}$ for $d\times d$ symmetric matrices $A$ and $B$.
		
		\item The space $\Omega_3$ of positive definite covariance matrices, equipped with the log-Euclidean metric $d_E(A,B) = d_{F}(\log_m(A),\log_m(B))$ for $d\times d$ symmetric matrices $A$ and $B$, where $\log_m$ is the matrix-log function.
	\end{enumerate}
	
	According to part (iv) in \cite{lyonserrata2}, it suffice to show $(\Omega_1,d_W^2)$, $(\Omega_2,d_F^2)$ and $(\Omega_3,d_E^2)$ are separable semimetric spaces of negative type. The separability of $(\Omega_1,d_W^2)$ is implied by the separability of $(\Omega_1,d_W)$, which is shown in \cite{bolley2008}. For any $n\geq 2$, $x_1,x_2,\dots,x_n\in\Omega_1$, $a_1,a_2,\dots,a_n\in\R$ with $\sum_{i=1}^{n}a_i=0$, 
	\begin{align}
		\sum_{i=1}^{n}\sum_{j=1}^{n}a_ia_jd^2_W(x_i,x_j) = &\int_0^1 	\sum_{i=1}^{n}\sum_{j=1}^{n}a_ia_j(F^{-1}_{x_i}(u)-F^{-1}_{x_j}(u))^2du \nonumber \\
		=& -2\int_0^1 	\sum_{i=1}^{n}\sum_{j=1}^{n}a_ia_jF^{-1}_{x_i}(u)F^{-1}_{x_j}(u)du \label{eq_snt} \\
		=& -2\int_0^1 \big( \sum_{i=1}^{n}  a_iF^{-1}_{x_i}(u)    \big)^2du \leq 0 \nonumber,
	\end{align}
	so $(\Omega_1,d_W^2)$ is of negative type and the strong negativity of $(\Omega_1,d_W)$ is shown. For $(\Omega_2,d_F^2)$, the separability follows from the fact that $\R^{d(d+1)/2}$ with the Euclidean metric is separable and every subspace of a separable metric space is separable (see Proposition 26 in Section 9.6 of \cite{royden2010real}). Denote $A_{rs}$ as the $(r,s)$ element of a matrix $A\in\Omega_2$, for any $n\geq 2$, $A^1,A^2,\dots,A^n\in\Omega_2$, $a_1,a_2,\dots,a_n\in\R$ with $\sum_{i=1}^{n}a_i=0$, we have 
	\begin{align}
		\sum_{i=1}^{n}\sum_{j=1}^{n}a_ia_jd^2_F(A^i,A^j) = \sum_{r=1}^{n}\sum_{s=1}^{n}	\sum_{i=1}^{n}\sum_{j=1}^{n}a_ia_j(A^i_{rs}-A^j_{rs})^2 = -2\sum_{r=1}^{n}\sum_{s=1}^{n} \big( \sum_{i=1}^{n}  a_iA^i_{rs} \big)^2\leq 0 \nonumber,
	\end{align}
	so $(\Omega_2,d_F^2)$ is of negative type, which implies $(\Omega_2,d_F)$ is of strong negative type. The strong negativity of $(\Omega_3,d_E)$ can be shown in a similar way as $(\Omega_2,d_F)$ and we omit the details.

\section{Multiple Change Point Estimation: the WBS-SN Algorithm}\label{sec_3}

The original WBS algorithm was proposed in \cite{wbs} to detect changes in the mean $f_t$ for $\R$-valued random variables $X_t=f_t+\epsilon_t$ with noise $\epsilon_t\stackrel{i.i.d}{\sim}\N(0,\sigma^2)$. It operates in an iterative manner as follows: At the beginning, $M$ independent random intervals $\{(s_m,e_m)\}_{m=1}^M$ are drawn from $\{1,2,\dots,n\}$ and $C_m = \max_{k\in\{{s_m},\dots,{e_m-1}\}}|\tilde T_m(k)|$ are calculated where the CUSUM statistic $\tilde T_m(k)$ (calculated on $\{X_{s_m},\dots,X_{e_m}\}$) is a rescaled version of $T_n(k)$ defined in Equation (\ref{sn_num}). If $\max\{C_m:m=1,\dots,M\}>K_n$ for some threshold $K_n$, then the first estimated change point is chosen as $k_1$ where $(m_1,k_1) = \argmax_{m\in\{1,2,\dots,M\},k\in\{{s_m},\dots,{e_m-1}\}}|\tilde T_m(k)|$. The same procedure is then repeated for $C_m$ calculated on the intervals inside subsamples $\{X_1,\dots,X_{k_1}\}$ and $\{X_{k_1+1},\dots,X_{n}\}$, which gives the second and third change point estimator. The splitting and testing continue until no change points are found in each subsample.  
%on the four subsamples produced by the three change point estimators $k_1,k_2$ and $k_3$. 
That is, the algorithm stops if, after several iterations, the maximum of $C_m$ calculated in all resulting subsamples are less than $K_n$.

As stated in Section 3.4 of \cite{wbs}, the performance of WBS is sensitive to the choice of the threshold $K_n$, which  depends on the estimation of the noise level $\sigma^2$. Alternatively, \cite{wbs} proposed another implementation of WBS by setting $K_n=0$ in the above algorithm and select the change point estimators from the output of WBS using ``strengthened Schwarz Information Criterion'' (sSIC), which also depends on an estimator of $\sigma^2$. The extension of WBS to mean change point estimation in time series data inevitably involves the LRV estimation, which is a difficult task due to the bandwidth selection and the presence of change points. The popularly adopted kernel estimator of LRV tends to incur downward bias \citep{chan2017high} and requires careful choice of bandwidth parameters. So far there seems very little theoretical work on the consistency of WBS in the time series setting. 
%The difficulty in quantifying noise level, as well as the temporal dependence in time series data leads to more sophisticated statistic than $\tilde T_m(k)$ to be combined with WBS to achieve good performance, which makes the theoretical justification of estimation consistency more challenging. 
To the best of our knowledge, the only consistency result for the use of WBS was presented in \cite{korkas2017multiple}, which is developed for detecting variance change in a sequence of random scaled $\chi^2_1$-distributed time series data with piecewise constant variance, where the CUSUM statistic $\tilde T_m(k)$ is scaled by $q_m = (e_m-s_m+1)^{-1}\sum_{t=s_m}^{e_m}X_t^2$. 

Recently, \cite{fryzlewicz2020detecting} proposed the Wild Binary Segmentation 2 (WBS2) algorithm for i.i.d data, where instead of generating $M$ random intervals at the beginning, a smaller number of random intervals are drawn in each iteration and the algorithm stops after $(n{-}1)$ change point candidates are selected (see Theorem 3.1 in \cite{fryzlewicz2020detecting}). Based on WBS2, \cite{cho2023multiple} proposed the WCM.gSa algorithm for detecting multiple change points in the mean of time series data, where a new application of the Schwarz criterion \citep{schwarz1978estimating} is used to select the change point estimators from selected sub-collections in the output of WBS2. Although WBS2 does not impose any model assumption on the data, the estimation consistency result of WCM.gSa requires the demeaned time series coming from an AR(p) model.

In this section, we present an algorithm (WBS-SN) that estimates the locations of unknown but fixed number of change points in $\{\mu_{t}\}$ based on WBS. For positive integer $m_0$ and $r^\ast_0=0<r^\ast_1<r^\ast_2<\cdots<r^\ast_{m_0}<1=r^\ast_{m_0+1}$, we assume the change points locate at $k^\ast_i=\lfloor n r^\ast_i \rfloor$ and have magnitude $\Delta^\ast_{i} = \mu_{k^\ast_i+1}- \mu_{k^\ast_i}$ for $i=1,2,\dots,m_0$. 
%Note that we do not need the assumption that $\min_{1\leq i\leq m_0{+}1} (r^\ast_{i}{-}r^\ast_{i-1} )>c_0$ for some pre-specified window length parameter $c_0$, which is a common assumption in the literature of change point testing and estimation for time series data. 
Then the data generating process under the multiple change alternative corresponds to 
$$H_1:P_1=\cdots=P_{k_1^\ast}\not=P_{k_1^\ast+1}=\cdots=P_{k_2^\ast}\not =P_{k_2^\ast+1}\cdots P_{k_{m_0}^\ast}\not= P_{k_{m_0}^\ast+1}=\cdots=P_n.$$

We generate $M$ small intervals such that the length of each interval is proportional to $n$ in the following way. Let $\{U_{1m}\}_{m=1}^M, \{U_{2m}\}_{m=1}^M$ be two independent sequences of i.i.d $Uniform[0,1]$ random variables, $\underline U_m = \min\{U_{1m},U_{2m}\}$ and $ \bar U_m = \max\{U_{1m},U_{2m}\}$, then each interval is represented as $(\underline U_m,\bar U_m)$.
%, which is uniformly distributed on $\{(x,y)\in\R^2:0<x<y<1\}$. 
Consider the event that each change point is contained in the middle of at least one interval, which does not contain any other change points. To be specific, define
\begin{align}
	\M =& \Big\{\text{For each }i=1,2,\dots,m_0, \text{ there exist } m\in\{1,2,\dots,M\}, \text{ such that }\nonumber \\
	&\qquad r^\ast_{i-1}<\underline U_m<r^\ast_i<\bar U_m<r^\ast_{i+1} \text{ and }\frac{r^\ast_i{-}\underline U_m}{\bar U_m{-}\underline U_m}\in(b,1-b)\Big\}.
\end{align}
Denote $\underline r = \min_{1\leq i\leq m_0{+}1} (r^\ast_{i}{-}r^\ast_{i-1} )\leq \frac{1}{m_0+1}$, then we have 
\begin{align}
	1-P(\M)\leq& \sum_{i=1}^{m_0}\big\{1{-}P\big(\{  r^\ast_{i-1}<\underline U_1<r^\ast_i<\bar U_1<r^\ast_{i+1} \text{ and }\frac{r^\ast_i{-}\underline U_1}{\bar U_1{-}\underline U_1}\in(b,1-b)  \}\big)    \big\}^M \nonumber \\
	\leq& m_0\big\{1{-}\frac{2\underline r (1{-}2b)}{1-b}   \big\}^M, \nonumber
\end{align}
which implies $P(\M)\to 1$ as $M\to \infty$. Denote $b_m = \lfloor (\lfloor n\bar U_m\rfloor-\lfloor n\underline U_m\rfloor{+}1)b  \rfloor$, and let $T_{nm}(k)$ and $V_{nm}(k)$ be the numerator and squared denominator of the SS-SN statistic, as defined in Equations (\ref{sn_num}) and (\ref{sn_den}), which are calculated using data $\{X_{\lfloor n\underline U_m\rfloor}, X_{\lfloor n\underline U_m\rfloor+1},\dots,X_{\lfloor n\bar U_m\rfloor}\}$ for some $k\in\{\lfloor n\underline U_m\rfloor{+}b_m,\dots,\lfloor n\bar U_m\rfloor{-}b_m{-}1\}$. Using pseudocode, the WBS-SN algorithm is defined in Algorithm \ref{algo1}.

\SetKwProg{Fn}{Function}{:}{}
\SetKwFunction{wbssn}{WBS-SN}
\begin{algorithm}%[H]	
	\Fn{\wbssn{$s,e,K_n, L_0, M,b$}}{
		
		\eIf{$e-s<L_0$}{
			STOP
		}{
			$\mathcal{M}_{s:e} :=$ set of those indices $m\in\{1,2,\dots,M\}$ for which $[\lfloor n\underline U_m\rfloor,\lfloor n\bar U_m\rfloor]  \subset [s,e]$ and $\lfloor n\bar U_m\rfloor-\lfloor n\underline U_m\rfloor\geq L_0$. \\
			$(m',k'):= \argmax\limits_{m\in \mathcal{M}_{s:e}, k\in \{\lfloor n\underline U_m\rfloor{+}b_m,\dots,\lfloor n\bar U_m\rfloor{-}b_m{-}1\}}\frac{T_{nm}(k)}{\sqrt{V_{nm}(k)}}$
			
			\eIf{$\frac{T_{nm'}(k')}{\sqrt{V_{nm'}(k')}} >K_n$}
			{add $k'$ to the set of estimated change points \\
				WBS-SN($s,k',K_n, L_0, M,b$) \\
				WBS-SN($k'+1,e,K_n, L_0, M,b$) }{STOP}

		}
		
	}
	\caption{WBS-SN for object-valued times series change point estimation}
	\label{algo1}
\end{algorithm}		

%To initialize the estimation procedure, 
To start, we apply Algorithm \ref{algo1} with WBS-SN($1,n,K_n, L_0, M,b$). Here $L_0$ denotes the minimal interval length, $b$ is the proportion of samples used to calculate the direction of projection $\bar Y_1{-}\bar Y_3$ on each random small interval and $K_n$ is the threshold. In practice, following the WBS procedure in \cite{wzvs2022}, $K_n$ is determined by simulations as follows: we generate $R$ samples of i.i.d. standard normal random variables with the same sample length $n$ as the original object-valued data. For the $r$th sample, we calculate
\begin{equation}
	\hat K_n^r =  \max\limits_{m\in \mathcal{M}_{1:n}, k\in \{\lfloor n\underline U_m\rfloor{+}b_m,\dots,\lfloor n\bar U_m\rfloor{-}b_m{-}1\}}\frac{T_{nm}^{(r)}(k)}{\sqrt{V_{nm}^{(r)}(k)}}\mbox{, }r=1,2,\dots,R,
\end{equation}
where $T_{nm}^{(r)}(k)$ and $V_{nm}^{(r)}(k)$ are calculated based on the $r$th sample. Finally, we choose $K_n$ as the $100\gamma$\% sample quantile of $\{\hat K_n^1,\hat K_n^2,\dots,\hat K_n^R\}$ for some $\gamma\in [0.8,0.95]$.

\begin{assumption}\label{assump_wbs}
	For $i=1,2,\dots,m_0$, we assume that 
	\begin{enumerate}%[(a)]
		\item  $\Delta^\ast_{i}  = n^{-\delta_i}f_i$ for some $\delta_i\in [0,\frac{1}{2})$ and non-zero element $f_i\in \cH$.
		\item $K_n\to \infty$ and $K_n = o(n^{1/2-\bar  \delta })$ where $\bar \delta = \max\{\delta_1,\delta_2,\dots,\delta_{m_0}\}$. 
	\end{enumerate}
\end{assumption}

The following theorem shows the estimation consistency for WBS-SN, which is proved in the appendix of the supplement.

\begin{theorem}\label{cp_est_wbs}
	Suppose Assumptions \ref{assump1} and \ref{assump_wbs} holds. Let $\hat m$ denote the number, and $\hat k_1\leq \hat k_2\leq\cdots\leq \hat k_{\hat m}$ the locations of the change point estimator obtained by the WBS-SN algorithm. Then we have
	\begin{multline}\label{eq_wbs1}
		\lim_{n\to \infty}P\big(\hat m =m_0,\, |\hat k_i -k^\ast_i|<\epsilon_{in} \text{ for any }i=1,2,\dots,m_0\big) \geq 1{- } m_0 \left\{1{-}\frac{2\underline r (1{-}2b)}{1-b}   \right\}^M,
	\end{multline}
	where $\{\epsilon_{in}\}$ is an integer sequence which satisfies that $\frac{\epsilon_{in}}{n}\to 0$ and $n(\epsilon_{in}\|{\Delta}^\ast_i\|_\cH)^{-2}\to 0$ for $i=1,2,\dots,m_0$.
	
\end{theorem}

\begin{remark}\label{rmk_sec3}
	As shown in Equation (\ref{eq_wbs1}), the asymptotic probability of correctly estimating the number and locations of change points is lower bounded by the same value as for $P(\M)$, which converges to one as $M\to \infty$. Also, the convergence rate $\epsilon_{in}$ for the true change point $k^\ast_i$ depends on the change magnitude $\|\Delta^\ast_{i}\|_\cH$ at $k^\ast_i$, so change points of larger magnitude correspond to faster convergence rate. By contrast,  a single rate is derived for $\max_{1\leq i\leq m_0}|\hat k_i -k^\ast_i|$ in Theorem 3 of \cite{zjs2022}, which develops a SN-based segmentation algorithm for Euclidean valued time series.  
\end{remark}

\section{Simulation Results} \label{sec4simu}
	
	In this section, we examine the size and power properties of our proposed test statistic, as well as the estimation accuracy of WBS-SN algorithm, in finite sample. Specifically in Section \ref{sec_simu1}, we consider the single change point setting for bivariate distributional data. In Section \ref{sec_simu2}, we examine empirical size and power
	of our test statistic for three data generation processes (DGPs, hereafter) studied in \cite{jiang2024two}, as well as another DGP where the Fr\'{e}chet mean and variance remain the same under the alternative. The favorable estimation performance of WBS-SN is shown in Section \ref{sec_simu3}. We set $b=0.15$ for all the simulation results and let $M=50$, $L_0=20$ for the WBS-SN algorithm used in Section \ref{sec_simu3}.
    
    % In Appendix \ref{sec2_3}, we show that the metric spaces where some frequently encountered object-valued data live in are of strong negative type, which may be of independent interest.
	
	\subsection{Bivariate Distributional Data}\label{sec_simu1}
	
	Consider the same DGP as in Section 3.2 of \cite{dubey2023change} and let $\{X_t\}$ take values in the space of bivariate distributions with metric 
    $$
    d(x,y) {=} \{\int_\R\int_\R |F_x(u,v){-}F_y(u,v)|^2dudv\}^{1/2},
    $$ 
    where $F_x$ and $F_y$ are the cdf functions of $x$ and $y$. The strong negativity of this metric space can be shown in a similar way as $(\Omega_1,d_W)$ in Section \ref{sec2_3}. We test for a change point in the marginal distribution of $X_t$. As comparison, we consider the statistics $SN_1$ and $SN_2$ defined in \cite{jiang2024two}, which have pivotal limiting distributions under the null, and the statistics defined in \cite{dubey2023change} (DZ) and \cite{dubey2020frechet} (DM) with rejection criteria calculated using 200 bootstrap/permutation replicates. Under the alternative, the size adjusted power for DM and DZ are calculated according to \cite{dominguez2000size}. The experiments are repeated 1000 times with nominal level at 5\%. 
	
	For the empirical size, we let $X_t=\N (Z_t,0.25I_2)$, where $I_2$ is the two dimensional identity matrix and $Z_t = \rho I_2 Z_{t-1}+0.5\sqrt{1-\rho^2}\epsilon_{t}$ with $\epsilon_{t}\sim^{i.i.d}\N(0,I_2)$. We set $n\in\{200,400\}$ and $ \rho\in\{-0.4,0.4,0,0.7\}$. As shown in Table \ref{tab_cp_bd}, our proposed statistic SS-SN has accurate size for $\rho\in\{-0.4,0,0.4\}$ and is over-sized for $\rho=0.7$ when $n=200$. The empirical sizes are close to the nominal level for all values of $\rho$ when $n=400$. For DM and DZ, the size distortion is large when $\rho\in\{0.4,0.7\}$ and the size accuracy does not improve much as $n$ increases. The large size distortion is expected since both DM and DZ are developed for independent object data. $SN_2$ is under-sized for all values of $\rho$ and the size distortion becomes larger for $\rho\in\{0.4,0.7\}$ as $n$ increases. For $SN_1$, the size accuracy is comparable to SS-SN. However, as will be shown later, $SN_1$ can only detect change in Fr\'{e}chet variance.
	\begin{table}
		%	\vspace{0.5em}
		\centering
        \setlength{\tabcolsep}{13pt}
		%\scalebox{0.95}{
			\begin{tabular}{cc|ccccc}
				\toprule
				\midrule
				n&$\rho$	 &  SS-SN &   $SN_1$ &   $SN_2$&DM &DZ     \\\hline
				\multirow{4}{*}{200}&-0.4 &    4.1 &5.2   &2.9  &6.7  &6.0  \\ \cline{3-7}
				&0    &    4.2 &5.1   &0.9  &3.2  &6.6  \\ \cline{3-7}
				&0.4  &    6.2 &6.7   &1.1  &25.0 &43.4    \\ \cline{3-7}
				&0.7  &    9.2 &5.2   &4.1  &85.1 &97.1   \\\cline{2-7}
				\multirow{4}{*}{400}&-0.4 &    5.1 &4.9   &3.9  &10.7  &4.3  \\ \cline{3-7}
				&0    &    4.2 &5.5   &2.3  &4.3   &5.5  \\ \cline{3-7}
				&0.4  &    4.3 &4.1   &0.5  &21.2  &42.5    \\ \cline{3-7}
				&0.7  &    5.3 &5.2   &1.4  &83.3  &95.5   \\\bottomrule
		\end{tabular}
  %}
		\caption{Empirical rejection rate (in percentage) under the null when testing for change point in marginal distribution of bivariate distributional data.}\label{tab_cp_bd}
	\end{table}
	
	For the size adjusted power we set $n{=}200,\rho\in\{-0.4,0,0.4\}$ and assume $X_t=\N (Z_t,0.25I_2)$ with $Z_t$ coming from the following two DGPs.
	$$\text{DGP1}: Z_t-\mu_t = \rho I_2(Z_{t-1}{-}u_{t-1})+0.5\sqrt{1-\rho^2}\epsilon_t\text{ with }
	\begin{cases}
		u_t = (0,0)^\top,  & 1\leq t\leq \lfloor \frac{n}{2}\rfloor \\
		u_t = (\delta,0)^\top, &\lfloor \frac{n}{2}\rfloor< t\leq n
	\end{cases}$$
	$$\text{DGP2}: Z_tM_t=\rho I_2Z_{t-1}M_{t-1}+0.4\sqrt{1-\rho^2}\epsilon_t\text{ with }
	\begin{cases}
		M_t = I_2,  & 1\leq t\leq \lfloor \frac{n}{2}\rfloor \\
		M_t = diag(\frac{0.4}{0.4+\delta},1), &\lfloor \frac{n}{2}\rfloor< t\leq n
	\end{cases}$$
	\normalsize
	Note that for DGP1 only the Fr\'{e}chet mean of $\{X_t\}$ changes. The size adjusted power against $\delta\in[0,0.8]$ are plotted in Figure \ref{fig_cp_bd} (DM and DZ are not included for $\rho{=}0.4$ because their size adjusted rejection criteria are both zero, due to the large size distortion under the null). For DGP1, SS-SN achieves the best power as compared to other methods when $\rho=-0.4$ and it has similar performance with the best method in comparison when $\rho\in\{0,0.4\}$. Except for $SN_1$, which has trivial power under DGP1, $SN_2$ has the largest power loss when $\rho\in\{-0.4,0\}$. For DGP2, SS-SN has mild power loss when $\rho=-0.4$ and the power loss becomes larger as $\rho$ increases, while it slightly outperforms $SN_2$ when $\rho=0.4$.
	
\begin{figure}
\begin{tikzpicture} 
\matrix (m) [row sep = 0em, column sep = 1.5em]{ 
    \node (p11) {\includegraphics[scale=0.42]{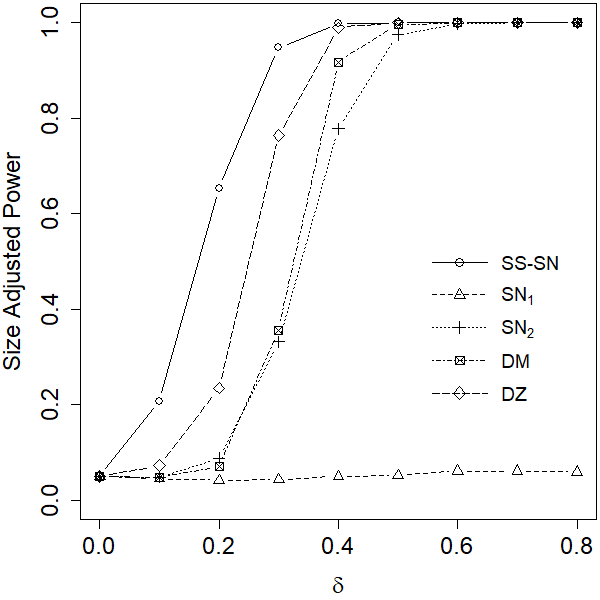}};
    %\node[above = 0cm of p11] (t11) {Manhattan};
    &
    \node (p12) {\includegraphics[scale=0.42]{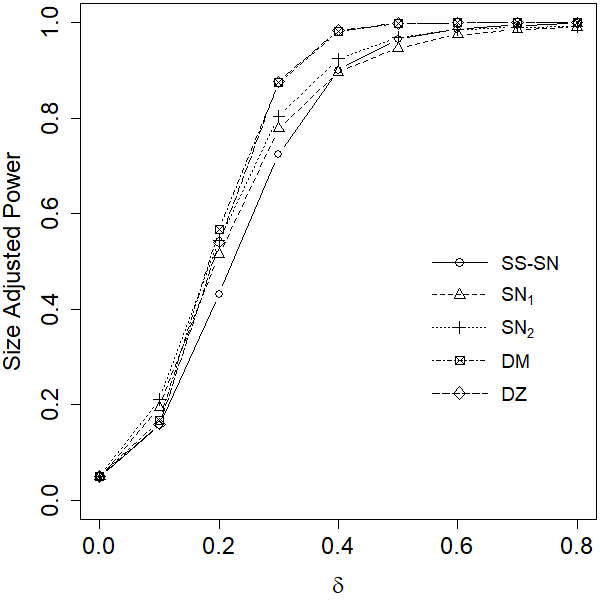}}; 
    %\node[above = 0cm of p12] (12) {Queens};
    \\
     \node (p21) {\includegraphics[scale=0.42]{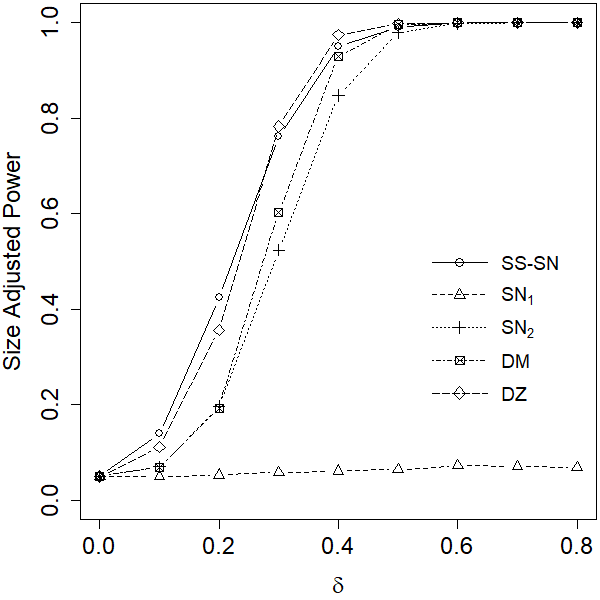}}; 
     %\node[above = 0cm of p21] (t21) {EWR};
     &
     \node (p22) {\includegraphics[scale=0.42]{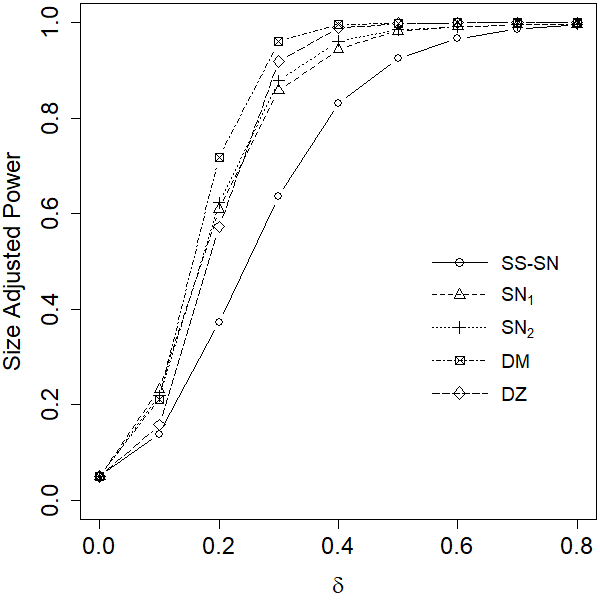}}; 
     %\node[above = 0cm of p22] (t22) {Brooklyn};
     \\
     \node (p31) {\includegraphics[scale=0.42]{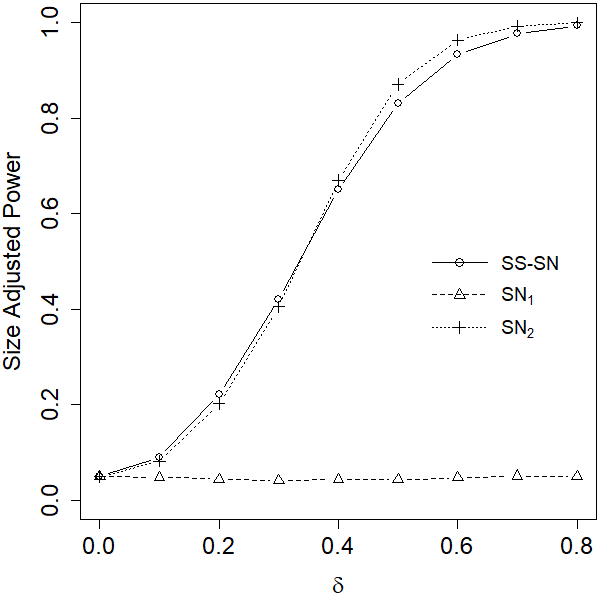}}; 
     %\node[above = 0cm of p31] (t31) {Bronx};
     &
     \node (p32) {\includegraphics[scale=0.42]{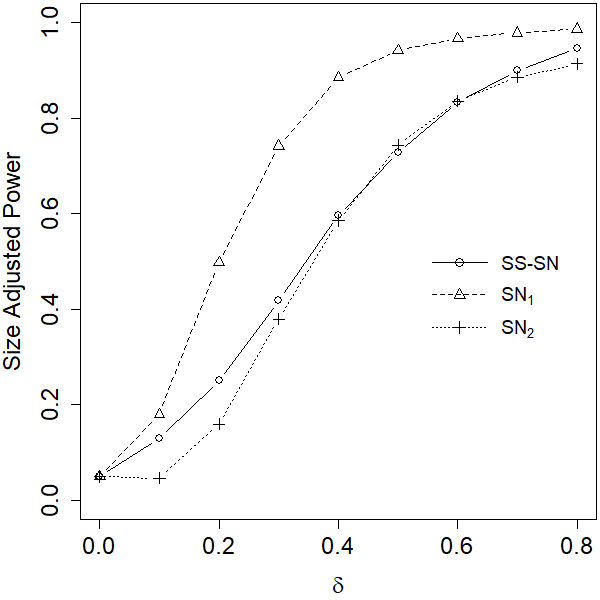}}; 
     %\node[above = 0cm of p32] (t32) {Staten Island};
     \\
};
\end{tikzpicture}
\centering
    \caption{Size adjusted power when testing for change point in marginal distribution of bivariate distributional data from DGP1 (left) and DGP2 (right) with $\rho=-0.4$ (first row), $\rho=0$ (second row) and $\rho=0.4$ (last row).}
    \label{fig_cp_bd}
\end{figure}

	\subsection{Univariate Distribution, Graph Laplacian and Covariance Matrix} \label{sec_simu2}
	
	In this section, we consider three DGPs of object valued time series studied in \cite{jiang2024two}, as well as another DGP where the Fr\'{e}chet mean and variance remain the same under the alternative. 
	The experiments are repeated 1000 times with nominal level at 5\%. For DGP3 and DGP4, $\{X_t\}$ takes value in the space of probability measures on $\R$ equipped with the 2-Wasserstein metric $d_W$. For DGP5, $\{X_t\}$ resides in the space of Graph Laplacian of weighted graphs equipped with Frobenius metric $d_{F}$ and for DGP6, $\{X_t\}$ takes value in the space of positive definite covariance matrices with the log-Euclidean metric $d_E$.
	
	Let $U_t = (U_{t,1},U_{t,2})$ be generated from the VAR(1) process $U_t = \rho I_2 U_{t-1}+\epsilon_t$, where $\{\epsilon_t\}$ are i.i.d standard bivariate normal random vectors, and define $\beta_t \sim^{i.i.d}Bernoulli(\delta_1)$, $e_t \sim^{i.i.d}EXP(1)$ such that $\{\epsilon_t\}$, $\{\beta_t\}$ and $\{e_t\}$ are independent. Assume $X_t = X_t^{(1)}$ for $1\leq t\leq \lfloor n/2 \rfloor$ and $X_t = X_t^{(2)}$ for $1{+} \lfloor n/2 \rfloor\leq t\leq n$, we consider the following four DGPs. DGP3: $X_{t}^{(1)}=\N \left(U_{t, 1}{+}1,1\right)$ and $X_{t}^{(2)}=\N \left((1{-}\beta_t)(U_{t, 1}{+}1){+}\beta_te_t,1\right)$. DGP4: $X_{t}^{(1)}=\N (\arctan (U_{t, 1}),[\arctan (U_{t, 1}^{2}){+}1]^{2})$ and $X_{t}^{(2)}=\N (\arctan (U_{t, 1}){+}\delta_{1}, \delta_{2}^{2}[\arctan (U_{t, 1}^{2}){+}1]^{2})$. DGP5: Graph Laplacians where each graph has five nodes that are categorized into two communities with two and three nodes respectively, and the edge weights for the first community, the second community and between communities are set as $0.4+\arctan (U_{t, 1}^{2}), 0.2+\arctan (U_{t, 2}^{2}), 0.1$ for $X_{t}^{(1)}$ and $\delta_{2}[0.4+\arctan (U_{t, 1}^{2})], \delta_{2}[0.2+\arctan (U_{t, 2}^{2})], 0.1+\delta_{1}$ for $X_{t}^{(2)}$. DGP6: Covariance matrix $X_{t}^{(i)}=\left(2I_{3}+{Z}_{t, i}\right)\left(2I_{3}+{Z}_{t, i}\right)^{\top}$ for $i=1,2$, where all the entries of ${Z}_{t, 1}$ (resp. $Z_{t,2}$) are independent copies of $\arctan \left(U_{t, 1}\right)$ (resp. $\delta_{1}+\delta_{2} \arctan \left(U_{t, 1}\right)$). Since the applicability of both DM and DZ is limited to independent object data, we shall not include them for comparison as they are expected to exhibit large size distortion when temporal dependence is present as in Section~\ref{sec_simu1}; also see \cite{jiang2024two}. 
	
	For the empirical size, we set $\delta_1{=}0$, $\delta_2{=}1$, $n\in\{200,400,800\}$ and $\rho\in\{-0.4,0.4,0,0.7\}$. As shown in Table \ref{tab_cpnon}, SS-SN is mildly over-sized when $\rho{=}0.7$ for DGP3 and DGP6, as well as when $n{=}200$, $\rho{=}0.7$ for DGP4 and DGP5. The size accuracy for $SN_1$ and SS-SN is comparable for $\rho\in\{-0.4,0,0.4\}$, while $SN_1$ has more accurate size when $\rho=0.7$. Similar to the empirical size result in Section \ref{sec_simu1}, $SN_2$ is under-sized for most combinations of $(n,\rho)$, except when $n{=}200$ and $\rho{=}0.7$ for DGP3 and DGP4, in which case $SN_2$ is over-sized.

		For the size adjusted power of DGP4, DGP5 and DGP6, we fix $n{=}400$, $\rho{=-}0.4$ and consider two alternatives. ALT1: Fix $\delta_2=1$ and plot the size adjusted power against $\delta_1\in[0,0.3]$. ALT2: Fix $\delta_1=0$ and plot the size adjusted power against $\delta_2\in[0.7,1]$. For DGP4 and DGP5 (see Section 5 \cite{jiang2024two}), the Fr\'{e}chet mean changes and the Fr\'{e}chet variance remains the same under ALT1, while ALT2 has a break in  Fr\'{e}chet variance with Fr\'{e}chet mean being held fixed. For DGP6, both Fr\'{e}chet mean and variance changes under ALT1 or ALT2. As shown in Figure \ref{fig_non}, SS-SN outperforms $SN_1$ and $SN_2$ in all settings except DGP2 under ALT2. Compared with SS-SN, the power loss for $SN_1$ and $SN_2$ are very large under ALT1 (as in Section \ref{sec_simu1}, $SN_2$ has trivial power under ALT1 for DGP4 and DGP5) and for DGP4 under ALT2, while all three methods have similar performance under ALT2 for DGP6.   
	
	For the size adjusted power of DGP3, we fix $n=800$, $\rho=-0.4$ and the size adjusted power against $\delta_1\in[0,1]$ is plotted in Figure \ref{fig_cp_dgp0}. Note that the Fr\'{e}chet mean and variance are fixed at $\N(1,1)$ and $1$ before and after the change. For this DGP, $SN_1$ and $SN_2$ both have trivial power while SS-SN successfully detected the change in marginal distribution.
\begin{table}
    %	\vspace{1em}
    \centering
    \setlength{\tabcolsep}{3pt}
    %\scalebox{0.95}{
        \begin{tabular}{cc|ccc|ccc|ccc|ccc}
            \toprule
            \midrule
            \multirow{2}{*}{$n$} &	\multirow{2}{*}{$\rho$} & 	\multicolumn{3}{c|}{DGP3} & 	\multicolumn{3}{c|}{DGP4}  & \multicolumn{3}{c|}{DGP5}& \multicolumn{3}{c}{DGP6} \\
            &        & SS-SN&$SN_1$ &$SN_2$   &    SS-SN&$SN_1$ &$SN_2$      &    SS-SN&$SN_1$ &$SN_2$  &    SS-SN&$SN_1$ &$SN_2$      \\\hline
            \multirow{4}{*}{200}&-0.4 &    3.9 &4.1   &3.4  &4.0 &4.3 &3.6 &4.7 &5.3 &   2.1    &3.7    &    5.4    &3.6 \\ \cline{3-14}
            &0    &    4.1 &5.0   &1.5  &3.5 &6.1 &3.2 &6.0 &4.5 &   2.9    &4.6    &    6.0    &2.8          \\ \cline{3-14} 
            &0.4  &    5.1 &6.3   &3.7  &6.3 &5.0 &4.1 &6.5 &4.6 &   3.5    &6.5    &    6.8    &1.8      \\ \cline{3-14} 
            &0.7  &    7.3 &7.6   &10.6 &8.5 &4.2 &10.5&7.4 &4.5 &   5.3    &9.7    &    6.9    &0.2      \\ \cline{2-14} 
            \multirow{4}{*}{400}    &-0.4 &    5.0 &6.2   &5.3  &4.8 &5.1 &3.8 &5.0 &5.4 &   2.2    &4.5       &    6.3 &6.0   \\ \cline{3-14}
            &0    &    4.1 &6.2   &3.3  &4.6 &6.0 &2.5 &5.5 &5.6 &   2.4    &7.0       &    5.7 &4.8      \\ \cline{3-14} 
            &0.4  &    5.4 &6.5   &2.6  &4.7 &6.4 &2.6 &4.6 &5.4 &   2.1    &7.1       &    6.4 &2.0    \\ \cline{3-14} 
            &0.7  &    5.7 &6.5   &6.3  &4.7 &5.2 &7.6 &5.6 &4.7 &   3.0    &8.0       &    7.1 &0.6     \\ \cline{2-14} 
            \multirow{4}{*}{800}&-0.4 &7.1     &5.6  &4.9  &4.6  &5.5 &4.7 &4.1 &5.0   &2.5     &6.0         &5.4     &5.2  \\ \cline{3-14}
            &0    &5.7     &5.3  &3.2  &5.1  &5.0 &3.2 &4.2 &5.2   &2.6     &5.3         &5.1     &4.4       \\ \cline{3-14} 
            &0.4  &6.3     &4.5  &2.2  &5.5  &4.2 &1.6 &5.1 &4.7   &1.7     &5.6         &5.8     &3.4   \\ \cline{3-14} 
            &0.7  &7.2     &4.8  &3.4  &5.4  &5.6 &4.0 &4.8 &5.2   &2.0     &6.1         &5.5     &1.2     \\ \bottomrule		
    \end{tabular}
    %}
    \caption{Empirical rejection rate (in percentage) under the null for DGP3-6.}
    \label{tab_cpnon}
\end{table}
	
 \begin{figure}
\begin{tikzpicture} 
\matrix (m) [row sep = 0em, column sep = 1.5em]{ 
    \node (p11) {\includegraphics[scale=0.42]{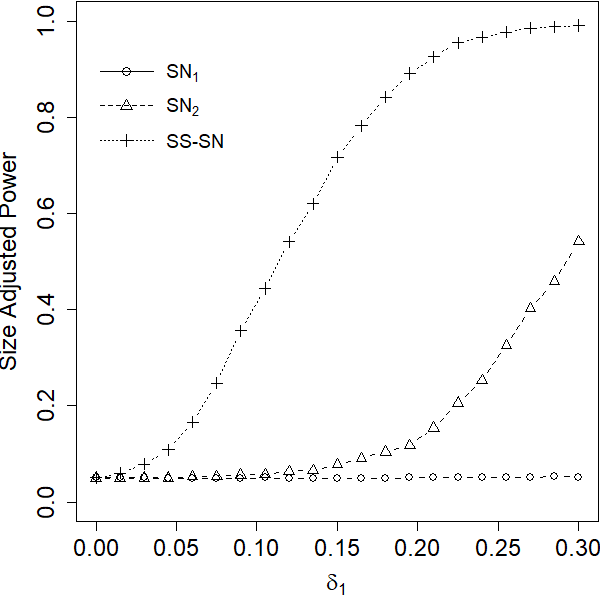}};
    %\node[above = 0cm of p11] (t11) {Manhattan};
    &
    \node (p12) {\includegraphics[scale=0.42]{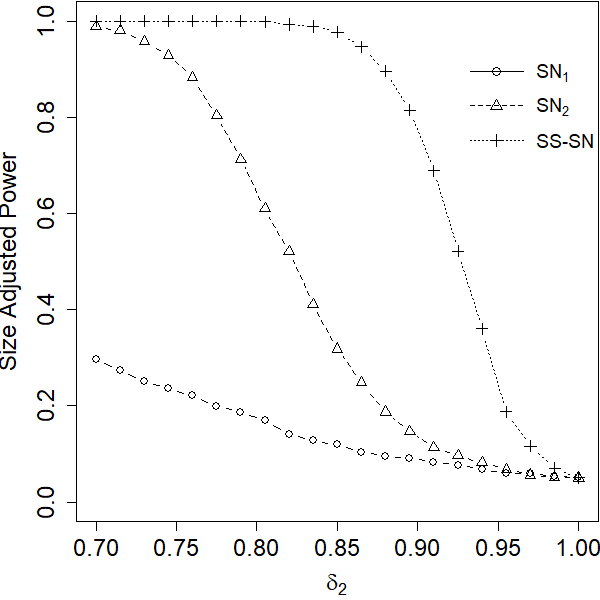}}; 
    %\node[above = 0cm of p12] (12) {Queens};
    \\
     \node (p21) {\includegraphics[scale=0.42]{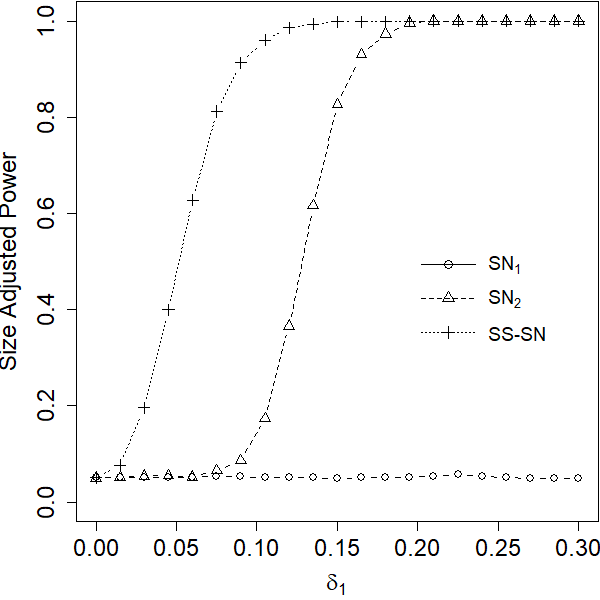}}; 
     %\node[above = 0cm of p21] (t21) {EWR};
     &
     \node (p22) {\includegraphics[scale=0.42]{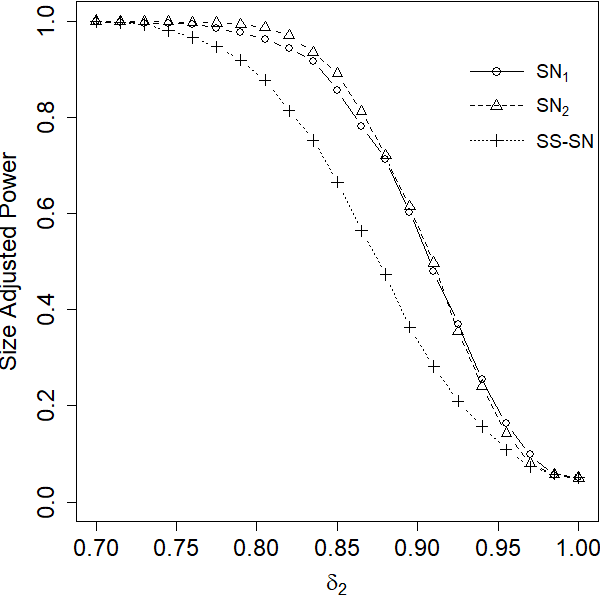}}; 
     %\node[above = 0cm of p22] (t22) {Brooklyn};
     \\
     \node (p31) {\includegraphics[scale=0.42]{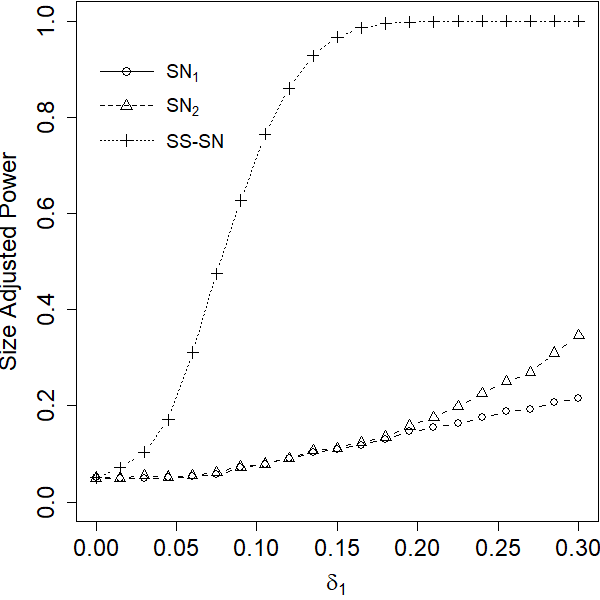}}; 
     %\node[above = 0cm of p31] (t31) {Bronx};
     &
     \node (p32) {\includegraphics[scale=0.42]{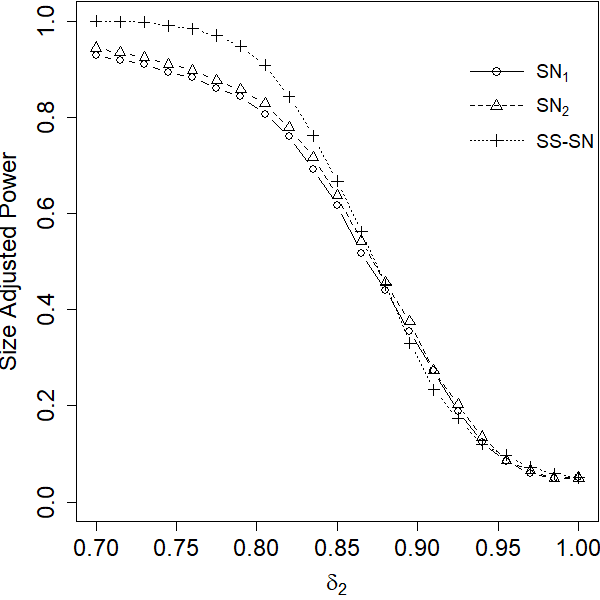}}; 
     %\node[above = 0cm of p32] (t32) {Staten Island};
     \\
};
\end{tikzpicture}
\centering
\caption{Size adjusted power for testing hypothesis on the existence of a change point for ALT1 (first column) and ALT2 (second column) and for DGP4 (first row), DGP5 (second row) and DGP6 (last row).}
\label{fig_non}
\end{figure}

\begin{figure}
    \centering
    \includegraphics[width=.6\textwidth]{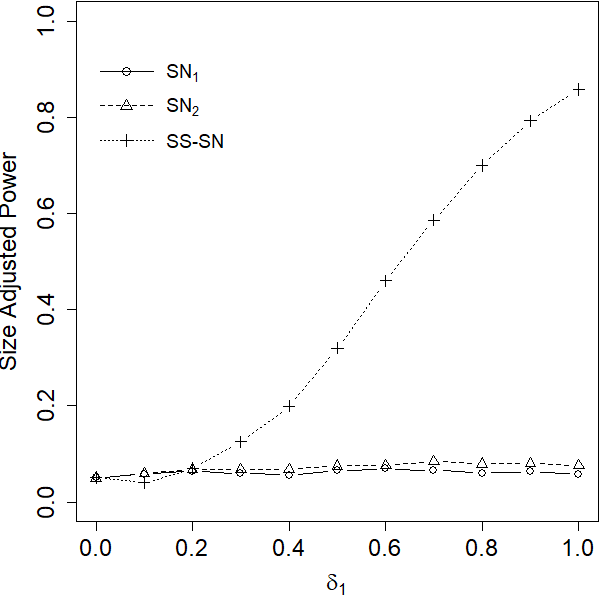}
    \caption{Size adjusted power for testing hypothesis on the existence of a change point for DGP3.}
    \label{fig_cp_dgp0}
\end{figure}

	\subsection{Multiple Change Point Estimation}\label{sec_simu3}
	
	For the setting with multiple change points, we compare our method with \cite{jiang2024two}, which combines WBS and SN$_2$ test statistic proposed therein for the object-valued time series. The sample size is set to be $n=200$ to be comparable with the sample size in applications. Similar to the simulations conducted in \cite{jiang2024two}, we employ Wasserstein distance and consider the following data generating process: for $t = 1,2, \dots, 200,$
	\begin{align*}
		X_t=\mathcal{N}(\arctan (U_{t})+\delta_{t,1},  \delta_{t, 2}^2[\arctan(U_{t}^2)+1]^2), 
	\end{align*}
	where $\{U_{t}\}$ is a AR(1) process, i.e.,  $U_t=\rho U_{t-1}{+}\epsilon_t$, $\epsilon_t\overset{i.i.d.}{\sim}\mathcal{N}(0,1)$; $\{ \delta_{t,1} \}$ and $\{\delta_{t,2}\}$ reflecting change points locations are defined as
	\begin{align*}
		& \delta_{t,1} = a_1 \mathbf{1}_{\{t \leq 65 \} }  + a_2 \mathbf{1}_{\{ 65 < t \leq 135 \} } + a_3 \mathbf{1}_{\{ 135 < t \leq 200 \} }, \\
		& \delta_{t,2} = b_1 \mathbf{1}_{\{ t \leq 65 \}}  + b_2 \mathbf{1}_{\{ 65 < t \leq 135 \}} + b_3 \mathbf{1}_{ \{ 135 < t \leq 200 \} }.
	\end{align*}
	Constants $a_i$ and $b_i$ determine the signal strength and we consider the following two cases: (i) $(a_1, a_2, a_3)=(0, 0.7, 0.1)$ and  $(b_1, b_2, b_3) = (1.2, 1.5, 1.8)$; (ii) $(a_1, a_2, a_3)=(-0.1, 0, 0.1)$ and  $(b_1, b_2, b_3) = (1, 1.6, 0.9)$.  The parameter $\rho$ determines the level of dependency between data points and its effects on the change point detection  are investigated in this simulation. Table \ref{tab:mult1} reports the proportions for detected number of change points and adjusted Rand index (ARI) based on 200 Monte Carlo repetitions. Note that ARI $\in [0, 1]$ measures the accuracy of
	change point estimation and larger ARI corresponds to more accurate estimation. 	We can observe that our method is uniformly better than \cite{jiang2024two} as reflected by ARI. As the performance declines for both methods with increasing values of $\rho$, the method proposed by \cite{jiang2024two} tends to more severely overestimate the number of change points.
	
\begin{table}
		\centering 
		\scalebox{0.89}{
			\begin{tabular}{@{\extracolsep{4pt}} cccccccc} 
				\\[-1.8ex]\hline 
				\hline \\[-1.8ex] 
				\multirow{2}{*}{$\rho$} & \multirow{2}{*}{Case} & \multirow{2}{*}{Method} & \multicolumn{4}{c}{Proportions (\%)} & \multirow{2}{*}{ARI} \\ \cline{4-7}
				& & & 0 & 1 & 2 & $\geq 3$ & \\
				\hline \\[-1.8ex] 
				\multirow{4}{*}{-$0.5$} &\multirow{2}{*}{(i)} & Ours & $0$ & $13.78$ & $84.69$ & $1.53$ & $0.889$ \\ 
				&  & Jiang et al. & $0.51$ & $45.92$ & $51.02$ & $2.55$ & $0.734$ \\ \cline{3-8}
				& \multirow{2}{*}{(ii)} & Ours & $0$ & $2.59$ & $95.85$ & $1.55$ & $0.956$ \\ 
				&  & Jiang et al. & $0$ & $5.70$ & $90.16$ & $4.15$ & $0.915$ \\ \cline{2-8}
				\multirow{4}{*}{$0$} & \multirow{2}{*}{(i)} & Ours & $0$ & $25.64$ & $74.36$ & $0$ & $0.835$ \\ 
				&  & Jiang et al. & $0.51$ & $37.43$ & $61.54$ & $0.51$ & $0.772$ \\  \cline{3-8}
				& \multirow{2}{*}{(ii)} & Ours & $0$ & $0$ & $99.5$ & $0.5$ & $0.971$ \\ 
				&  & Jiang et al. & $0$ & $4.5$ & $91.5$ & $4$ & $0.936$ \\ \cline{2-8} 
				\multirow{4}{*}{$0.5$} & \multirow{2}{*}{(i)} & Ours & $5$ & $48.5$ & $43.5$ & $3$ & $0.639$ \\ 
				&  & Jiang et al. & $9$ & $46.5$ & $36$ & $8.5$ & $0.589$ \\ \cline{3-8}
				& \multirow{2}{*}{(ii)} & Ours & $0$ & $8.67$ & $81.12$ & $10.20$ & $0.907$ \\ 
				&  & Jiang et al. & $2.55$ & $16.33$ & $62.76$ & $18.37$ & $0.796$ \\ \cline{2-8}
				\multirow{4}{*}{$0.7$} & \multirow{2}{*}{(i)} & Ours & $9$ & $43$ & $34$ & $14$ & $0.549$ \\ 
				&  & Jiang et al. & $10.5$ & $36$ & $32$ & $21.5$ & $0.525$ \\ \cline{3-8} 
				& \multirow{2}{*}{(ii)} & Ours & $1.51$ & $15.08$ & $57.29$ & $26.13$ & $0.800$ \\ 
				&  & Jiang et al. & $3.02$ & $23.62$ & $43.22$ & $30.15$ & $0.663$ \\ 
				\hline \\[-1.8ex] 
		\end{tabular} 
            }
		\caption{Simulation results for multiple change point detection problems based on 200 Monte Carlo repetitions. Proportions for detected number of change points and adjust Rand index (ARI) are reported.} 
		\label{tab:mult1} 
	\end{table}

	\section{Applications}\label{sec5_real}
	In this section, we apply our WBS-SN algorithm to datasets arising from the stock market and the taxi traffic in New York city. We set $b=0.15$, $M=50$ and $L_0=20$ for the WBS-SN algorithm.
	
	\subsection{Stock Return Data}\label{sec5_1}
	Density function for log-returns of S\&P 500 companies can reflect the overall stock market situation. In this study, we are interested in monitoring abrupt changes in marginal distributions of the time series of log-return densities. The daily log-returns of a stock can be downloaded using R package ``BatchGetSymbols" \citep{BatchGetSymbols}. For each month from year 2015 to year 2022, a density function is constructed from the daily log-returns of companies listed in S\&P 500, resulting in a total of 96 density functions. These densities are constructed using a kernel density estimator with the bandwidth set to be one-tenth of range between 1st and 99th percentiles of all log-returns. The heat map is shown in Figure \ref{fig:stock} with each column representing a density function evaluated on 51 grid points. Our test indicates significant changes in the marginal distributions of the monthly density time series before and after Jan. 2020, aligning with the onset of 2020 stock market downturn and increasing volatility due to the COVID-19 pandemic. 
	%For more details about 2020 stock market crash, we refer to the Wikipedia record at \url{https://en.wikipedia.org/wiki/2020_stock_market_crash}.
	
\begin{figure}
    \centering
    \begin{tikzpicture}
        \node (p1)  {\includegraphics[width=0.8\textwidth]{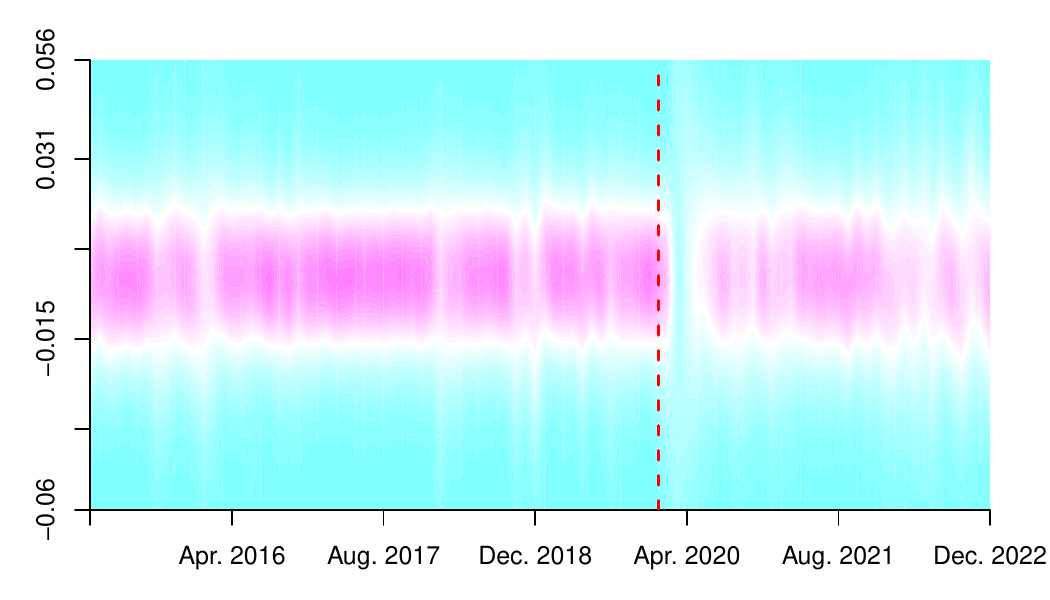}};
        \node[rotate=0,  right=-0.7cm of p1] (legend) {\includegraphics[scale=0.2]{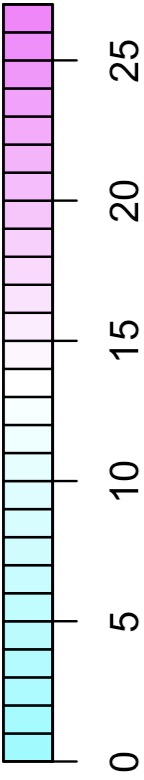}} ;
    \end{tikzpicture}
    \caption{Heat plot of 96 density functions. Each column represents a density function. The vertical dashed red line represents the detected change point location by our method.}
    \label{fig:stock}
\end{figure}

\subsection{TLC Trip Record Data}\label{sec5_2}
In this section, we apply our WBS-SN algorithm to datasets arising from the taxi traffic in New York city. To be specific, we aim to analyze the taxi traffic patterns in New York city through a change point perspective. The taxi traffic situations can be summarized using daily networks. Utilizing the TCL trip record data (\url{https://www.nyc.gov/site/tlc/about/tlc-trip-record-data.page}), for each day between January 2023 and June 2023, we construct a weighted undirected network comprising seven nodes. Six of these nodes represent Newark Airport (EWR), Queens, Bronx, Manhattan, Staten Island, Brooklyn in New York, while the seventh node represents areas outside these regions. For any two nodes $i,j \in \{ \text{EWR}, \text{Queens}, \text{Bronx}, \text{Manhattan}, \text{Staten Island}, \text{Brooklyn}, \\\text{Unknown} \}$, let $m_{ij, t}$, represents the number of taxi trips originating in one of the two nodes and ending in the other on day $t$. The weight $w_{ij,t}$ on the edge of the network on day $t$ between node $i$ and node $j$  is set to be $w_{ij,t} = m_{ij, t}/\max_t \{ m_{ij,t} \}$ if $\max_t \{ m_{ij,t} \}>0$. Self-edges, i.e., trips that start and end in the same region, are excluded from our analysis. This yields a network time series spanning 181 days. A selection of these daily networks are shown in Figure \ref{fig:taxi} with the edge width corresponding to the magnitude of $w_{ij,t}$. Two change points, located on January 31 and April 30, are identified by our method, indicating significant differences in the marginal distributions of the network time series before and after these dates. To better interpret the result, we plot the edge weights $\{w_{ij, t} : j \neq i \}$ as multivariate time series for each region in Figure \ref{fig:taxi2}. For Manhattan, a shift in mean of the edge weights are clearly observed before and after each of the change point locations.

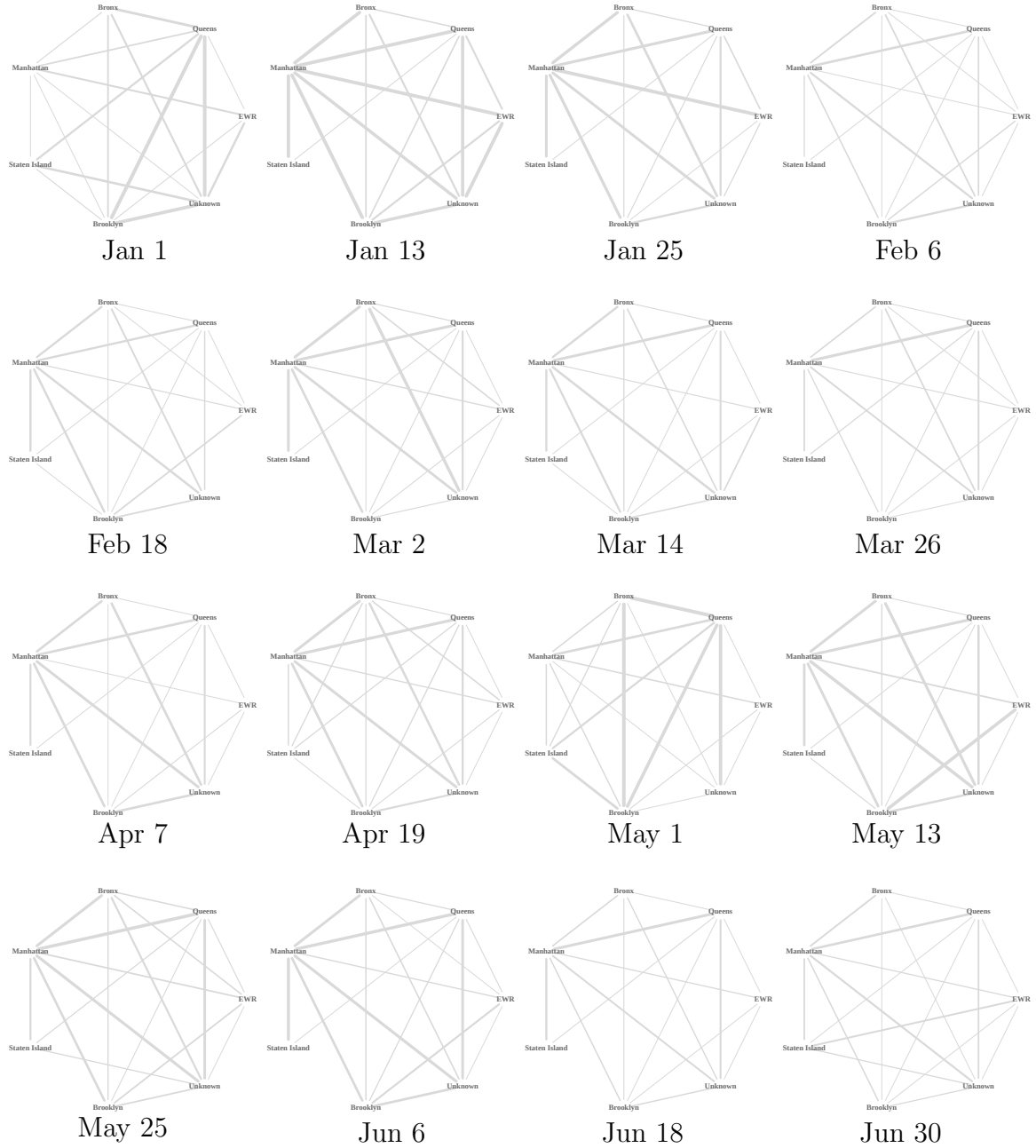
\begin{figure}
    \centering
    \begin{tikzpicture}
        \matrix (m) [row sep = 0em, column sep = - 1.2em]{    
            \node (p11) {\includegraphics[scale=0.4]{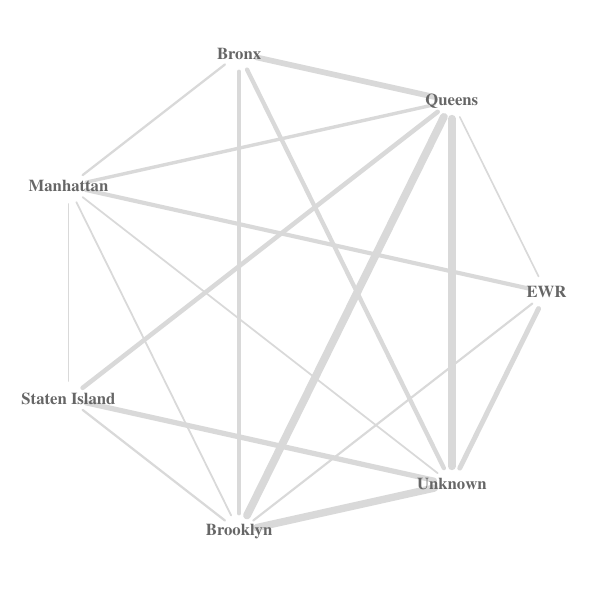}}; 
            \node[above left = -4.4cm and -2.8cm of p11] (t12) {Jan 1}; &
            \node (p12) {\includegraphics[scale=0.4]{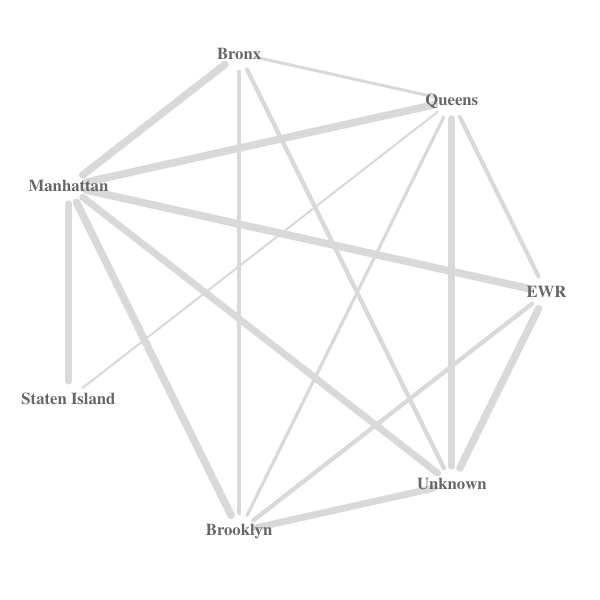}};
            \node[above left = -4.4cm and -2.8cm of p12] (t21) {Jan 13}; &
            \node (p13) {\includegraphics[scale=0.4]{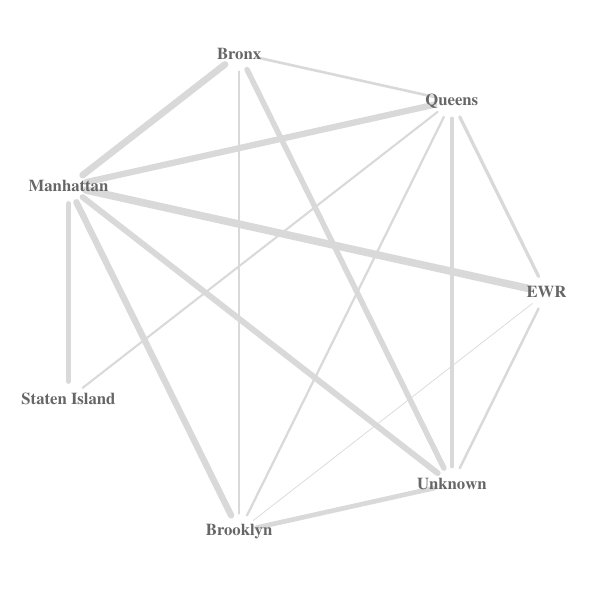}}; 
            \node[above left = -4.4cm and -2.8cm of p13] (t22) {Jan 25};
            &
            \node (p14) {\includegraphics[scale=0.4]{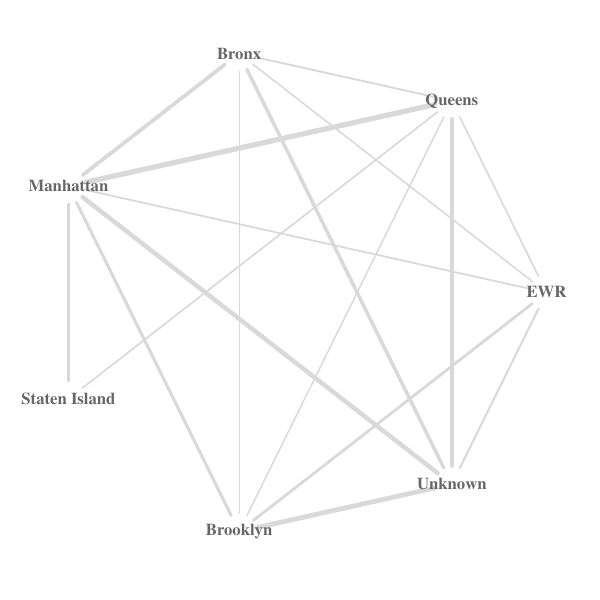}}; 
            \node[above left = -4.4cm and -2.8cm of p14] (t22) {Feb 6};
            \\ 
            \node (p21) {\includegraphics[scale=0.4]{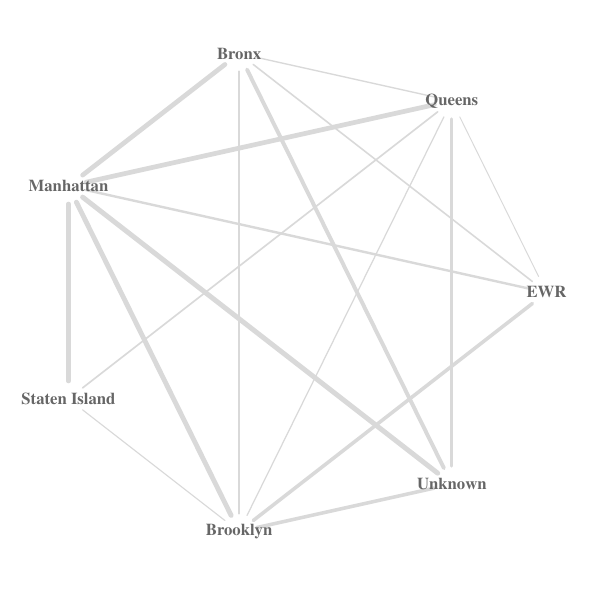}};
            \node[above left = -4.4cm and -2.8cm of p21] (t31) {Feb 18}; &
            \node (p22) {\includegraphics[scale=0.4]{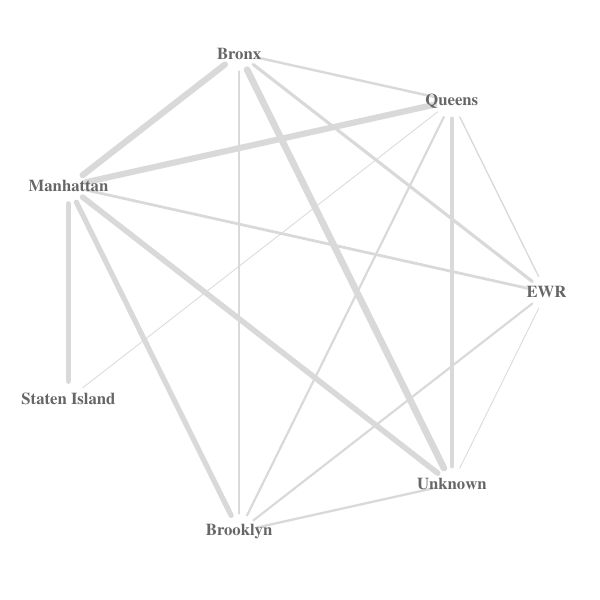}}; 
            \node[above left = -4.4cm and -2.8cm of p22] (t32) {Mar 2}; & 
            \node (p23) {\includegraphics[scale=0.4]{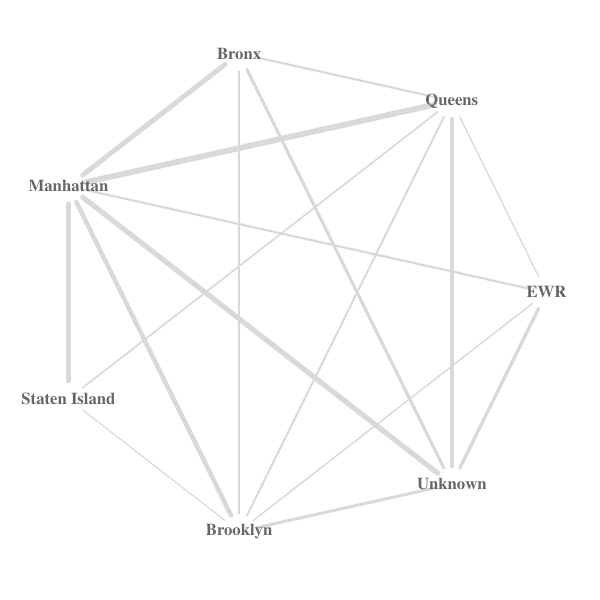}};
            \node[above left = -4.4cm and -2.8cm of p23] (t11) {Mar 14}; 
            & 
            \node (p24) {\includegraphics[scale=0.4]{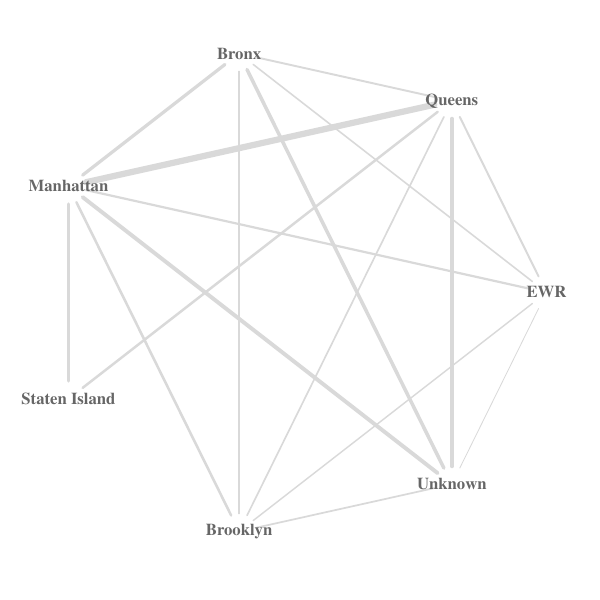}};
            \node[above left = -4.4cm and -2.8cm of p24] (t11) {Mar 26}; 
            \\
            \node (p31) {\includegraphics[scale=0.4]{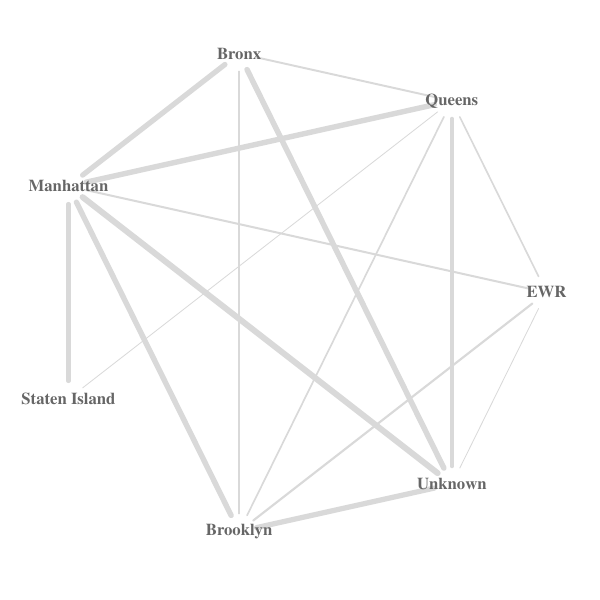}};
            \node[above left = -4.4cm and -2.8cm of p22] (t22) {Apr 7};
            & \node (p32) {\includegraphics[scale=0.4]{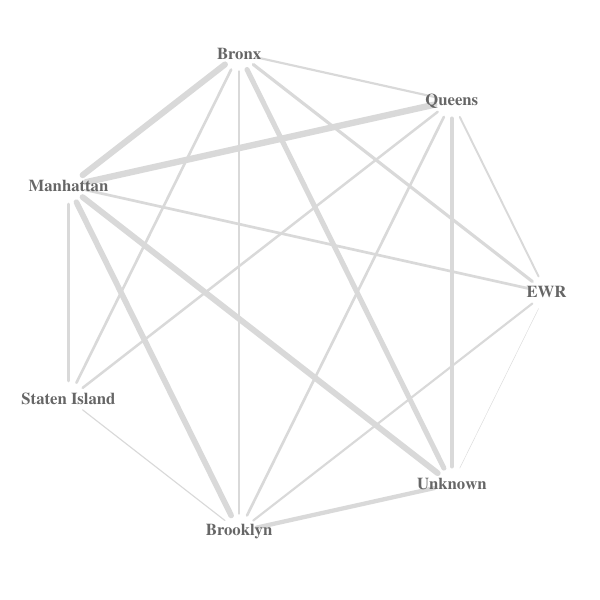}};
            \node[above left = -4.4cm and -2.8cm of p32] (t31) {Apr 19};
            & \node (p33) {\includegraphics[scale=0.4]{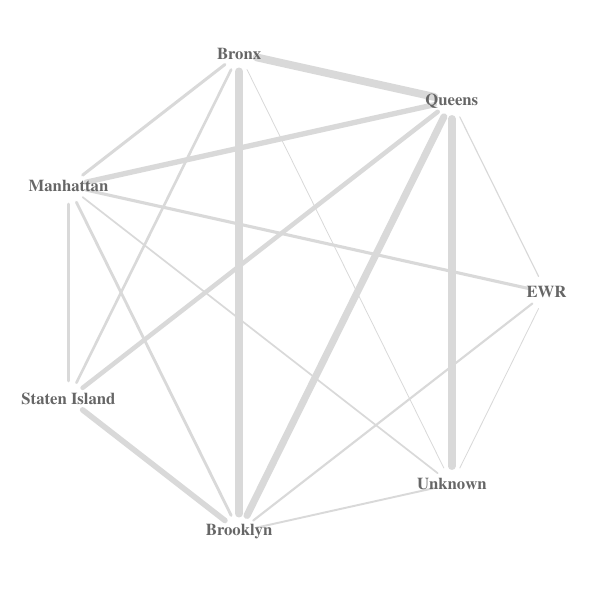}}; 
            \node[above left = -4.4cm and -2.8cm of p33] (t32) {May 1};
            & \node (p34) {\includegraphics[scale=0.4]{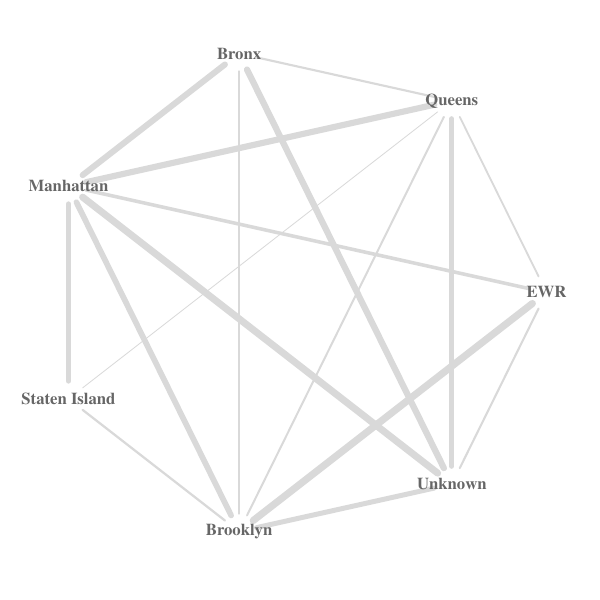}}; 
            \node[above left = -4.4cm and -2.8cm of p34] (t32) {May 13};
            \\
            \node (p41) {\includegraphics[scale=0.4]{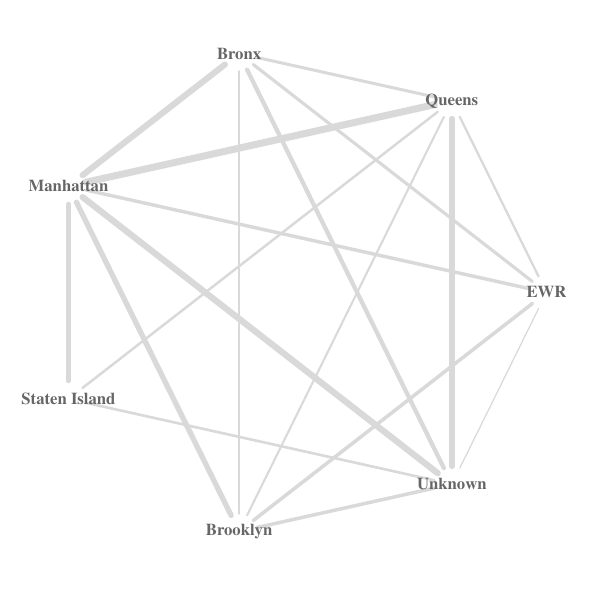}};
            \node[above left = -4.4cm and -2.8cm of p22] (t22) {May 25};
            & \node (p42) {\includegraphics[scale=0.4]{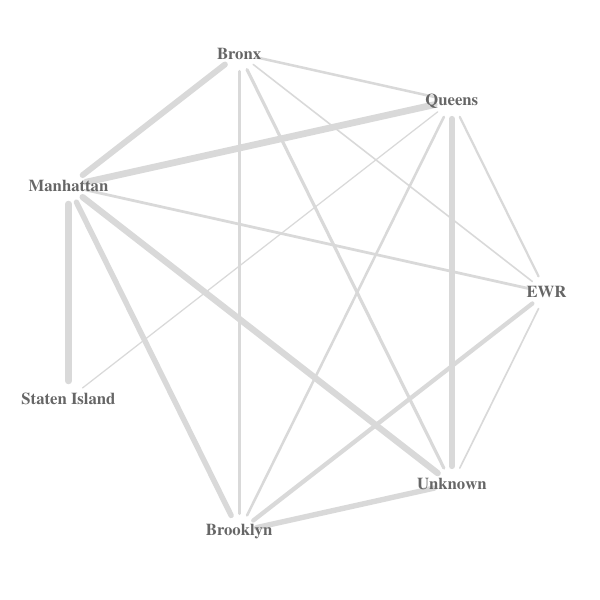}};
            \node[above left = -4.4cm and -2.8cm of p32] (t31) {Jun 6};
            & \node (p43) {\includegraphics[scale=0.4]{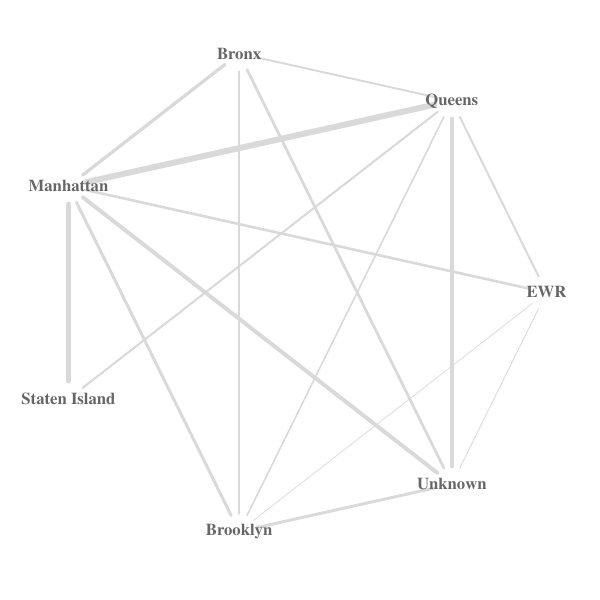}}; 
            \node[above left = -4.4cm and -2.8cm of p33] (t32) {Jun 18};
            & \node (p44) {\includegraphics[scale=0.4]{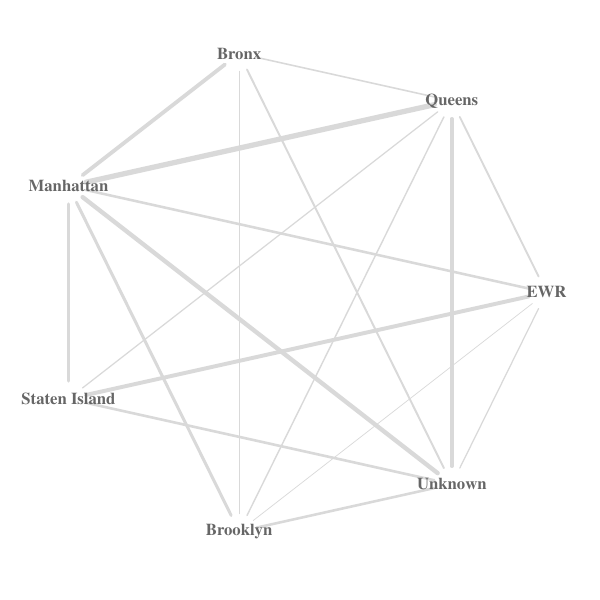}}; 
            \node[above left = -4.4cm and -2.8cm of p34] (t32) {Jun 30};
            \\
        };
        %\node[rotate=270,  below right=  -0.6cm and 0cm of p32] (legend) {\includegraphics[scale=0.25]{laxlegend.jpeg}} ;
    %	\node[above= -0.0cm and 0cm  of m] (title) {Time Series of Taxi Trip Networks} ;
    \end{tikzpicture}
    \centering
    \caption{Visualization of weighted undirected taxi-trip networks indexed by date.}
    \label{fig:taxi}
\end{figure}

\begin{figure}
\begin{tikzpicture} 
\matrix (m) [row sep = 1em, column sep = -0.0em]{ 
    \node (p11) {\includegraphics[scale=0.44] {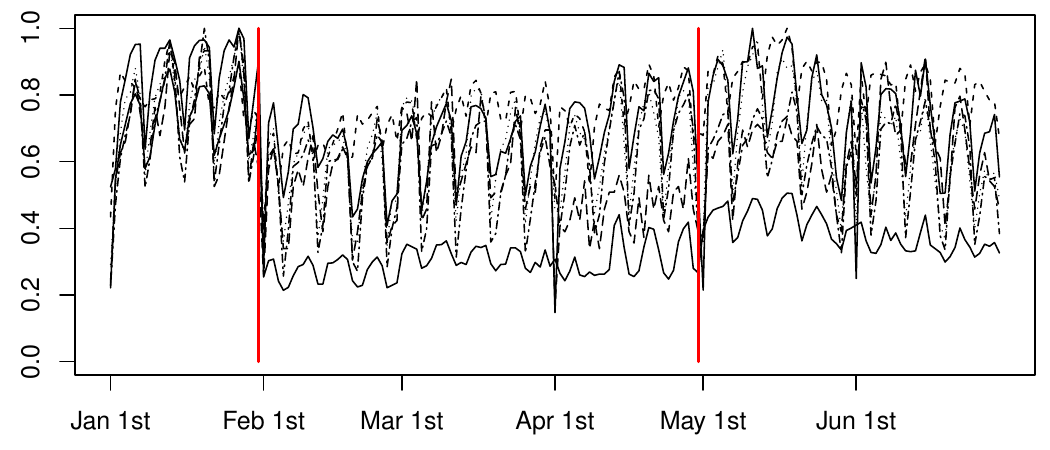}};
    \node[above = 0cm of p11] (t11) {Manhattan};
    &
    \node (p12) {\includegraphics[scale=0.44]{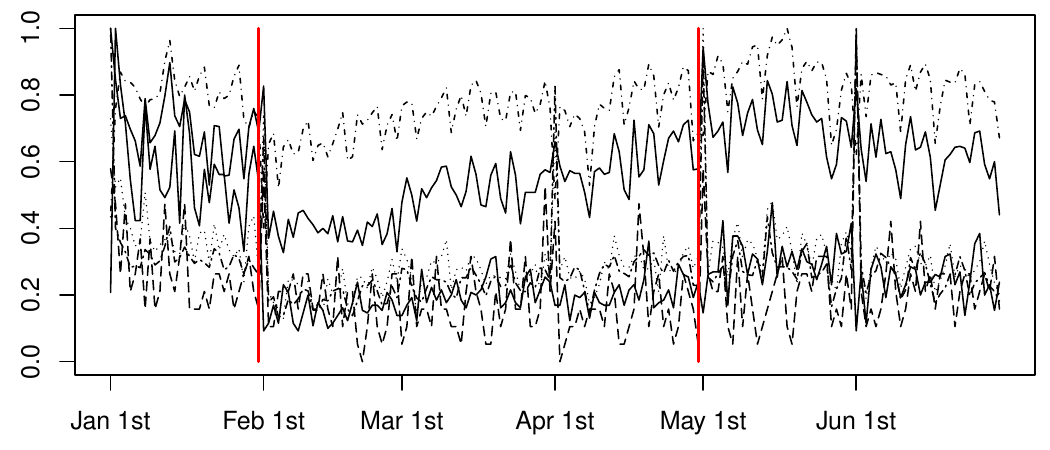}}; 
    \node[above = 0cm of p12] (12) {Queens};
    \\
     \node (p21) {\includegraphics[scale=0.44]{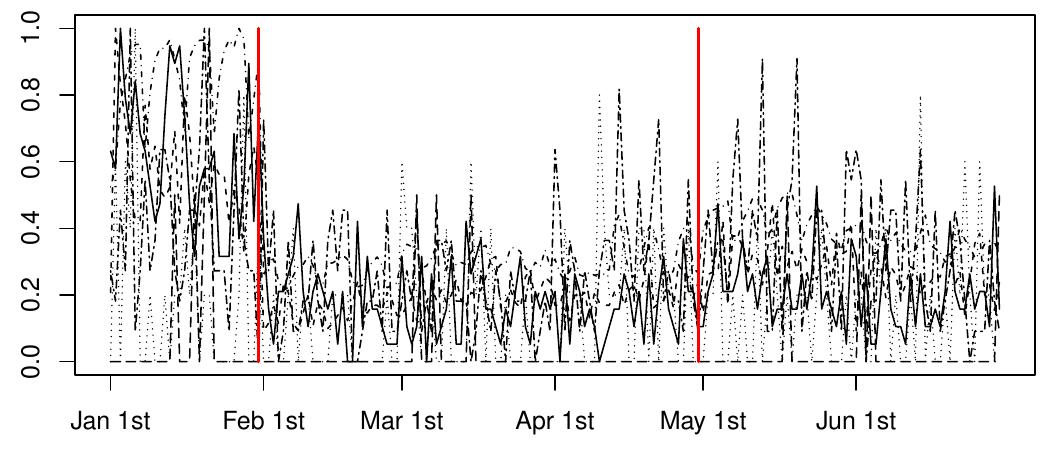}}; 
     \node[above = 0cm of p21] (t21) {EWR};
     &
     \node (p22) {\includegraphics[scale=0.44]{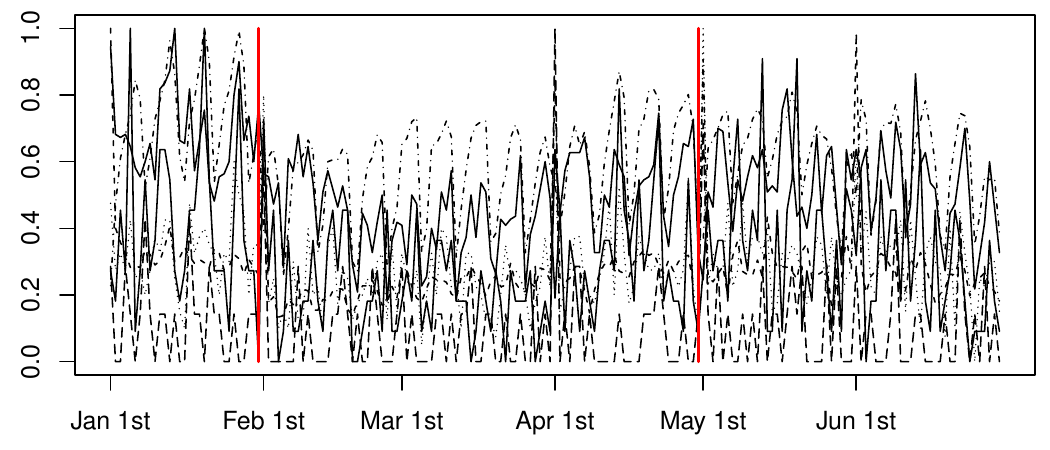}}; 
     \node[above = 0cm of p22] (t22) {Brooklyn};
     \\
     \node (p31) {\includegraphics[scale=0.44]{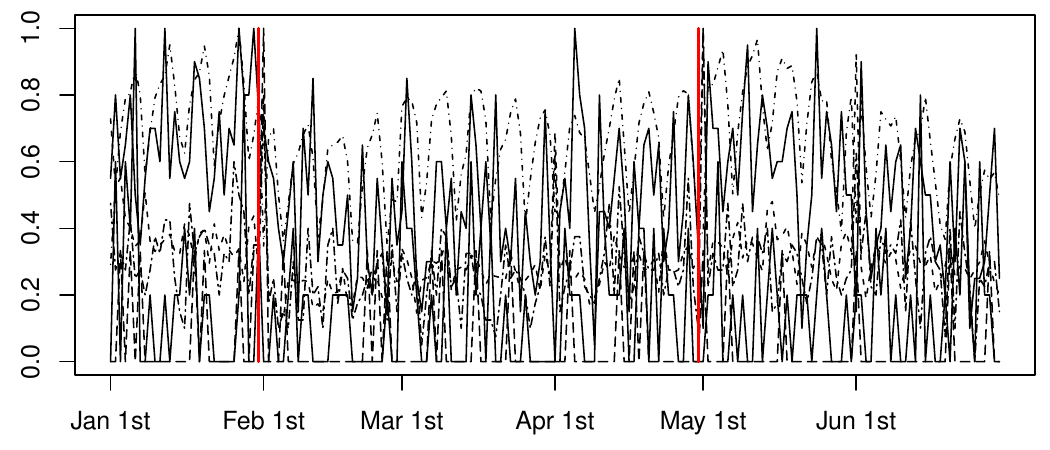}}; 
     \node[above = 0cm of p31] (t31) {Bronx};
     &
     \node (p32) {\includegraphics[scale=0.44]{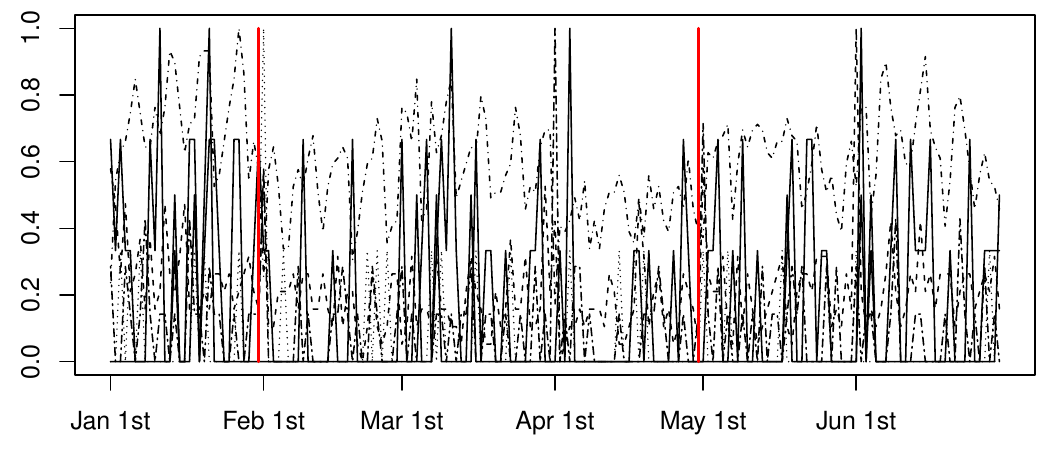}}; 
     \node[above = 0cm of p32] (t32) {Staten Island};
     \\
};
\end{tikzpicture}
\centering
\caption{For each node $i \in \{\text{Manhattan}, \text{Queens}, \text{EWR}, \text{Brooklyn},  \text{Bronx},  \text{Staten Island} \}$, the edge weights $\{w_{ij, t} : j \neq i \}$ are plotted as multivariate time series.}
\label{fig:taxi2}
\end{figure}
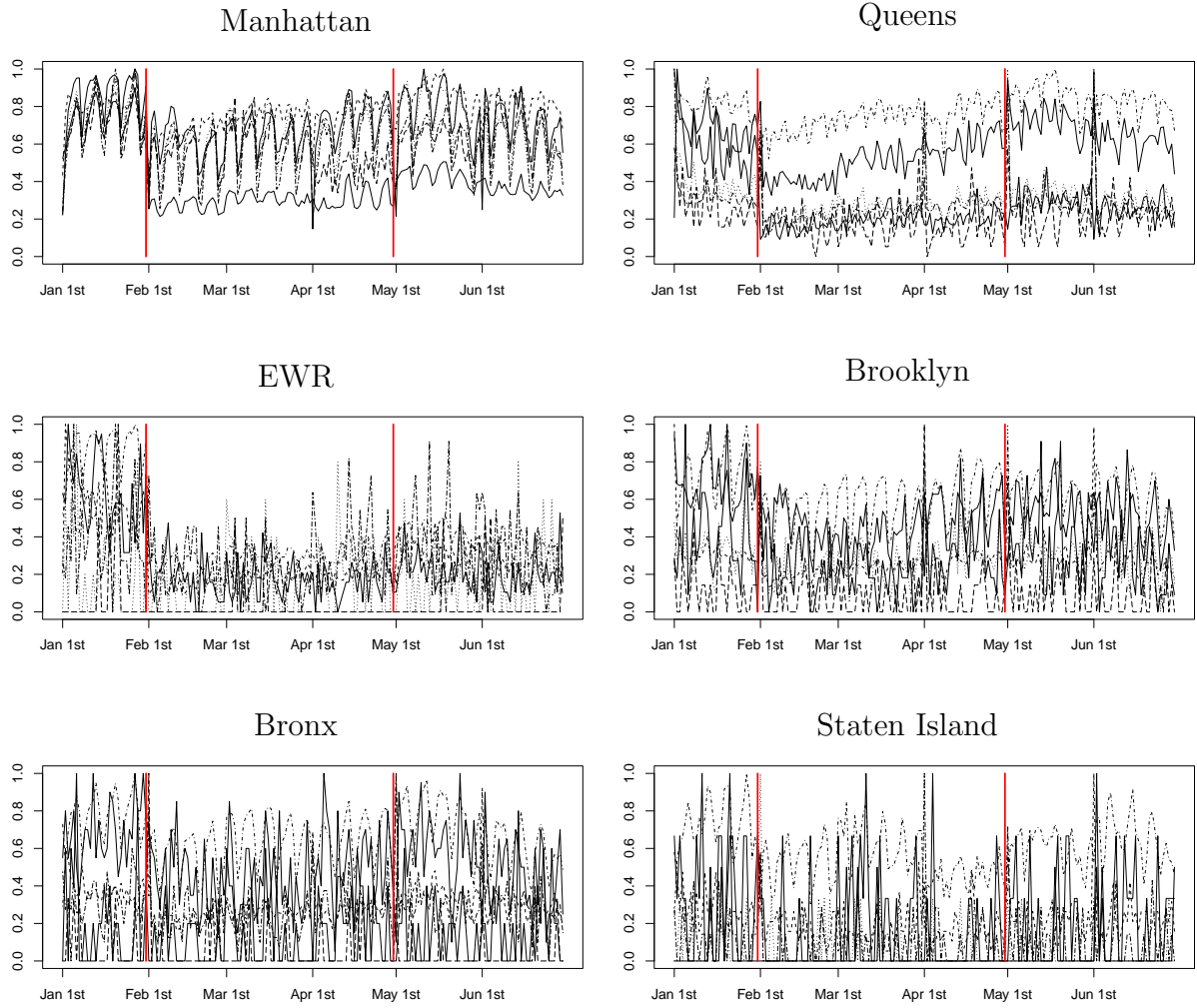

\section{Conclusion}\label{sec6conclu}
	
	In this paper, we propose a novel change point detection method for a broad class of object-valued time series that can detect marginal distributional changes. By using Hilbert space embedding, sample splitting and self-normalization, the statistic we developed is almost tuning parameter free, has pivotal limiting null distribution and only uses the distance between two non-Euclidean data objects. For multiple change point estimation, the estimation consistency and convergence rate of change point estimators are established for the proposed WBS-SN algorithm,  and the convergence rates for each change point estimator are adaptive to the heterogeneous change magnitudes. To the best of our knowledge, we are the first to show the WBS estimation consistency for a broad class of object-valued time series and in a nonparametric setting. Extensive simulations are conducted to show the accurate size and satisfactory power performance of the proposed method in comparison with existing ones. The data examples further illustrate the versatility of our proposed method in real-world applications. 
	
To conclude, we mention several possible extensions. For multiple change point estimation, this paper only considers unknown but fixed number of change points. It would be helpful to consider the setting where the number of change points can  grow with sample size. Also, it may be desirable to relax the abrupt change assumption and develop a test that captures gradual change. 
	
	%allow the parameter to vary smoothly within each segment.
	
\section*{Data Availability Statement}
All data used in this research are publicly available. The stock return data can be freely downloaded from Yahoo Finance (\url{https://finance.yahoo.com/}) and the TCL trip record data is available at \url{https://www.nyc.gov/site/tlc/about/tlc-trip-record-data.page}.
	
%\section*{Acknowledgement}
%Zhu's research is partially supported by NSF grant DMS-2412832.
%Shao's research is partially supported by NSF grants DMS-2210002 and  DMS-2412833 and research board grants at University of Illinois.

	%\newpage
		{\fontsize{12.95}{12.95}\selectfont % \footnotesize \scriptsize
	  \setlength{\parskip}{0pt}
	\setlength{\bibsep}{0pt plus 0.24ex}
	
\bibliographystyle{apalike}
	
\bibliography{Bib_SS_infty}
	}
	%\printbibliography

%\begin{comment}
	\newpage
	\bigskip

	\begin{center}
		{\large\bf SUPPLEMENT TO ``Change-Point Detection for Object-valued Time Series''}
	\end{center}
	
	\begin{center}
		%{BY Yi Zhang, Changbo Zhu AND Xiaofeng Shao}
	\end{center}

The supplement is organized as follows. Appendix \ref{app_th_cp} and \ref{appen_proof_est} contain the proofs for Theorem \ref{theorem_cp} and \ref{cp_est} separately. The estimation consistency of WBS-SN algorithm (Theorem \ref{cp_est_wbs}) is shown in Appendix \ref{app_th_wbs}.

	\appendix
	
	\renewcommand{\theequation}{A-\arabic{equation}}
	% redefine the command that creates the equation no.
	\setcounter{equation}{0}  % reset counte

	\section{Proof of Theorem \ref{theorem_cp}}\label{app_th_cp}
	
	\begin{proof}
		
		First we prove part (i) of Theorem \ref{theorem_cp}. By simple calculation we can see that
		\begin{align}
			\bar Y_1-\bar Y_3 \stackrel{D}{\to} B_Q(b)-\big[B_Q(1)-B_Q(1-b)\big]{=}\tilde Y.\nonumber
		\end{align}
		Also, let $k= \lfloor \tau n \rfloor $ for some $\tau\in(b,1-b)$, by the continuous mapping theorem, we can derive
		\begin{align}
			T_n( k)& \stackrel{D}{\to}\frac{1}{\sqrt{1{-}2b}} \Big\{\li \tilde Y, B_Q(\tau)-B_Q(b)\ri_\cH-\frac{\tau-b}{1-2b}\li\tilde Y,B_Q(1-b)-B_Q(b)\ri_\cH \Big\} \nonumber\\ 
			& {=}T(\tilde Y,\tau) ~\mbox{and}\nonumber 
		\end{align}
		\begin{align}
			V_n( k) \stackrel{D}{\to}& \frac{1}{(1-2b)^2} \Big\{\int_b^{\tau} \big\{\li\tilde Y,B_Q(s)-B_Q(b)\ri_\cH - \frac{(s-b)}{\tau-b}\li\tilde Y,B_Q(\tau)-B_Q(b)\ri_\cH\big\}^2 ds  \nonumber\\
			&\qquad\qquad + \int_{\tau}^{1-b} \big\{\li\tilde Y,B_Q(1-b)-B_Q(s)\ri_\cH - \frac{1-b-s}{1-b-\tau} \li\tilde Y,B_Q(1-b)-B_Q(\tau)\ri_\cH \big\}^2 ds\Big\} \nonumber\\
			{=}&V(\tilde Y,\tau),\nonumber
		\end{align}
		so the limiting distribution of $G_n$ is $\sup_{\tau \in [b,1-b]}\frac{T(\tilde Y,\tau)}{\sqrt{V(\tilde Y,\tau)}}$. By the orthogonal increment property of Brownian motion, we have $\{B_Q(r)\}_{r\in[0,b]}$, $\{B_Q(s)-B_Q(b)\}_{s\in[b,\tau]}$, $\{B_Q(1-b)-B_Q(t)\}_{t\in[\tau,1-b]}$ and $\{B_Q(x)-B_Q(1-b)\}_{x\in[1-b,1]}$ are independent, which implies that the conditional distribution of $\sup_{\tau \in [0,1]}\frac{T(\tilde Y,\tau)}{\sqrt{V(\tilde Y,\tau)}}$ given $\tilde Y$ is 
		
		\footnotesize 
		\begin{align}
			&\sup_{\tau \in [b,1-b]}\frac{T(\tilde Y,\tau)}{\sqrt{V(\tilde Y,\tau)}} \Big| \tilde Y \stackrel{d}{=} \sup_{\tau \in [b,1{-}b]}\frac{\frac{T(\tilde Y,\tau)}{{ \sqrt{\li Q\tilde Y,\tilde Y\ri_\cH}}}}{\sqrt{\frac{V(\tilde Y,\tau)}{{ \li Q\tilde Y,\tilde Y\ri_\cH}}}} \Big| \tilde Y \nonumber\\
			\stackrel{d}{=} &\sup_{\tau \in [b,1{-}b]}  \frac{\sqrt{1{-}2b}\Big[B(\tau){-}B(b){-} \frac{\tau{-}b}{1{-}2b}(B(1{-}b){-}B(b))  \Big]}  {\Big\{\int_{b}^{\tau} \big[B(s){-}B(b){-}\frac{s{-}b}{\tau{-}b}(B(\tau){-}B(b))\big]^2 ds + \int_{\tau}^{1{-}b}\big[ B(1{-}b){-}B(s){-}\frac{1{-}b{-}s}{1{-}b{-}\tau}(B(1{-}b){-}B(\tau))    \big]^2 ds    \Big\}^{1/2}   }  \label{cpeq11}\\
			\stackrel{d}{=}& \sup_{\tau \in [b,1{-}b]}  \frac{\sqrt{1{-}2b}\Big[B(\tau){-}B(b){-} \frac{\tau{-}b}{1{-}2b}(B(1{-}b){-}B(b))  \Big]}  {\Big\{\int_{0}^{\tau{-}b} \big[B(s+b){-}B(b){-}\frac{s}{\tau{-}b}(B(\tau){-}B(b))\big]^2 ds + \int_{\tau{-}b}^{1{-}2b}\big[B(1{-}b){-}B(s+b){-}\frac{1{-}2b{-}s}{1{-}b{-}\tau}(B(1{-}b){-}B(\tau))    \big]^2 ds   \Big\}{1/2}    }  \label{cpeq21}\\ 
			\stackrel{d}{=}& \sup_{\tau \in [b,1{-}b]}  \frac{\sqrt{1{-}2b}\Big[B(\tau{-}b){-} \frac{\tau{-}b}{1{-}2b}(B(1{-}2b)) \Big]}  {\Big\{\int_{0}^{\tau{-}b} \big[B(s){-}\frac{s}{\tau{-}b}B(\tau{-}b)\big]^2 ds + \int_{\tau{-}b}^{1{-}2b}\big[B(1{-}2b){-}B(s){-}\frac{1{-}2b{-}s}{1{-}b{-}\tau}(B(1{-}2b){-}B(\tau{-}b)) \big]^2 ds\Big\}^{1/2} } \label{cpeq31}\\
			\stackrel{d}{=}& \sup_{r \in [0,1]}  \frac{\sqrt{1{-}2b}\Big[B((1{-}2b)r){-} rB(1{-}2b) \Big]}  {\Big\{\int_{0}^{(1{-}2b)r} \big[B(s){-}\frac{s}{(1{-}2b)r}B((1{-}2b)r)\big]^2 ds + \int_{(1{-}2b)r}^{1{-}2b}\big[B(1{-}2b){-}B(s){-}\frac{1{-}2b{-}s}{(1{-}2b)(1{-}r)}(B(1{-}2b){-}B((1{-}2b)r)) \big]^2 ds\Big\}^{1/2} } \label{cpeq41} \\
			\stackrel{d}{=}& \sup_{r \in [0,1]}  \frac{B(r){-} rB(1) }  {\Big\{\int_{0}^{r} \big[B(s){-}\frac{s}{r}B(r)\big]^2 ds + \int_{r}^{1}\big[B(1){-}B(s){-}\frac{1{-}s}{1{-}r}(B(1){-}B(r)) \big]^2 ds \Big\}^{1/2}}. \label{cpeq51}
		\end{align}
		\normalsize
		For Equation (\ref{cpeq11}), we used the fact that for any fixed $y\in\cH$, $\frac{T(y,\tau)}{\sqrt{V( y,\tau)}}$ is independent of $ \tilde Y$. For Equation (\ref{cpeq21}) we used the change of variable property for integration. For Equation (\ref{cpeq31}) we used the property $\{B(r+b)-B(b):r\in[0,1-2b]\}\stackrel{d}{=}\{B(r):r\in[0,1-2b]\}$. For Equation (\ref{cpeq41}) we substitute $\tau-b$ with $(1-2b)r$ and for Equation (\ref{cpeq51}) we used the change of variable property for integration and the fact $\{B((1-2b)r):r\in[0,1]\} \stackrel{d}{=}\{\sqrt{1-2b}B(r):r\in[0,1]\}$.
		
		To prove part (ii).1 of Theorem \ref{theorem_cp}, denote $Y_t' = Y_t-\mu_t$, $\bar Y'_1 =\frac{1}{\sqrt{n}} \sum_{t=1}^{\fb}Y'_t$, $\bar Y'_3 = \frac{1}{\sqrt{n}} \sum_{t=n-\fb+1}^{n} Y'_t$ and $Z'_t=\li\bar Y'_1-\bar Y'_3,Y'_t\ri_\cH$. Let $S_{a,b}^A=\sum_{t=a}^{b}A_t$ for $A=Y',Z',\mu$, we have 
		\footnotesize
		\begin{align}
			&T_n( k^\ast) \nonumber\\
			=& \frac{1}{\sqrt{n-2\fb}}\sum_{t=\fb+1}^{k^\ast}\Big(Z_{t}-\frac {S_{\fb+1,n-\fb}}{n-2\fb}\Big) \nonumber \\
			=&\frac{S^{Z'}_{\fb{+}1,k^\ast}}{\sqrt{n-2\fb}}{-}\frac{\fb}{\sqrt{n(n-2\fb)}}\li \Delta_n, S^{Y'}_{\fb{+}1,k^\ast}\ri_\cH
			{-}\frac{\fb}{\sqrt{n(n-2\fb)}}\li \Delta_n, S^{\mu}_{\fb+1,k^\ast}\ri_\cH
			{+}\frac{\li\bar Y'_1{-}\bar Y'_3, S^{\mu}_{\fb+1,k^\ast}\ri}{\sqrt{n-2\fb}} \nonumber \\
			&{-}\frac{k^\ast-\fb}{\sqrt{(n-2\fb)^3}}\Big\{ S^{Z'}_{\fb{+}1,n{-}\fb} {-}     \frac{\fb}{\sqrt{n}}\li \Delta_n, S^{Y'}_{\fb{+}1,n{-}\fb}\ri_\cH               
			{-}     \frac{\fb}{\sqrt{n}}\li \Delta_n, S^{\mu}_{\fb{+}1,n{-}\fb}\ri_\cH      
			{+}     \li\bar Y'_1{-}\bar Y'_3, S^{\mu}_{\fb+1,n{-}\fb}\ri_\cH        \Big\}. \label{eqap13}
		\end{align}
		\normalsize
		
		Note that the summation of $3$rd and $7$th term on the right hand side of Equation (\ref{eqap13}) dominates, and their sum is
		\begin{align}
			&{-}\frac{\fb}{\sqrt{n(n-2\fb)}}\li \Delta_n, S^{\mu}_{\fb+1,k^\ast}\ri_\cH {+}\frac{k^\ast-\fb}{\sqrt{(n-2\fb)^3}}  \frac{\fb}{\sqrt{n}}\li \Delta_n, S^{\mu}_{\fb{+}1,n{-}\fb}\ri_\cH      \nonumber \\
			= &\frac{(n{-}\fb{-}k^\ast)(k^\ast{-}\fb)}{\sqrt{(n-2\fb)^3}}  \frac{\fb}{\sqrt{n}}\|\Delta_n\|^2_\cH,
		\end{align}
		So $T_n( k^\ast)\stackrel{p}{\to}\infty$ at the same rate as $n\|\Delta_n\|^2_\cH$. For $V_n( k^\ast)$, note that when $t=k^\ast+1,\dots,n-\fb$,
		\begin{align}
			&\Big\{S_{t,n-\fb}-\frac{n{-}\fb{-}t{+}1}{n{-}\fb{-}k^\ast}S_{k^\ast{+}1,n{-}\fb}\Big\}^2 \nonumber \\
			\leq & \Big\{S^{Z'}_{t,n-\fb}-\frac{n{-}\fb{-}t{+}1}{n{-}\fb{-}k^\ast}S^{Z'}_{k^\ast{+}1,n{-}\fb}\Big\}^2 {+}\frac{\fb^2}{n}\li \Delta_n,\Big\{S^{Y'}_{t,n-\fb}-\frac{n{-}\fb{-}t{+}1}{n{-}\fb{-}k^\ast}S^{Y'}_{k^\ast{+}1,n{-}\fb}\Big\} \ri_\cH^2 \nonumber \\
			=&\D_1(t)+\D_2(t). \nonumber
		\end{align}
		Since $(n-2\fb)^{-2}\sum_{t=k^\ast+1}^{n-\fb}\D_1(t)=O_p(1)$ and
		\begin{align}
			(n-2\fb)^{-2}\sum_{t=k^\ast+1}^{n-\fb}\D_2(t)\leq& \nonumber \frac{\fb^2\|\Delta\|^2_\cH}{n}(n-2\fb)^{-2}\sum_{t=k^\ast+1}^{n-\fb}\|S^{Y'}_{t,n-\fb}-\frac{n{-}\fb{-}t{+}1}{n{-}\fb{-}k^\ast}S^{Y'}_{k^\ast{+}1,n{-}\fb}\|_\cH^2 \nonumber \\
			=&n\|\Delta\|^2_\cH O_p(1), \nonumber
		\end{align}
		we have 
		\begin{align}
			(n-2\fb)^{-2}\sum_{t=k^\ast+1}^{n-\fb}\Big\{S_{t,n-\fb}-\frac{n-\fb-t+1}{n-\fb-k^\ast}S_{k^\ast+1,n-\fb}\Big\}^2 = O_p(n\|\Delta\|^2_\cH). \label{eqap2}
		\end{align}
		For the same reason, we can show that 
		\begin{align}
			(n-2\fb)^{-2}\sum_{t=\fb+1}^{k^\ast}\Big\{S_{\fb+1,t}-\frac{t-\fb}{k^\ast-\fb}S_{\fb+1,k^\ast}\Big\}^2 = O_p(n\|\Delta\|^2_\cH).\label{eqap3}
		\end{align}
		Combining Equations (\ref{eqap2}) and (\ref{eqap3}), we have $V_n(k^\ast)=O_p(n\|\Delta\|^2_\cH)$, so $G_n\geq \frac{T_n(k^\ast)}{\sqrt{V_n(k^\ast)}}\stackrel{P}{\to}\infty$ and part (ii).1 of Theorem \ref{theorem_cp} is proved.

		To prove part (ii).2 and (ii).3, for $r\in[0,1]$ define ${F}_n(r)=n^{-1/2}\sum_{t=1}^\fr({Y}_{t} -{\mu}_t)$, $ {\tilde F}_n(r)=n^{-1/2}\sum_{t=1}^{k^\ast}({Y}_{t} -{\mu}_t)+n^{-1/2}\sum_{t=k^\ast+1}^{\fr}({Y}_{t} -{\mu}_t+{\Delta}_n)$ and ${H}_n(r)=n^{-1/2}\sum_{t=k^\ast+1}^\fr{\Delta}_n$ (For any nonnegative integer $a,b$ and sequence $\{y_t\}\subset \cH$, we set $\sum_{t=a}^{b}y_t=0$ if $a<b$). By Assumption \ref{assump1}, ${F}_n(r)={\tilde F}_n(r)-{H}_n(r)\Rightarrow B_Q(t)$. Since ${H}_n(r)$ converges to ${H}(r)=(r-r_0){c}\mathbf{1}_{r\geq r_0}$ in the uniform metric and ${H}(r)$ is a continuous function in $D_\cH[0,1]$, we have 
		$${\tilde F}_n(r)\rightsquigarrow B_Q(r)+{H}(r).$$
		Define $H_y(r)=\frac{1}{\sqrt{\li Q y,y \ri_\cH}}\li y,{H}(r)\ri_\cH$, $C(r)=B(r)+H_y(r)-B(b)$, $C'(r)=B(r)+H_y(r)$, $D(r)=B(r-b)+H_y(r)$, $a_s=\frac{s}{(1-2b)r}$ and $b_s = \frac{1-2b-s}{(1-2b)(1-r)}$. From this we can use the same argument and derive the conditional distribution as follows. 
		\footnotesize
		\begin{align}
			&G^\ast\big|_{Y^\ast=y} \nonumber\\
			\stackrel{d}{=} &\sup_{r \in [b,1{-}b]}  \frac{\sqrt{1{-}2b}\Big[C(r){-} \frac{r{-}b}{1{-}2b}(C(1-b))  \Big]}  {\Big\{\int_{b}^{r} \big[C(s){-}\frac{s{-}b}{r{-}b}(C(r))\big]^2 ds {+} \int_{r}^{1{-}b}\big[C'(1-b){-}C'(s){-}\frac{1{-}b{-}s}{1{-}b{-}r}(C'(1-b){-}C'(r))    \big]^2 ds  \Big\}^{1/2}  }  \nonumber \\
			\stackrel{d}{=} &\sup_{r \in [b,1{-}b]}  \frac{\sqrt{1{-}2b}\Big[C(r){-} \frac{r{-}b}{1{-}2b}C(1-b) \Big]}  {\Big\{\int_{0}^{r{-}b} \big[C(s+b){-}\frac{s}{r{-}b}C(r)\big]^2 ds {+} \int_{r{-}b}^{1{-}2b}\big[C'(1-b){-}C'(s+b){-}\frac{1{-}2b{-}s}{1{-}b{-}r}(C'(1-b){-}C'(r))    \big]^2 ds   \Big\}^{1/2}    }  \nonumber \\
			\stackrel{d}{=} &\sup_{r \in [b,1{-}b]}  \frac{\sqrt{1{-}2b}\Big[D(r){-} \frac{r{-}b}{1{-}2b}D(1-b) \Big]}  {\Big\{\int_{0}^{r{-}b} \big[D(s+b){-}\frac{s}{r{-}b}D(r)\big]^2 ds {+} \int_{r{-}b}^{1{-}2b}\big[D(1-b){-}D(s+b){-}\frac{1{-}2b{-}s}{1{-}b{-}r}(D(1-b){-}D(r))    \big]^2 ds   \Big\}^{1/2}    }  \nonumber \\
			\stackrel{d}{=} &\sup_{r \in [b,1{-}b]}  \frac{\sqrt{1{-}2b}\Big[D((1{-}2b)r{+}b){-} rD(1{-}b) \Big]}  {\Big\{\int\limits_{0}^{(1{-}2b)r} \big[D(s{+}b){-}a_sD((1{-}2b)r{+}b)\big]^2 ds {+} \int\limits_{(1{-}2b)r}^{1{-}2b}\big[ D(1{-}b){-}D(s{+}b){-}b_s(D(1{-}b){-}D((1{-}2b)r{+}b))    \big]^2 ds\Big\}^{1/2}  }  \nonumber\\		
			\stackrel{d}{=} &\sup_{r \in [b,1{-}b]}  \frac{D((1{-}2b)r{+}b){-} rD(1{-}b) }  {\Big\{ \int\limits_{0}^{r} \big[D((1{-}2b)s{+}b){-}\frac{s}{r}D((1{-}2b)r{+}b)\big]^2 ds {+} \int\limits_{r}^{1}\big[D(1{-}b){-}D((1{-}2b)s{+}b){-}\frac{1{-}s}{1{-}r}(D(1{-}b){-}D((1{-}2b)r{+}b))    \big]^2 ds \Big\}^{1/2} }  \nonumber\\		
			\stackrel{d}{=} &\sup_{r \in [0,1]}  \frac{B'(r)- rB'(1) }  {\Big\{\int_{0}^{r} \big[B'(s)-\frac{s}{r}B'(r)\big]^2 ds + \int_{r}^{1}\big[B'(1)-B'(s)-\frac{1-s}{1-r}(B'(1)-B'(r)) \big]^2 ds\Big\}^{1/2} },\nonumber
		\end{align}
		\normalsize
		So part (ii).2 and (ii).3 are proved.
	\end{proof}

	\section{Proof of Theorem \ref{cp_est}} \label{appen_proof_est}
	
	\begin{proof}
		Due to the symmetric nature of $\frac{T_n(k)}{\sqrt{V_n(k)}}$, we only prove $P(\hat k > k^\ast-\epsilon_n)\to 1$. Define $A=\{k> \fb+1|k\leq k^\ast-\epsilon_n,\,k\in \mathcal{Z}\}$, we divide our proof into two parts:
		
		(a) $\max\limits_{k \in A} |T_n(k)|=O_p(n\|\Delta_n\|_\cH^2).$
		
		(b) $\max\limits_{k \in A} V_n(k)^{-1}=o_p(n^{-1}\|\Delta_n\|_\cH^{-2}).$
		
		If part (a) and (b) holds, then we have $\max\limits_{k \in A}\frac{T_n(k)}{\sqrt{V_n(k)}} = o_p(\sqrt{n}\|\Delta_n\|_\cH)$ and in the proof of part(ii).1 of Theorem \ref{theorem_cp}, we have shown that $\sqrt{n}\|\Delta_n\|_\cH(\frac{T_n(k^\ast)}{\sqrt{V_n(k^\ast)}})^{-1} = O_p(1)$. So $\max\limits_{k \in A}\frac{T_n(k)}{\sqrt{V_n(k)}}/(\frac{T_n(k^\ast)}{\sqrt{V_n(k^\ast)}}) =o_p(1)$ and the theorem is proved.
		
		\textbf{\textit{Proof of part (a)}}: As in the proof of part(ii).1 of Theorem \ref{theorem_cp}, denote $Y_t' = Y_t-\mu_t$, $\bar Y'_1 =\frac{1}{\sqrt{n}} \sum_{t=1}^{\fb}Y'_t$, $\bar Y'_3 = \frac{1}{\sqrt{n}} \sum_{t=n-\fb+1}^{n} Y'_t$ and $Z'_t=\li\bar Y'_1-\bar Y'_3,Y'_t\ri_\cH$. Let $S_{a,b}^A=\sum_{t=a}^{b}A_t$ for $A=Y',Z',\mu$, then we have for $k\in A$
		\footnotesize
		\begin{align}
			T_n( k) 
			=& \frac{1}{\sqrt{n-2\fb}}\sum_{t=\fb+1}^{k}\Big(Z_{t}-\frac {S_{\fb+1,n-\fb}}{n-2\fb}\Big) \nonumber \\
			=&\frac{S^{Z'}_{\fb{+}1,k}}{\sqrt{n-2\fb}}{-}\frac{\fb}{\sqrt{n(n-2\fb)}}\li \Delta_n, S^{Y'}_{\fb{+}1,k}\ri_\cH
			{-}\frac{\fb}{\sqrt{n(n-2\fb)}}\li \Delta_n, (k{-}\fb)\mu_1\ri_\cH
			{+}\frac{\li\bar Y'_1{-}\bar Y'_3, (k{-}\fb)\mu_1\ri}{\sqrt{n-2\fb}} \nonumber \\
			&{-}\frac{k-\fb}{\sqrt{(n-2\fb)^3}}\Big\{ S^{Z'}_{\fb{+}1,n{-}\fb} {-}     \frac{\fb}{\sqrt{n}}\li \Delta_n, S^{Y'}_{\fb{+}1,n{-}\fb}\ri_\cH    \Big\}  \nonumber \\
			&         {+}\frac{k-\fb}{\sqrt{(n-2\fb)^3}}  \Big\{
			\frac{\fb}{\sqrt{n}}\li \Delta_n, (n{-}2\fb)\mu_1{+}(n{-}\fb{-}k^\ast)\Delta_n\ri_\cH      
			{-}     \li\bar Y'_1{-}\bar Y'_3, (n{-}2\fb)\mu_1{+}(n{-}\fb{-}k^\ast)\Delta_n\ri_\cH        \Big\}\nonumber \\
			=& \frac{S^{Z'}_{\fb{+}1,k}}{\sqrt{n-2\fb}}{-}\frac{\fb}{\sqrt{n(n-2\fb)}}\li \Delta_n, S^{Y'}_{\fb{+}1,k}\ri_\cH    {-}\frac{k-\fb}{\sqrt{(n-2\fb)^3}}\Big\{ S^{Z'}_{\fb{+}1,n{-}\fb} {-}     \frac{\fb}{\sqrt{n}}\li \Delta_n, S^{Y'}_{\fb{+}1,n{-}\fb}\ri_\cH    \Big\}  \nonumber \\
			&{+}\frac{k-\fb}{\sqrt{(n-2\fb)^3}}  \Big\{
			\frac{\fb}{\sqrt{n}}\li \Delta_n,(n{-}\fb{-}k^\ast)\Delta_n\ri_\cH      
			{-}     \li\bar Y'_1{-}\bar Y'_3,(n{-}\fb{-}k^\ast)\Delta_n\ri_\cH        \Big\}.\nonumber 
		\end{align}
		\normalsize
		Note that $\max\limits_{k \in A} | \frac{S^{Z'}_{\fb{+}1,k}}{\sqrt{n-2\fb}} | = O_p(1)$, $\max\limits_{k \in A} \frac{k-\fb}{\sqrt{(n-2\fb)^3}}  |  S^{Z'}_{\fb{+}1,n{-}\fb} | = O_p(1)$, and by the Cauchy Schwarz inequality we have $$\max\limits_{k \in A} \frac{k-\fb}{\sqrt{(n-2\fb)^3}} | \frac{\fb}{\sqrt{n}}\li \Delta_n, S^{Y'}_{\fb{+}1,n{-}\fb}\ri_\cH | = O_p(\sqrt{n}\|\Delta_n\|_\cH)$$ and $$\max\limits_{k \in A}\frac{k-\fb}{\sqrt{(n-2\fb)^3}}| \li\bar Y'_1{-}\bar Y'_3,(n{-}\fb{-}k^\ast)\Delta_n\ri_\cH      | =   O_p(\sqrt{n}\|\Delta_n\|_\cH).$$ Then 
		$$\max\limits_{k \in A}|T_n( k) | = \max\limits_{k \in A} \frac{k-\fb}{\sqrt{(n-2\fb)^3}} 	\frac{\fb(n{-}\fb{-}k^\ast)}{\sqrt{n}}\|\Delta_n\|_\cH^2+O_p(\sqrt{n}\|\Delta_n\|_\cH) = O_p(n\|\Delta_n\|_\cH^2)    $$
		and part (a) is proved.
		
		\textbf{\textit{Proof of part (b)}}: Note that for $k\in A$, 
		\begin{align}
			V_n(k)\geq & \frac{1}{(n{-}2\fb)^{2}}\sum_{t=k^\ast+1}^{n-\fb}\Big\{S_{t,n-\fb}{-}\frac{n{-}\fb{-}t{+}1}{n{-}\fb{-}k}S_{k+1,n-\fb}\Big\}^2 \nonumber \\
			=&  \frac{1}{(n{-}2\fb)^{2}}\sum_{t=k^\ast+1}^{n-\fb}\Bigg\{\big\li \bar Y'_1{-}\bar Y'_3{-}\frac{\fb}{\sqrt{n}}\Delta_n, S^{Y'}_{t,n-\fb} {+}(n{-}\fb{-}t{+}1)\Delta_n{+}(n{-}\fb{-}t{+}1)\mu_1      \big\ri_\cH \nonumber \\
			&\qquad\quad
			{-}\frac{n{-}\fb{-}t{+}1}{n{-}\fb{-}k}      \big\li \bar Y'_1{-}\bar Y'_3{-}\frac{\fb}{\sqrt{n}}\Delta_n, S^{Y'}_{k+1,n-\fb} {+}(n{-}\fb{-}k^\ast)\Delta_n{+}(n{-}\fb{-}k)\mu_1      \big\ri_\cH\Bigg\}^2 \nonumber \\
			=&  \frac{1}{(n{-}2\fb)^{2}}\sum_{t=k^\ast+1}^{n-\fb}\Bigg\{     \big\li \bar Y'_1{-}\bar Y'_3{-}\frac{\fb}{\sqrt{n}}\Delta_n, S^{Y'}_{t,n-\fb} {-}\frac{n{-}\fb{-}t{+}1}{n{-}\fb{-}k}S^{Y'}_{k+1,n-\fb}    \big\ri_\cH   \nonumber\\
			&\qquad\qquad\qquad\qquad\qquad\qquad\qquad\qquad{+}  \frac{(n{-}\fb{-}t{+}1)(k^\ast{-}k)}{n{-}\fb{-}k}    \big\li \bar Y'_1{-}\bar Y'_3{-}\frac{\fb}{\sqrt{n}}\Delta_n,\Delta_n    \big\ri_\cH        
			\Bigg\}^2.\nonumber
		\end{align}
		By the continuous mapping theorem and Cauchy Schwarz inequality,
		\footnotesize
		\begin{align}
			\max\limits_{k \in A} \frac{1}{(n{-}2\fb)^{2}}\sum_{t=k^\ast+1}^{n-\fb}  \big\li \bar Y'_1{-}\bar Y'_3, S^{Y'}_{t,n-\fb} {-}\frac{n{-}\fb{-}t{+}1}{n{-}\fb{-}k}S^{Y'}_{k+1,n-\fb}    \big\ri_\cH^2   &= O_P(1) \nonumber \\
			\max\limits_{k \in A} \frac{1}{(n{-}2\fb)^{2}}\sum_{t=k^\ast+1}^{n-\fb}  \big\li {-}\frac{\fb}{\sqrt{n}}\Delta_n, S^{Y'}_{t,n-\fb} {-}\frac{n{-}\fb{-}t{+}1}{n{-}\fb{-}k}S^{Y'}_{k+1,n-\fb}    \big\ri_\cH^2   &= O_P(n\|\Delta_n\|_\cH^2) \nonumber \\
			\max\limits_{k \in A} \frac{1}{(n{-}2\fb)^{2}}\sum_{t=k^\ast+1}^{n-\fb} \frac{(n{-}\fb{-}t{+}1)^2(k^\ast{-}k)^2}{(n{-}\fb{-}k)^2}    \big\li \bar Y'_1{-}\bar Y'_3,\Delta_n    \big\ri_\cH ^2   \leq \frac{(k^\ast)^2(n{-}\fb{-}k^\ast)}{(n{-}2\fb)^{2}}\big\li \bar Y'_1{-}\bar Y'_3,\Delta_n    \big\ri_\cH ^2&= O_P(n\|\Delta_n\|_\cH^2), \nonumber 
		\end{align}
		\normalsize
		and 
		\begin{align}
			&\min\limits_{k \in A} \frac{1}{(n{-}2\fb)^{2}}\sum_{t=k^\ast+1}^{n-\fb} \frac{(n{-}\fb{-}t{+}1)^2(k^\ast{-}k)^2}{(n{-}\fb{-}k)^2}    \big\li {-}\frac{\fb}{\sqrt{n}}\Delta_n,\Delta_n    \big\ri_\cH ^2  \nonumber \\
			\geq&  \frac{\fb^2\epsilon_n^2\|\Delta_n\|_\cH^4}{n(n{-}2\fb)^{2}(n{-}\fb)^{2}}\sum_{t=k^\ast+1}^{n-\fb}(n{-}\fb{-}t{+}1)^2 \nonumber \\
			\geq &C \epsilon_n^2\|\Delta_n\|_\cH^4 = Cn^{-1}(\epsilon_n\|{\Delta}_n\|_\cH)^{2}n\|\Delta_n\|_\cH^2. \nonumber 
		\end{align}
		By the assumption, $n^{-1}(\epsilon_n\|{\Delta}_n\|_\cH)^{2}\to \infty$ as $n\to \infty$, so $\min\limits_{k \in A} 	V_n(k)$ diverges to $\infty$ faster than $n\|\Delta_n\|_\cH^2$ and part (b) is proved.

	\end{proof}

	\section{Proof of Theorem \ref{cp_est_wbs}}\label{app_th_wbs}
	First, we introduce some notations. For sequences of random variables $\{x_n\}$, real numbers $\{a_n\}$ and $\{b_n\}$, denote $X_n=O_p^s(a_n)$ if $x_n/a_n = O_p(1)$ and $a_n/x_n = O_p(1)$, $a_n\sim b_n$ if $\lim_{n\to \infty}|a_n/b_n|\to C>0$ and denote $a_n=O^s(b_n)$ if $a_n=O(b_n)$ and $b_n=O(a_n)$. Also for simplicity, we assume $\{nr^\ast_i\}$ and $nb$ are all integers for $i=1,2,\dots,m_0$.

	After initializing Algorithm 1 with WBS-SN($1,n,K_n, L_0, M,b$), for each $m\in \{1,2,\dots M\}$, if the interval
	$[\lfloor n\underline U_m\rfloor,\lfloor n\bar U_m\rfloor]$ contains one change point in the middle, in the sense that there exist $i\in \{1,2,\dots m_0\}$ such that $r^\ast_{i-1}<\underline U_m<r^\ast_i<\bar U_m<r^\ast_{i+1} \text{ and }\frac{r^\ast_i{-}\underline U_m}{\bar U_m{-}\underline U_m}\in(b,1-b)$, according to the proof for part (ii).1 of Theorem \ref{theorem_cp}, the SS-SN statistic $G_n$ calculated on $[\lfloor n\underline U_m\rfloor,\lfloor n\bar U_m\rfloor]$ diverges to infinity at the rate $n^{1/2-\delta_i}>K_n$. 
	
	If $[\lfloor n\underline U_m\rfloor,\lfloor n\bar U_m\rfloor]$ falls into one of the four settings as described in Sections \ref{sec_0cp}-\ref{sec_3cp}, no change point can be selected inside $[\lfloor n\underline U_m\rfloor,\lfloor n\bar U_m\rfloor]$ because the SS-SN statistic calculated on this interval is either bounded or diverges to infinity at a slower rate (see Lemmas \ref{lemma_0cp}-\ref{lemma_3cp}). So the first change point estimator will be from one of the small intervals which contains one change point in the middle. Let this estimator be $\hat k_j$ for the true change point $k_j^\ast$, by Theorem \ref{cp_est} we have $P(|\hat k_j -k^\ast_j|<\epsilon_{jn})\to 1$, so the estimation effect will be negligible, in the sense that for any $m=1,2,\dots,M$, $\lfloor n\underline U_m\rfloor$ and $\lfloor n\bar U_m\rfloor]$ does not fall into the interval between $\hat k_j$ and $k^\ast_j$ for large enough $n$. The change point estimators from WBS-SN($1,\hat k_j,K_n, L_0, M,b$) and WBS-SN($\hat k_j{+}1,n,K_n, L_0, M,b$) will also comes from small intervals which contain one change point in the middle, and the algorithm will continue.
	
	For $i=1,2,\dots,m_0+1$, let $\M_i$ denote set of those indices $m\in\{1,2,\dots,M\}$ for which $(\underline U_m,\bar U_m)\subset (r^\ast_{i-1},r^\ast_i)$. Since the number of small intervals $M$ is fixed, by continuous mapping theorem, $\max\limits_{m\in \mathcal{M}_{i}, k\in \{\lfloor n\underline U_m\rfloor{+}b_m,\dots,\lfloor n\bar U_m\rfloor{-}b_m{-}1\}}\frac{T_{nm}(k)}{\sqrt{V_{nm}(k)}} = O_p(1)$, so the algorithm will stop after estimators for all $m_0$ change points are selected, and the proof is finished.

	\subsection{No middle change point setting}\label{sec_0cp}

	For nonnegative integers $q,p$ such that $0\leq p+q\leq m_0$ and $0<r^\ast_1<\cdots<r^\ast_p<b<1-b<r^\ast_{p+1}<\cdots<r^\ast_{p+q}<1$, let $\A_{1n}=\{\fb+1\leq k\leq n-\fb,\,k\in \mathcal{Z}\}$, we show that no change point can be selected inside $\A_{1n}$.
	
	\begin{lemma}\label{lemma_0cp}
		Suppose Assumption \ref{assump_wbs} holds, then
		\begin{align}
			\max_{k\in\A_{1n}}\frac{T_n(k)}{\sqrt{V_n(k)}} = O_p(1)
		\end{align}
	\end{lemma}
	
	\begin{proof}
		As defined in Equation (\ref{eq_proj}),
		\begin{align}
			\bar Y_1 =& \frac{1}{\sqrt{n}} \sum_{t=1}^{nb} Y'_t +\frac{nb}{\sqrt{n}}\sum_{i=1}^{p+1}a_i\mu_{k^\ast_i} \nonumber \\
			= & \frac{1}{\sqrt{n}} \sum_{t=1}^{nb} Y'_t+\frac{nb}{\sqrt{n}}\big\{\mu_{k^\ast_1}{+}a_2\Delta^\ast_{1}{+}a_3(\Delta^\ast_1+\Delta^\ast_2)+\dots+a_{p+1}(\Delta^\ast_1+\dots+\Delta^\ast_p)\big\} \nonumber \\
			=& \frac{1}{\sqrt{n}} \sum_{t=1}^{nb} Y'_t+\frac{nb}{\sqrt{n}}\big\{\mu_{k^\ast_1}{+}(1{-}a_1)\Delta^\ast_1  {+}  (1{-}a_1{-}a_2) \Delta^\ast_2{+}\cdots{+}(1{-}a_1{-}\cdots{-}a_p)\Delta^\ast_p          \big\},
		\end{align}
		where $a_1 = \frac{r^\ast_1}{b}, a_{p+1} = \frac{b-r^\ast_p}{b}$ and $a_i = \frac{r^\ast_i-r^\ast_{i-1}}{b}$ for $i=2,3,\dots,p$. Similarly, let $b_1 = \frac{r^\ast_{p+1}-1+b}{1-b},b_{q+1} = \frac{1-r^\ast_{p+q}}{1-b}$ and $b_i = \frac{r^\ast_{p+i}-r^\ast_{p+i-1}}{1-b}$ for $i=2,3,\dots,q$, we have 
		\begin{align}
			\bar Y_3 =  \frac{1}{\sqrt{n}} \sum_{t=n-nb+1}^{n} Y'_t+\frac{nb}{\sqrt{n}}\big\{\mu_{k^\ast_{p+1}}{+}(1{-}b_1)\Delta^\ast_{p+1}  {+}  (1{-}b_1{-}b_2) \Delta^\ast_{p+2}{+}\cdots{+}(1{-}b_1{-}\cdots{-}b_q)\Delta^\ast_{p+q}          \big\}.\nonumber
		\end{align}
		Let $\bar Y'_1 =  \frac{1}{\sqrt{n}}S_{1,nb}^{Y'}, \bar Y'_3 =  \frac{1}{\sqrt{n}}S_{n-nb+1,n}^{Y'}$ and
		\begin{align}
			\cP_n = &a_1\Delta^\ast_1{+}(a_1{+}a_2)\Delta^\ast_2{+}\cdots{+}({a_1{+}\cdots{+}a_p})\Delta^\ast_p\nonumber \\
			&{+}(1{-}b_1)\Delta^\ast_{p+1}  {+}  (1{-}b_1{-}b_2) \Delta^\ast_{p+2}{+}\cdots{+}(1{-}b_1{-}\cdots{-}b_q)\Delta^\ast_{p+q},   \label{def_pn_0cp}
		\end{align}		
		then $\bar Y_1-\bar Y_3 = \bar Y_1'-\bar Y_3'-b\sqrt{n}\cP_n$ and for $k\in\A_{1n}$, we have 
		\begin{align}
			T_n(k) = &\big\li \bar Y_1'{-}\bar Y_3'{-}b\sqrt{n}\cP_n, \frac{1}{\sqrt{n-2nb}}\big[S_{nb+1,k}^{Y'}-\frac{k{-}nb}{n{-}2nb}S_{nb+1,n-nb}^{Y'}\big]\big\ri_\cH, \label{eq_0cp_1} \\
			V_n(k)\geq& \frac{1}{(n{-}2nb)^{2}}\sum_{t=nb+1}^{k}\Big\{S_{nb+1,t}{-}\frac{t{-}nb}{k{-}nb}S_{nb+1,k}\Big\}^2 \nonumber \\
			= & \frac{1}{(n{-}2nb)^{2}}\sum_{t=nb{+}1}^{k}\big\li \bar Y_1'{-}\bar Y_3'{-}b\sqrt{n}\cP_n,S^{Y'}_{nb+1,t} {-}\frac{t{-}nb}{k{-}nb}S^{Y'}_{nb+1,k}    \big\ri_\cH^2 \nonumber \\
			=& L_n(k).
		\end{align}
		If $\cP_n= 0$, it is easy to see $\max_{k\in\A_{1n}}\frac{T_n(k)}{\sqrt{V_n(k)}} = O_p(1)$. If $\cP_n\neq 0$, we can write $\cP_n = \kappa_0 f_0'{+}f^n_0$ for some $f'_{0},f^n_{0}\in \cH$ such that $\kappa_{0} \sim n^{-\zeta_{0}'}$, $\zeta_{0}'\in [0,1/2)$ and $\|f^n_{0}\|_\cH =o(\kappa_0)$. since $\min_{k\in\A_{1n}}L_n(k) = O_p^s(n^{1-2\zeta'_{0}})$ and $\max_{k\in\A_{1n}}|T_n(k) | = O_p(n^{1/2-\zeta'_{0}})$, we have $\max_{k\in\A_{1n}} V_n(k)^{-1} = O_p(n^{-1+2\zeta_{0}})$ and  $\max_{k\in\A_{1n}}\frac{T_n(k)}{\sqrt{V_n(k)}} = O_p(1)$.
		
	\end{proof}

	\subsection{One change point setting}\label{sec_1cp}
	
	For nonnegative integers $q,p$ such that $1\leq p+q\leq m_0-1$ and $0<r^\ast_1<\cdots<r^\ast_p<b<r^\ast_{p+1}<1-b<r^\ast_{p+2}<\cdots<r^\ast_{p+q+1}<1$, let $\A_{1n}=\{k\geq \fb+1|k\leq k^\ast_{p+1}-\epsilon_{(p+1)n},\,k\in \mathcal{Z}\}$ and $\A_{2n}=\{k\leq n-\fb|k\geq k^\ast_{p+1}+\epsilon_{(p+1)n},\,k\in \mathcal{Z}\}$, we show that no change point can be selected inside $\A_{1n}\cup\A_{2n}$.
	\begin{lemma}\label{lemma_1cp}
		Suppose Assumption \ref{assump_wbs} holds, then
		\begin{align}
			\max_{k\in\A_{1n}\cup\A_{2n}}\frac{T_n(k)}{\sqrt{V_n(k)}} = o_p(n^{1/2-\delta_{p+1}})
		\end{align}
	\end{lemma}
	
	\begin{proof}
		we only prove $	\max_{k\in\A_{1n}}\frac{T_n(k)}{\sqrt{V_n(k)}} = o_p(n^{1/2-\delta_{p+1}})$ since the proof for $	\max_{k\in\A_{2n}}\frac{T_n(k)}{\sqrt{V_n(k)}} = o_p(n^{1/2-\delta_{p+1}})$ is analogous. Note that in this setting,
		\begin{align}
			\cP_n = &a_1\Delta^\ast_1{+}(a_1{+}a_2)\Delta^\ast_2{+}\cdots{+}({a_1{+}\cdots{+}a_p})\Delta^\ast_p{+}\Delta^\ast_{p+1}\nonumber \\
			&{+}(1{-}b_1)\Delta^\ast_{p+2}  {+}  (1{-}b_1{-}b_2) \Delta^\ast_{p+3}{+}\cdots{+}(1{-}b_1{-}\cdots{-}b_q)\Delta^\ast_{p+q+1},   \label{def_pn_1cp}
		\end{align}		
		where $\{a_i\}$ is defined in Section \ref{sec_0cp}, $b_1 = \frac{r^\ast_{p+2}-1+b}{1-b},b_{q+1} = \frac{1-r^\ast_{p+q+1}}{1-b}$ and $b_i = \frac{r^\ast_{p+i+1}-r^\ast_{p+i}}{1-b}$ for $i=2,3,\dots,q$. For $k\in\A_{1n}$, we have 
		\begin{align}
			T_n(k) = & {-}\frac{(k{-}nb)(n{-}nb{-}k^\ast_{p+1})}{(n-2nb)^{3/2}}\big\li \bar Y_1'{-}\bar Y_3'{-}b\sqrt{n}\cP_n, \Delta^\ast_{p+1}\big\ri_\cH \nonumber \\
			& {+}\big\li \bar Y_1'{-}\bar Y_3'{-}b\sqrt{n}\cP_n, \frac{1}{\sqrt{n-2nb}}\big[S_{nb+1,k}^{Y'}-\frac{k{-}nb}{n{-}2nb}S_{nb+1,n-nb}^{Y'}\big]\big\ri_\cH \label{eq_1cp_1}
		\end{align}
		\begin{align}
			V_n(k) =& \frac{1}{(n{-}2nb)^{2}}\sum_{t=k+1}^{n-nb}\Big\{S_{t,n-nb}{-}\frac{n{-}nb{-}t{+}1}{n{-}nb{-}k}S_{k+1,n-nb}\Big\}^2{+}\frac{1}{(n{-}2nb)^{2}}\sum_{t=nb+1}^{k}\Big\{S_{nb+1,t}{-}\frac{t{-}nb}{k{-}nb}S_{nb+1,k}\Big\}^2 \nonumber \\
			=&R_n(k)+L_n(k) .
		\end{align}		
		Let $R_{1n}(k)=\frac{1}{(n{-}2nb)^{2}}\sum_{t=k^\ast_{p+1}{+}1}^{n{-}nb}\Big\{S_{t,n-nb}{-}\frac{n{-}nb{-}t{+}1}{n{-}nb{-}k}S_{k+1,n-nb}\Big\}^2$, then
		\begin{align}
			R_{1n}(k) = & \frac{1}{(n{-}2nb)^{2}}\sum_{t=k^\ast_{p+1}{+}1}^{n{-}nb}\Big\{\big\li \bar Y_1'{-}\bar Y_3'{-}b\sqrt{n}\cP_n,S^{Y'}_{t,n-nb} {-}\frac{n{-}nb{-}t{+}1}{n{-}nb{-}k}S^{Y'}_{k+1,n-nb}    \big\ri_\cH  \nonumber \\
			&\qquad\qquad\qquad\qquad{+}\frac{(n{-}nb{-}t{+}1)(k^\ast_{p+1}{-}k)}{n{-}nb{-}k}    \big\li \bar Y'_1{-}\bar Y'_3{-}b\sqrt{n}\cP_n,\Delta^\ast_{p+1}   \big\ri_\cH    \Big\}^2 \nonumber \\
			=& \frac{1}{(n{-}2nb)^{2}}\sum_{t=k^\ast_{p+1}{+}1}^{n{-}nb}\Big\{R_{11}(k,t){+}R_{12}(k,t) \Big\}^2,
		\end{align}				
		\begin{align}
			L_n(k) = & \frac{1}{(n{-}2nb)^{2}}\sum_{t=nb{+}1}^{k}\big\li \bar Y_1'{-}\bar Y_3'{-}b\sqrt{n}\cP_n,S^{Y'}_{nb+1,t} {-}\frac{t{-}nb}{k{-}nb}S^{Y'}_{nb+1,k}    \big\ri_\cH^2 \nonumber 
		\end{align}			
		and $R_n(k) \geq R_{1n}(k)$. According to the definition of $\cP_n$ in Equation (\ref{def_pn_1cp}) and Assumption \ref{assump_wbs}, we know $\cP_n = \theta_{p+1}f_{p+1}+\kappa_{p+1} f'_{p+1}{+}f^n_{p+1}$ for some $f'_{p+1},f^n_{p+1}\in \{f_{p+1}\}^\perp$ such that $\|f^n_{p+1}\|_\cH =o(\kappa_{p+1})$. For $\theta_{p+1}$, it is either $\theta_{p+1}=0$ or $\theta_{p+1} \sim n^{-\zeta_{p+1}}$ for some $\zeta_{p+1}\in [0,1/2)$. Similarly $\kappa_{p+1}=0$ (in this case we have $\|f^n_{p+1}\|_\cH =0$) or $\kappa_{p+1} \sim n^{-\zeta_{p+1}'}$ for some $\zeta_{p+1}'\in [0,1/2)$. We consider six different scenarios.
		
		\noindent
		\textbf{(1) $\theta_{p+1} = \kappa_{p+1}=0$.}
		
		From Equation (\ref{eq_1cp_1}), $\max_{k\in\A_{1n}}|T_n(k)| = O_p(n^{1/2-\delta_{p+1}})$. Since 
		\small
		\begin{align}
			& \max_{k\in\A_{1n}} \frac{1}{(n{-}2nb)^{2}}\sum_{t=k^\ast_{p+1}{+}1}^{n{-}nb}R_{11}(k,t)^2  \\
            = & \max_{k\in\A_{1n}} \frac{1}{(n{-}2nb)^{2}}\sum_{t=k^\ast_{p+1}{+}1}^{n{-}nb}\big\li \bar Y_1'{-}\bar Y_3',S^{Y'}_{t,n-nb} {-}\frac{n{-}nb{-}t{+}1}{n{-}nb{-}k}S^{Y'}_{k+1,n-nb}    \big\ri_\cH ^2 = O_p(1) \nonumber
		\end{align}\normalsize
		and
		\begin{align}
			\min_{k\in\A_{1n}}	\frac{1}{(n{-}2nb)^{2}}\sum_{t=k^\ast_{p+1}{+}1}^{n{-}nb}R_{12}(k,t)^2  =& \min_{k\in\A_{1n}}	\frac{n^{-2\delta_{p+1}}(k^\ast_{p+1}{-}k)^2}{(n{-}2nb)^{2}(n{-}nb{-}k)^2}\li\bar Y'_1{-}\bar Y'_3,f_{p+1} \ri_\cH^2\sum_{t=k^\ast_{p+1}{+}1}^{n{-}nb}(n{-}nb{-}t{+}1)^2 \nonumber \\
			\geq & Cn^{-1}(\epsilon_{(p+1)n}n^{-\delta_{p+1}})^2\li\bar Y'_1{-}\bar Y'_3,f_{p+1} \ri_\cH^2 \stackrel{p}{\to} \infty,
		\end{align}
		we have $\max_{k\in\A_{1n}}V_n(k)^{-1} = o_p(1)$, which implies $\max_{k\in\A_{1n}}\frac{T_n(k)}{\sqrt{V_n(k)}} = o_p(n^{1/2-\delta_{p+1}})$.
		
		\noindent
		\textbf{(2) $\theta_{p+1} =0$, $ \kappa_{p+1}\sim n^{-\zeta_{p+1}'}$ and $\zeta_{p+1}'<\delta_{p+1}$.}
		
		From Equation (\ref{eq_1cp_1}), $\max_{k\in\A_{1n}}|T_n(k)| = O_p(n^{1/2-\zeta_{p+1}'})$. Since 
		\small
		\begin{align}
			& \min_{k\in\A_{1n}} \frac{1}{(n{-}2nb)^{2}}\sum_{t=k^\ast_{p+1}{+}1}^{n{-}nb}R_{11}(k,t)^2 \\ = &\min_{k\in\A_{1n}} \frac{1}{(n{-}2nb)^{2}}\sum_{t=k^\ast_{p+1}{+}1}^{n{-}nb}\big\li \bar Y_1'{-}\bar Y_3'{-}b\sqrt{n}\cP_n,S^{Y'}_{t,n-nb} {-}\frac{n{-}nb{-}t{+}1}{n{-}nb{-}k}S^{Y'}_{k+1,n-nb}    \big\ri_\cH ^2 \nonumber \\
			= & O_p^s(n^{1-2\zeta_{p+1}'})
		\end{align}\normalsize
		and 
		\begin{align}
			\max_{k\in\A_{1n}}	\frac{1}{(n{-}2nb)^{2}}\sum_{t=k^\ast_{p+1}{+}1}^{n{-}nb}R_{12}(k,t)^2  \leq& \max_{k\in\A_{1n}}	\frac{Cn^{2-2\delta_{p+1}}}{(n{-}2nb)^{2}(n{-}nb{-}k)^2}\li\bar Y'_1{-}\bar Y'_3,f_{p+1} \ri_\cH^2\sum_{t=k^\ast_{p+1}{+}1}^{n{-}nb}(n{-}nb{-}t{+}1)^2 \nonumber \\
			\leq & Cn^{1-2\delta_{p+1}}\li\bar Y'_1{-}\bar Y'_3,f_{p+1} \ri_\cH^2 ,
		\end{align}
		we have $\max_{k\in\A_{1n}}V_n(k)^{-1} = O_p(n^{-1+2\zeta_{p+1}'})$, which implies $\max_{k\in\A_{1n}}\frac{T_n(k)}{\sqrt{V_n(k)}} = O_p(1)$.

		\noindent
		\textbf{(3) $\theta_{p+1} =0$, $ \kappa_{p+1}\sim n^{-\zeta_{p+1}'}$ and $\zeta_{p+1}'>\delta_{p+1}$.}
		
		In this case $\max_{k\in\A_{1n}}|T_n(k)| = O_p(n^{1/2-\delta_{p+1}})$. Let $\A_{1n}^1 = \A_{1n}\cap \{k\leq n\frac{b+r^\ast_{p+1}}{2}\}$ and $\A_{1n}^2 = \A_{1n}\cap \{k> n\frac{b+r^\ast_{p+1}}{2}\}$, then we have 
		\small
		\begin{align}
			\min_{k\in\A_{1n}^1}R_n(k) \geq  &\min_{k\in\A_{1n}^1} \frac{1}{(n{-}2nb)^{2}}\sum_{t=k^\ast_{p+1}{+}1}^{n{-}nb}  R_{12}(k,t)^2 = O_p^s(n^{1-2\delta_{p+1}})
		\end{align}\normalsize
		and $\min_{k\in\A_{1n}^2}V_n(k) \geq \min_{k\in\A_{1n}^2}L_n(k)  =O_p^s(n^{1-2\zeta_{p+1}'})$, which implies $\max_{k\in\A_{1n}}V_n(k)^{-1} = o_p(1)$ and $\max_{k\in\A_{1n}}\frac{T_n(k)}{\sqrt{V_n(k)}} =  o_p(n^{1/2-\delta_{p+1}})$.
		
		\noindent
		\textbf{(4) $\theta_{p+1} =0$, $ \kappa_{p+1}\sim n^{-\zeta_{p+1}'}$ and $\zeta_{p+1}'=\delta_{p+1}$.}
		
		In this case $\max_{k\in\A_{1n}}|T_n(k)| = O_p(n^{1/2-\delta_{p+1}})$ and
		\small\begin{align}
			&\min_{nb+1\leq k\leq k^\ast_{p+1}}\frac{n^{-1+2\delta_{p+1}}}{(n{-}2nb)^{2}}\sum_{t=k^\ast_{p+1}+1}^{n-nb}\big\{R_{11}(k,t){+}R_{12}(k,t)\big\}^2\nonumber \\
			\stackrel{\D}{\to}&\frac{1}{(1{-}2b)^2}\inf_{r\in[b,r^\ast_{p+1}]}\int_{r^\ast_{p+1}}^{1-b}\Big\{{-} \big\li Cbf_{p+1}', B_Q(1{-}b){-}B_Q(s){-} \frac{1{-}b{-}s}{1{-}b{-}r} \big[B_Q(1{-}b){-}B_Q(r)\big]   \big\ri_\cH                      \nonumber \\
			&\qquad\qquad\qquad\qquad\qquad\qquad\qquad{+}\big\li B_Q(b){-}B_Q(1){+}B_Q(1{-}b),  \frac{(1{-}b{-}r)(r^\ast_{p+1}{-}r)}{1-b-r} f_{p+2}  \big\ri_\cH\Big\}^2ds, \nonumber
		\end{align}	\normalsize	
		which implies  $\max_{k\in\A_{1n}}V_n(k)^{-1} = o_p(n^{-1+2\delta_{p+1}})$ and $\max_{k\in\A_{1n}}\frac{T_n(k)}{\sqrt{V_n(k)}} = O_p(1)$.
		
		\noindent
		\textbf{(5) $\theta_{p+1} \sim n^{-\zeta_{p+1}}$, $ \kappa_{p+1}=0$ or $ \kappa_{p+1}\sim n^{-\zeta_{p+1}'}$ with $\zeta_{p+1}'\geq\zeta_{p+1}$.}		
		
		From Equation (\ref{eq_1cp_1}), $\max_{k\in\A_{1n}}|T_n(k)| = O_p(n^{1-\zeta_{p+1}-\delta_{p+1}})$. Similar to the proof in part (1), we have
		\begin{align}
			\max_{k\in\A_{1n}} \frac{1}{(n{-}2nb)^{2}}\sum_{t=k^\ast_{p+1}{+}1}^{n{-}nb}R_{11}(k,t)^2 	= & O_p(n^{1-2\zeta_{p+1}}) \nonumber \\
			\min_{k\in\A_{1n}}	\frac{1}{(n{-}2nb)^{2}}\sum_{t=k^\ast_{p+1}{+}1}^{n{-}nb}R_{12}(k,t)^2 	\geq  & Cn^{-1}(\epsilon_{(p+1)n}n^{-\delta_{p+1}})^2n^{1-2\zeta_{p+1}}(1{+}o_p(1)),
		\end{align}
		which implies $\max_{k\in\A_{1n}}V_n(k)^{-1} = o_p(n^{-1+2\zeta_{p+1}})$ and $\max_{k\in\A_{1n}}\frac{T_n(k)}{\sqrt{V_n(k)}} = o_p(n^{1/2-\delta_{p+1}})$.

		\noindent
		\textbf{(6) $\theta_{p+1} \sim n^{-\zeta_{p+1}}$ and $ \kappa_{p+1}\sim n^{-\zeta_{p+1}'}$ with $\zeta_{p+1}'<\zeta_{p+1}$.}		
		
		From Equation (\ref{eq_1cp_1}), $\max_{k\in\A_{1n}}|T_n(k)| = O_p(n^{1-\zeta_{p+1}-\delta_{p+1}}+n^{1/2-\zeta_{p+1}'})$. Similar to the proof in part (4), we have
		\begin{align}
			\min_{k\in\A_{1n}^1}\frac{1}{(n{-}2nb)^{2}}\sum_{t=k^\ast_{p+1}{+}1}^{n{-}nb}\Big\{R_{11}(k,t){+}R_{12}(k,t) \Big\}^2=O_p^s(n^{2-2\zeta_{p+1}-2\delta_{p+1}}+n^{1-2\zeta_{p+1}'}),
		\end{align}		
		and $\min_{k\in\A_{1n}^2}L_n(k)  =O_p^s(n^{1-2\zeta_{p+1}'})$. If $1/2{-}\zeta_{p+1}'\geq 1{-}\zeta_{p+1}{-}\delta_{p+1}$, then we have $\max_{k\in\A_{1n}}V_n(k)^{-1} = O_p(n^{-1{+}2\zeta_{p+1}'})$ and $\max_{k\in\A_{1n}}\frac{T_n(k)}{\sqrt{V_n(k)}} = O_p(1)$. If $1/2{-}\zeta_{p+1}'< 1{-}\zeta_{p+1}{-}\delta_{p+1}$, then we have $\max_{k\in\A_{1n}}V_n(k)^{-1} = O_p(n^{-1{+}2\zeta_{p+1}'})=o_p(n^{-1+2\zeta_{p+1}})$ and $\max_{k\in\A_{1n}}\frac{T_n(k)}{\sqrt{V_n(k)}} =  o_p(n^{1/2-\delta_{p+1}})$.

	\end{proof}

	\subsection{Two change point setting}	\label{sec_2cp}		
	
	For nonnegative integers $q,p$ such that $0\leq p+q\leq m_0-2$, assume there are $p+q+2$ change points with relative location $0<r^\ast_1<\cdots<r^\ast_p<b<r^\ast_{p+1}<r^\ast_{p+2}<1-b<r^\ast_{p+3}<\cdots<r^\ast_{p+q+2}<1$. In this setting 
	\begin{align}
		\cP_n = &a_1\Delta^\ast_1{+}(a_1{+}a_2)\Delta^\ast_2{+}\cdots{+}({a_1{+}\cdots{+}a_p})\Delta^\ast_p{+}\Delta^\ast_{p+1}{+}\Delta^\ast_{p+2}\nonumber \\
		&{+}(1{-}b_1)\Delta^\ast_{p+3}{+}(1{-}b_1{-}b_2) \Delta^\ast_{p+4}{+}\cdots{+}(1{-}b_1{-}\cdots{-}b_q)\Delta^\ast_{p+q+2},   \label{def_pn_2cp}
	\end{align}		
	where $\{a_i\}$ are defined in section \ref{sec_0cp}, $b_1 = \frac{r^\ast_{p+3}-1+b}{1-b},b_{q+1} = \frac{1-r^\ast_{p+q+2}}{1-b}$ and $b_i = \frac{r^\ast_{p+i+2}-r^\ast_{p+i+1}}{1-b}$ for $i=2,3,\dots,q$.
	Let $\A_{1n}=\{k\geq \fb+1|k\leq k^\ast_{p+1}-\epsilon_{(p+1)n},\,k\in \mathcal{Z}\}$, $\A_{2n}=\{k^\ast_{p+1}{+}\epsilon_{(p+1)n} \leq k\leq k^\ast_{p+2}-\epsilon_{(p+2)n},\,k\in \mathcal{Z}\}$ and $\A_{3n}=\{k\leq n-\fb|k\geq k^\ast_{p+2}+\epsilon_{(p+2)n},\,k\in \mathcal{Z}\}$, we show that no change point can be selected inside $\A_{1n}\cup\A_{2n}\cup\A_{3n}$.
	
	\begin{lemma}\label{lemma_2cp}
		Suppose Assumption \ref{assump_wbs} holds, then
		\begin{align}
			\max_{k\in\A_{1n}\cup\A_{2n}\cup\A_{3n}}\frac{T_n(k)}{\sqrt{V_n(k)}} = o_p(n^{1/2-\min\{\delta_{p+1},\delta_{p+2}\}})
		\end{align}
	\end{lemma}
	
	\subsubsection{No change point can be selected on $\A_{1n}$}\label{sec_2cp_1}
	
	For $k\in\A_{1n}$,
	\begin{align}
		T_n(k) = & {-}\frac{(k{-}nb)n}{(n-2nb)^{3/2}}\big\li \bar Y_1'{-}\bar Y_3'{-}b\sqrt{n}\cP_n,(1{-}b{-}r^\ast_{p+1})\Delta^\ast_{p+1}{+} (1{-}b{-}r^\ast_{p+2})\Delta^\ast_{p+2} \big\ri_\cH \nonumber \\
		& {+}\big\li \bar Y_1'{-}\bar Y_3'{-}b\sqrt{n}\cP_n, \frac{1}{\sqrt{n-2nb}}\big[S_{nb+1,k}^{Y'}-\frac{k{-}nb}{n{-}2nb}S_{nb+1,n-nb}^{Y'}\big]\big\ri_\cH. \label{eq_2cp_1}
	\end{align}
	Similar to Section \ref{sec_1cp}, define $R_{1n}(k) = \frac{1}{(n{-}2nb)^{2}}\sum_{t=k^\ast_{p+1}+1}^{k^\ast_{p+2}}\Big\{S_{t,n-nb}{-}\frac{n{-}nb{-}t{+}1}{n{-}nb{-}k}S_{k+1,n-nb}\Big\}^2$ and $L_{n}(k) =\frac{1}{(n{-}2nb)^{2}}\sum_{t=nb+1}^{k}\Big\{S_{nb+1,t}{-}\frac{t{-}nb}{k{-}nb}S_{nb+1,k}\Big\}^2$, then
	\small\begin{align}		
		R_{1n}(k) = & \frac{1}{(n{-}2nb)^{2}}\sum_{t=k^\ast_{p+1}+1}^{k^\ast_{p+2}}\Big\{\big\li \bar Y_1'{-}\bar Y_3'{-}b\sqrt{n}\cP_n,S^{Y'}_{t,n-nb} {-}\frac{n{-}nb{-}t{+}1}{n{-}nb{-}k}S^{Y'}_{k+1,n-nb}    \big\ri_\cH  \nonumber \\
		&{+}  \frac{(n{-}nb{-}t{+}1)(k^\ast_{p+1}{-}k)}{n{-}nb{-}k}\big\li \bar Y'_1{-}\bar Y'_3{-}b\sqrt{n}\cP_n,  \Delta^\ast_{p+1}\big\ri_\cH{+}\frac{(n{-}nb{-}k^\ast_{p+2})(t{-}1{-}k)}{n{-}nb{-}k} \big\li\bar Y'_1{-}\bar Y'_3{-}b\sqrt{n}\cP_n,  \Delta^\ast_{p+2}   \big\ri_\cH    \Big\}^2 \nonumber \\
		=& \frac{1}{(n{-}2nb)^{2}}\sum_{t=k^\ast_{p+1}+1}^{k^\ast_{p+2}}\Big\{R_{11}(k,t){+}R_{12}(k,t){+}R_{13}(k,t) \Big\}^2,
	\end{align}	
	\begin{align}
		L_n(k) = & \frac{1}{(n{-}2nb)^{2}}\sum_{t=nb{+}1}^{k}\big\li \bar Y_1'{-}\bar Y_3'{-}b\sqrt{n}\cP_n,S^{Y'}_{nb+1,t} {-}\frac{t{-}nb}{k{-}nb}S^{Y'}_{nb+1,k}    \big\ri_\cH^2 \nonumber 
	\end{align}		\normalsize	
	$R_n(k) \geq R_{1n}(k)$ and $V_n(k) =R_n(k)+L_n(k)$. We consider the following different scenarios.	
	
	\noindent
	\textbf{(1) $\cP_n=0$.}	
	
	First, if $\delta_{p+1}>\delta_{p+2}$, then $\max_{k\in\A_{1n}}|T_n(k)| = O_p(n^{1/2-\delta_{p+2}})$ and
	\begin{align}
		\max_{k\in\A_{1n}}\frac{1}{(n{-}2nb)^{2}}\sum_{t=k^\ast_{p+1}+1}^{k^\ast_{p+2}}R_{11}(k,t)^2 =& O_p(1),\nonumber \\
		\max_{k\in\A_{1n}}\frac{1}{(n{-}2nb)^{2}}\sum_{t=k^\ast_{p+1}+1}^{k^\ast_{p+2}}R_{12}(k,t)^2 =& O_p(n^{1-2\delta_{p+1}}),\nonumber \\
		\min_{k\in\A_{1n}}\frac{1}{(n{-}2nb)^{2}}\sum_{t=k^\ast_{p+1}+1}^{k^\ast_{p+2}}R_{13}(k,t)^2 =&O_p^s(n^{1-2\delta_{p+2}}).\nonumber
	\end{align}
	So we have 	$\max_{k\in\A_{1n}}V_n(k)^{-1}\geq \max_{k\in\A_{1n}}R_{1n}(k)^{-1} = O_p(n^{-1+2\delta_{p+2}})$ and $\max_{k\in\A_{1n}}\frac{T_n(k)}{\sqrt{V_n(k)}} =  O_p(1)$.
	
	Second, if $\delta_{p+1}=\delta_{p+2}$, then $\max_{k\in\A_{1n}}|T_n(k)| = O_p(n^{1/2-\delta_{p+1}})$ and
	\small\begin{align}
		&\min_{nb+1\leq k\leq k^\ast_{p+1}}\frac{n^{-1+2\delta_{p+1}}}{(n{-}2nb)^{2}}\sum_{t=k^\ast_{p+1}+1}^{k^\ast_{p+2}}\big\{R_{12}(k,t){+}R_{13}(k,t)\big\}^2\nonumber \\
		\stackrel{\D}{\to}&\frac{1}{(1{-}2b)^2}\inf_{r\in[b,r^\ast_{p+1}]}\int_{r^\ast_{p+1}}^{r^\ast_{p+2}}\Big\li B_Q(b){-}B_Q(1){+}B_Q(1{-}b), \frac{(1{-}b{-}s)(r^\ast_{p+1}{-}r)}{1-b-r} f_{p+1}{+} \frac{(1{-}b{-}r^\ast_{p+2})(s{-}r)}{1-b-r} f_{p+2}  \Big\ri_\cH^2ds. \nonumber
	\end{align}	\normalsize
	So we have $\min_{k\in\A_{1n}}R_{1n}(k) = O^s_p(n^{1-2\delta_{p+1}})$ and $\max_{k\in\A_{1n}}V_n(k)^{-1} = O_p(n^{-1+2\delta_{p+1}})$, which implies $\max_{k\in\A_{1n}}\frac{T_n(k)}{\sqrt{V_n(k)}} =  O_p(1)$.
	
	At last, if $\delta_{p+1}<\delta_{p+2}$, then $\max_{k\in\A_{1n}}|T_n(k)| = O_p(n^{1/2-\delta_{p+1}})$. We now show that 
	\begin{align}
		\min_{k\in\A_{1n}}\frac{n^{-1+\epsilon+2\delta_{p+2}}}{(n{-}2nb)^{2}}\sum_{t=k^\ast_{p+1}+1}^{k^\ast_{p+2}}\big\{R_{12}(k,t){+}R_{13}(k,t)\big\}^2\stackrel{p}{\to}\infty \text{ for any small }\epsilon>0,\label{eq_2cp4}
	\end{align}
	which implies $\max_{k\in\A_{1n}}V_n(k)^{-1} = o_p(1)$ and $\max_{k\in\A_{1n}}\frac{T_n(k)}{\sqrt{V_n(k)}} =  o_p(n^{1/2-\delta_{p+1}})$. Denote $\tilde a _{p+1} = \li \bar Y'_1{-}\bar Y'_3,f_{p+1}\ri_\cH$, $ \tilde a _{p+2} = \li \bar Y'_1{-}\bar Y'_3,f_{p+2}\ri_\cH$	and note that $P(\tilde a _{p+1}\tilde a _{p+2}\neq 0)\to 1$. If $\tilde a _{p+1}\tilde a _{p+2}>0$, then 
	\begin{align}
		\min_{k\in\A_{1n}}\frac{1}{(n{-}2nb)^{2}}\sum_{t=k^\ast_{p+1}+1}^{k^\ast_{p+2}}\big\{R_{12}(k,t){+}R_{13}(k,t)\big\}^2\geq & \min_{k\in\A_{1n}}\frac{1}{(n{-}2nb)^{2}}\sum_{t=k^\ast_{p+1}+1}^{k^\ast_{p+2}}R_{13}(k,t)^2 \nonumber \\
		=&O^s(n^{1-2\delta_{p+2}}). \nonumber 
	\end{align}
	If $\tilde a _{p+1}<0<\tilde a _{p+2}$, denote $\tilde \A_{1n} = \{k\geq k^\ast_{p+1}{-}(\log n)^{-1}n^{1-\delta_{p+2}+\delta_{p+1}}\}$	and we have 
	\begin{align}
		&\min_{k\in\A_{1n}}\frac{1}{(n{-}2nb)^{2}}\sum_{t=k^\ast_{p+1}+1}^{k^\ast_{p+2}}\big\{R_{12}(k,t){+}R_{13}(k,t)\big\}^2 \nonumber\\
		\geq & \frac{1}{(n{-}2nb)^{2}}\sum_{t=k^\ast_{p+1}+1}^{k^\ast_{p+2}}\min_{k\in\A_{1n}}\Big\{ \frac{(n{-}nb{-}t{+}1)(k^\ast_{p+1}{-}k)n^{-\delta_{p+1}}}{n{-}nb{-}k}\tilde a_{p+1}{+}\frac{(n{-}nb{-}k^\ast_{p+2})(t{-}1{-}k)n^{-\delta_{p+2}}}{n{-}nb{-}k} \tilde a_{p+2}\Big\}^2. \nonumber
	\end{align}
	If $k\in \A_{1n}\cap\tilde \A_{1n}$ and $t\geq \lfloor(r^\ast_{p+1}{+}\epsilon)n\rfloor$ for some $\epsilon\in (r^\ast_{p+1},r^\ast_{p+2})$, we have $\frac{(n{-}nb{-}t{+}1)(k^\ast_{p+1}{-}k)n^{-\delta_{p+1}}}{n{-}nb{-}k}|\tilde a_{p+1}|\leq C(\log n)^{-1}n^{1-\delta_{p+2}}$ and $\frac{(n{-}nb{-}k^\ast_{p+2})(t{-}1{-}k)n^{-\delta_{p+2}}}{n{-}nb{-}k} \tilde a_{p+2}\geq Cn^{1-\delta_{p+2}}$ for large enough $n$ where the constant $C$ does not depend on $k$ and $t$. So we have 
	\begin{align}
		&\min_{k\in\A_{1n}\cap\tilde \A_{1n}}\frac{1}{(n{-}2nb)^{2}}\sum_{t=k^\ast_{p+1}+1}^{k^\ast_{p+2}}\big\{R_{12}(k,t){+}R_{13}(k,t)\big\}^2 \nonumber\\
		\geq &  \frac{C}{(n{-}2nb)^{2}}\sum_{t=\lfloor(r^\ast_{p+1}{+}\epsilon)n\rfloor}^{k^\ast_{p+2}}\min_{k\in\A_{1n}\cap\tilde \A_{1n}}R_{13}(k,t)^2\nonumber\\
		=& O^s(n^{1-2\delta_{p+2}}).\label{eq_2cp_2}
	\end{align}
	If 	$k\in \A_{1n}\setminus\tilde \A_{1n}$ and $t\leq k^\ast_{p+1}{+} \lfloor n^{1-\epsilon}\rfloor$ for some $\epsilon\in(0,\min\{1{-}2\delta_{p+2},\delta_{p+2}{-}\delta_{p+1}\})$, we have $\frac{(n{-}nb{-}t{+}1)(k^\ast_{p+1}{-}k)n^{-\delta_{p+1}}}{n{-}nb{-}k}|\tilde a_{p+1}|\geq C\max\{\epsilon_{(p+1)n}n^{-\delta_{p+1}},(\log n)^{-1}n^{1-\delta_{p+2}}\}$ and $\frac{(n{-}nb{-}k^\ast_{p+2})(t{-}1{-}k)n^{-\delta_{p+2}}}{n{-}nb{-}k} \tilde a_{p+2}\leq C\max\{\epsilon_{(p+1)n}n^{-\delta_{p+2}},n^{1-\delta_{p+2}-\epsilon}\}$ for large enough $n$ where the constant $C$ does not depend on $k$ and $t$. Since $\frac{\max\{\epsilon_{(p+1)n}n^{-\delta_{p+1}},(\log n)^{-1}n^{1-\delta_{p+2}}\}}{\max\{\epsilon_{(p+1)n}n^{-\delta_{p+2}},n^{1-\delta_{p+2}-\epsilon}\}}\geq (\log n)^{-1}n^{\epsilon}$, we have 
	\begin{align}
		&\min_{k\in\A_{1n}\setminus\tilde \A_{1n}}\frac{1}{(n{-}2nb)^{2}}\sum_{t=k^\ast_{p+1}+1}^{k^\ast_{p+2}}\big\{R_{12}(k,t){+}R_{13}(k,t)\big\}^2 \nonumber\\
		\geq &  \frac{C}{(n{-}2nb)^{2}}\sum_{t=k^\ast_{p+1}+1}^{k^\ast_{p+1}{+} \lfloor n^{1-\epsilon}\rfloor}\min_{k\in\A_{1n}\setminus\tilde \A_{1n}}R_{12}(k,t)^2\nonumber\\
		\geq & O^s((\log n)^{-2}n^{1-\epsilon-2\delta_{p+2}}){\to}\infty. \label{eq_2cp_3}
	\end{align}	
	By Equations (\ref{eq_2cp_2}) and (\ref{eq_2cp_3}), we get $\min_{k\in\A_{1n}}\frac{1}{(n{-}2nb)^{2}}\sum_{t=k^\ast_{p+1}+1}^{k^\ast_{p+2}}\big\{R_{12}(k,t){+}R_{13}(k,t)\big\}^2\stackrel{p}{\to}\infty$.	
	
	\noindent
	\textbf{(2) $\cP_n\neq 0$ and $\cP_n\in\{f_{p+1},f_{p+2}\}^\perp$.}	
	
	If $\cP_n\neq 0$, we can write $\cP_n = \kappa_{p+1} f'_{p+1}{+}f^n_{p+1}$ for some $f'_{p+1},f^n_{p+1}\in \cH$ such that $\kappa_{p+1} \sim n^{-\zeta_{p+1}'}$, $\zeta_{p+1}'\in [0,1/2)$ and $\|f^n_{p+1}\|_\cH =o(\kappa_{p+1})$.
	In this scenario we have $\li\cP_n,\Delta^\ast_{p+1}\ri_\cH = \li\cP_n,\Delta^\ast_{p+2}\ri_\cH=0$.
	
	First, if $\zeta_{p+1}'<\min\{\delta_{p+1},\delta_{p+2}\}$, then $\max_{k\in\A_{1n}}|T_n(k)| = O_p(n^{1/2-\zeta_{p+1}'})$ and
	\begin{align}
		\min_{k\in\A_{1n}}\frac{1}{(n{-}2nb)^{2}}\sum_{t=k^\ast_{p+1}+1}^{k^\ast_{p+2}}R_{11}(k,t)^2 =& O^s_p(n^{1-2\zeta_{p+1}'}),\nonumber \\
		\max_{k\in\A_{1n}}\frac{1}{(n{-}2nb)^{2}}\sum_{t=k^\ast_{p+1}+1}^{k^\ast_{p+2}}R_{12}(k,t)^2 =& O_p(n^{1-2\delta_{p+1}}),\nonumber \\
		\max_{k\in\A_{1n}}\frac{1}{(n{-}2nb)^{2}}\sum_{t=k^\ast_{p+1}+1}^{k^\ast_{p+2}}R_{13}(k,t)^2 =&O_p(n^{1-2\delta_{p+2}}).\nonumber
	\end{align}	
	So we have 	$\max_{k\in\A_{1n}}V_n(k)^{-1}\geq \max_{k\in\A_{1n}}R_{1n}(k)^{-1} = O_p(n^{-1+2\zeta_{p+1}'})$ and $\max_{k\in\A_{1n}}\frac{T_n(k)}{\sqrt{V_n(k)}} =  O_p(1)$.	
	
	Second, if $\zeta_{p+1}'>\min\{\delta_{p+1},\delta_{p+2}\}$, then $\max_{k\in\A_{1n}}|T_n(k)| = O_p(n^{1/2-\min\{\delta_{p+1},\delta_{p+2}\}})$. Similar to the proof in part (1), when $\delta_{p+1}>\delta_{p+2}$, $\delta_{p+1}=\delta_{p+2}$ or $\delta_{p+1}<\delta_{p+2}$, we can show that $\max_{k\in\A_{1n}}\frac{T_n(k)}{\sqrt{V_n(k)}} = o_p(n^{1/2-\delta_{p+1}})$.
	
	At last, consider the case when $\zeta_{p+1}'=\min\{\delta_{p+1},\delta_{p+2}\}$. If $\delta_{p+1}>\delta_{p+2}$, then $\max_{k\in\A_{1n}}|T_n(k)| = O_p(n^{1/2-\delta_{p+2}})$, $\max_{k\in\A_{1n}}\frac{1}{(n{-}2nb)^{2}}\sum_{t=k^\ast_{p+1}+1}^{k^\ast_{p+2}}R_{12}(k,t)^2 = O_p(n^{1-2\delta_{p+1}})$ and 
	\small\begin{align}
		&\min_{nb+1\leq k\leq k^\ast_{p+1}}\frac{n^{-1+2\delta_{p+2}}}{(n{-}2nb)^{2}}\sum_{t=k^\ast_{p+1}+1}^{k^\ast_{p+2}}\big\{R_{11}(k,t){+}R_{13}(k,t)\big\}^2\nonumber \\
		\stackrel{\D}{\to}&\frac{1}{(1{-}2b)^2}\inf_{r\in[b,r^\ast_{p+1}]}\int_{r^\ast_{p+1}}^{r^\ast_{p+2}}\Big\{{-} \big\li Cbf_{p+1}', B_Q(1{-}b){-}B_Q(s){-} \frac{1{-}b{-}s}{1{-}b{-}r} \big[B_Q(1{-}b){-}B_Q(r)\big]   \big\ri_\cH                      \nonumber \\
		&\qquad\qquad\qquad\qquad\qquad{+}\big\li B_Q(b){-}B_Q(1){+}B_Q(1{-}b),  \frac{(1{-}b{-}r^\ast_{p+2})(s{-}r)}{1-b-r} f_{p+2}  \big\ri_\cH\Big\}^2ds. \label{eq_2cp_4}
	\end{align}	\normalsize	
	So we have $\min_{k\in\A_{1n}}R_{1n}(k) = O^s_p(n^{1-2\delta_{p+2}})$ and $\max_{k\in\A_{1n}}V_n(k)^{-1} = O_p(n^{-1+2\delta_{p+2}})$, which implies $\max_{k\in\A_{1n}}\frac{T_n(k)}{\sqrt{V_n(k)}} =  O_p(1)$.	If $\delta_{p+1}\leq \delta_{p+2}$, $\max_{k\in\A_{1n}}|T_n(k)| = O_p(n^{1/2-\delta_{p+1}})$. As in part (3) of Section \ref{sec_1cp}, let $\A_{1n}^1 = \A_{1n}\cap \{k\leq n\frac{b+r^\ast_{p+1}}{2}\}$ and $\A_{1n}^2 = \A_{1n}\cap \{k> n\frac{b+r^\ast_{p+1}}{2}\}$, then we have $\min_{k\in\A_{1n}^1}V_{n}(k)\geq\min_{k\in\A_{1n}^1}R_{1n}(k) = O_p^s(n^{1-2\delta_{p+1}})$
	and $\min_{k\in\A_{1n}^2}V_n(k) \geq \min_{k\in\A_{1n}^2}L_n(k)  =O_p^s(n^{1-2\delta_{p+1}})$, which implies $\max_{k\in\A_{1n}}\frac{T_n(k)}{\sqrt{V_n(k)}} =  O_p(1)$.
	
	\noindent
	\textbf{(3) $\cP_n\neq 0$,  $\li\cP_n,\Delta^\ast_{p+1}\ri_\cH = 0$ and $\li\cP_n,\Delta^\ast_{p+2}\ri_\cH\neq0$.}		
	
	In this case, $\li\cP_n,\Delta^\ast_{p+2}\ri_\cH \sim n^{-\delta_{p+2}-\zeta_{p+2}}$ for some $\zeta_{p+2}\in[\zeta_{p+1}',1/2)$ and $\max_{k\in\A_{1n}}|T_n(k)| = O_p(n^{1/2-\delta_{p+1}}){+} O_p(n^{1/2-\zeta_{p+1}'}){+} O_p(n^{1-\delta_{p+2}-\zeta_{p+2}})$. If $\max\{n^{1/2-\zeta_{p+1}'},n^{1-\delta_{p+2}-\zeta_{p+2}}\}> n^{1/2-\delta_{p+1}}$, then $\min_{k\in\A_{1n}}R_{1n}(k)$ is of the same order as 
	$$ \min_{k\in\A_{1n}} \frac{1}{(n{-}2nb)^{2}}\sum_{t=k^\ast_{p+1}+1}^{k^\ast_{p+2}}\Big\{R_{11}(k,t){+}R_{13}(k,t) \Big\}^2 = O_p^s(\max\{n^{1-2\zeta_{p+1}'},n^{2-2\delta_{p+2}-2\zeta_{p+2}} \}),$$ so we have $\max_{k\in\A_{1n}}V_n(k)^{-1} = O_p(\min\{n^{-1+2\zeta_{p+1}'},n^{-2+2\delta_{p+2}+2\zeta_{p+2}} \})$, which implies $\max_{k\in\A_{1n}}\frac{T_n(k)}{\sqrt{V_n(k)}} =  O_p(1)$. If $\max\{n^{1/2-\zeta_{p+1}'},n^{1-\delta_{p+2}-\zeta_{p+2}}\}\leq n^{1/2-\delta_{p+1}}$, then $\min_{k\in\A_{1n}^1}R_{1n}(k) = O_p^s(n^{1/2-\delta_{p+1}})$, $\min_{k\in\A_{1n}^2}L_n(k)  =O_p^s(n^{1-2\zeta_{p+1}'})$ and $\max_{k\in\A_{1n}}V_n(k)^{-1} = o_p(1)$, which implies $\max_{k\in\A_{1n}}\frac{T_n(k)}{\sqrt{V_n(k)}} =  o_p(n^{1/2-\delta_{p+1}})$.
	
	\noindent
	\textbf{(4) $\cP_n\neq 0$,  $\li\cP_n,\Delta^\ast_{p+1}\ri_\cH \neq 0$ and $\li\cP_n,\Delta^\ast_{p+2}\ri_\cH=0$.}			
	
	In this case, $\li\cP_n,\Delta^\ast_{p+1}\ri_\cH \sim n^{-\delta_{p+1}-\zeta_{p+1}}$ for some $\zeta_{p+1}\in[\zeta_{p+1}',1/2)$ and $\max_{k\in\A_{1n}}|T_n(k)| = O_p(n^{1/2-\delta_{p+2}}){+} O_p(n^{1/2-\zeta_{p+1}'}){+} O_p(n^{1-\delta_{p+1}-\zeta_{p+1}})$.	
	
	If $\max\{n^{1/2-\zeta_{p+1}'},n^{1/2-\delta_{p+2}}\}\geq n^{1-\delta_{p+1}-\zeta_{p+1}} $, $\min_{k\in\A_{1n}}R_{1n}(k) = O_p^s(\max\{n^{1-2\zeta_{p+1}'},n^{1-2\delta_{p+2}} \})$, so we have $\max_{k\in\A_{1n}}V_n(k)^{-1} = O_p(\min\{n^{-1+2\zeta_{p+1}'},n^{-1+2\delta_{p+2}} \})$, which implies $\max_{k\in\A_{1n}}\frac{T_n(k)}{\sqrt{V_n(k)}} =  O_p(1)$. If $\max\{n^{1/2-\zeta_{p+1}'},n^{1/2-\delta_{p+2}}\}< n^{1-\delta_{p+1}-\zeta_{p+1}} $ and $\zeta_{p+1}'<\zeta_{p+1}$, then $\min_{k\in\A_{1n}^1}R_{1n}(k) = O_p^s(n^{2-2\delta_{p+1}-2\zeta_{p+1}})$, $\min_{k\in\A_{1n}^2}L_n(k)  =O_p^s(n^{1-2\zeta_{p+1}'})$ and $\max_{k\in\A_{1n}}V_n(k)^{-1} =O_p(n^{-1+2\zeta_{p+1}'})$, which implies $\max_{k\in\A_{1n}}\frac{T_n(k)}{\sqrt{V_n(k)}} =  o_p(n^{1/2-\delta_{p+1}})$.
	
	At last, consider the case when $\max\{n^{1/2-\zeta_{p+1}'},n^{1/2-\delta_{p+2}}\}< n^{1-\delta_{p+1}-\zeta_{p+1}} $ and $\zeta_{p+1}'=\zeta_{p+1}$. If $n^{1/2-\delta_{p+2}}\leq n^{1/2-\zeta_{p+1}'}$, 
	\small\begin{align}
		\min_{k\in\A_{1n}}	\frac{n^{-1+2\zeta_{p+1}}}{(n{-}2nb)^{2}}\sum_{t=k^\ast_{p+1}{+}1}^{k^\ast_{p+2}}R_{12}(k,t)^2  =& \min_{k\in\A_{1n}}	\frac{n^{-2\delta_{p+1}}(k^\ast_{p+1}{-}k)^2}{(n{-}2nb)^{2}(n{-}nb{-}k)^2}\frac{\li\bar Y'_1{-}\bar Y'_3{-}b\sqrt{n}\cP_n,\Delta^\ast_{p+1} \ri_\cH^2}{n^{1{-}2\delta_{p+1}{-}2\zeta_{p+1}}}\sum_{t=k^\ast_{p+1}{+}1}^{k^\ast_{p+2}}(n{-}nb{-}t{+}1)^2 \nonumber \\
		\geq & Cn^{-1}(\epsilon_{(p+1)n}n^{-\delta_{p+1}})^2\frac{\li\bar Y'_1{-}\bar Y'_3{-}b\sqrt{n}\cP_n,\Delta^\ast_{p+1} \ri_\cH^2}{n^{1{-}2\delta_{p+1}{-}2\zeta_{p+1}}}\nonumber\\
		\stackrel{p}{\to}& \infty,\label{eq_2cp_5}
	\end{align}	\normalsize
	so we have $\max_{k\in\A_{1n}}V_n(k)^{-1} =o_p(n^{-1+2\zeta_{p+1}})$, which implies $\max_{k\in\A_{1n}}\frac{T_n(k)}{\sqrt{V_n(k)}} =  o_p(n^{1/2-\delta_{p+1}})$. If $n^{1/2-\delta_{p+2}}> n^{1/2-\zeta_{p+1}'}$, using similar method as in the proof of Equation (\ref{eq_2cp4}), we can show $n^{-1+\epsilon+2\delta_{p+2}}\min_{k\in\A_{1n}}R_{1n}(k)\stackrel{p}{\to}\infty$ for any small $\epsilon>0$, which implies $\max_{k\in\A_{1n}}V_n(k)^{-1} =o_p(n^{-1+2\zeta_{p+1}})$ and $\max_{k\in\A_{1n}}\frac{T_n(k)}{\sqrt{V_n(k)}} =  o_p(n^{1/2-\delta_{p+1}})$.
	
	\noindent
	\textbf{(5) $\cP_n\neq 0$,  $\li\cP_n,\Delta^\ast_{p+1}\ri_\cH \neq 0$ and $\li\cP_n,\Delta^\ast_{p+2}\ri_\cH\neq0$.}		
	
	In this case, $\li\cP_n,\Delta^\ast_{p+2}\ri_\cH \sim n^{-\delta_{p+2}-\zeta_{p+2}}$, $\li\cP_n,\Delta^\ast_{p+1}\ri_\cH \sim n^{-\delta_{p+1}-\zeta_{p+1}}$ for some $\zeta_{p+1},\zeta_{p+2}\in[\zeta_{p+1}',1/2)$ and $\max_{k\in\A_{1n}}|T_n(k)| = O_p(n^{1-\delta_{p+2}-\zeta_{p+2}}){+} O_p(n^{1/2-\zeta_{p+1}'}){+} O_p(n^{1-\delta_{p+1}-\zeta_{p+1}})$. Following similar method as in part (4), we can show $\max_{k\in\A_{1n}}\frac{T_n(k)}{\sqrt{V_n(k)}} =  o_p(n^{1/2-\delta_{p+1}})$.

	\subsubsection{No change point can be selected on $\A_{2n}$} \label{sec_2cp_2}
	
	With out loss of generality, we assume $\delta_{p+1}\leq \delta_{p+2}$ and $k\in\A_{2n}$,
	\small\begin{align}
		T_n(k) = & {-}\frac{(n{-}nb{-}k)n}{(n-2nb)^{3/2}}\big\li \bar Y_1'{-}\bar Y_3'{-}b\sqrt{n}\cP_n,(r^\ast_{p+1}{-}b)\Delta^\ast_{p+1} \big\ri_\cH{-}\frac{(k{-}nb)n}{(n-2nb)^{3/2}}\big\li \bar Y_1'{-}\bar Y_3'{-}b\sqrt{n}\cP_n, (1{-}b{-}r^\ast_{p+2})\Delta^\ast_{p+2} \big\ri_\cH \nonumber \\
		& {+}\big\li \bar Y_1'{-}\bar Y_3'{-}b\sqrt{n}\cP_n, \frac{1}{\sqrt{n-2nb}}\big[S_{nb+1,k}^{Y'}-\frac{k{-}nb}{n{-}2nb}S_{nb+1,n-nb}^{Y'}\big]\big\ri_\cH. 
	\end{align}\normalsize
	Similar to Section \ref{sec_1cp}, define $R_{1n}(k) = \frac{1}{(n{-}2nb)^{2}}\sum_{t=k^\ast_{p+2}{+}1}^{n{-}nb}\Big\{S_{t,n-nb}{-}\frac{n{-}nb{-}t{+}1}{n{-}nb{-}k}S_{k+1,n-nb}\Big\}^2$ and $L_{1n}(k) = \frac{1}{(n{-}2nb)^{2}}\sum_{t=nb{+}1}^{k^\ast_{p+1}}\Big\{S_{nb+1,t}{-}\frac{t{-}nb}{k{-}nb}S_{nb+1,k}\Big\}^2$, then
	\small\begin{align}
		R_{1n}(k) = & \frac{1}{(n{-}2nb)^{2}}\sum_{t=k^\ast_{p+2}+1}^{n-nb}\Big\{\big\li \bar Y_1'{-}\bar Y_3'{-}b\sqrt{n}\cP_n,S^{Y'}_{t,n-nb} {-}\frac{n{-}nb{-}t{+}1}{n{-}nb{-}k}S^{Y'}_{k+1,n-nb}    \big\ri_\cH  \nonumber \\
		&\qquad\qquad\qquad\qquad{+} \frac{(n{-}nb{-}t{+}1)(k^\ast_{p+2}{-}k)}{n{-}nb{-}k}    \big\li \bar Y'_1{-}\bar Y'_3{-}b\sqrt{n}\cP_n,\Delta^\ast_{p+2}   \big\ri_\cH    \Big\}^2 \nonumber \\
		=& \frac{1}{(n{-}2nb)^{2}}\sum_{t=k^\ast_{p+2}{+}1}^{n{-}nb}\Big\{R_{11}(k,t){+}R_{12}(k,t) \Big\}^2,
	\end{align}	
	\begin{align}
		L_{1n}(k) = & \frac{1}{(n{-}2nb)^{2}}\sum_{t=nb{+}1}^{k^\ast_{p+1}}\Big\{\big\li \bar Y_1'{-}\bar Y_3'{-}b\sqrt{n}\cP_n,S^{Y'}_{nb+1,t} {-}\frac{t{-}nb}{k{-}nb}S^{Y'}_{nb+1,k}    \big\ri_\cH  \nonumber \\
		&\qquad\qquad\qquad{-}\frac{(t{-}nb)(k{-}k^\ast_{p+1})}{k{-}nb}    \big\li \bar Y'_1{-}\bar Y'_3{-}b\sqrt{n}\cP_n, \Delta^\ast_{p+1} \big\ri_\cH    \Big\}^2 \nonumber \\
		=& \frac{1}{(n{-}2nb)^{2}}\sum_{t=nb+1}^{k^\ast_{p+1}}\Big\{L_{11}(k,t){+}L_{12}(k,t) \Big\}^2,
	\end{align}\normalsize	
	$R_n(k) \geq R_{1n}(k)$, $L_n(k) \geq L_{1n}(k)$ and $V_n(k) =R_n(k)+L_n(k)$. We consider the following different scenarios.		
	
	\noindent
	\textbf{(1) $\cP_n=0$.}		
	
	In this case $\max_{k\in\A_{1n}}|T_n(k)| = O_p(n^{1/2-\delta_{p+1}})$ and it suffice to show $\max_{k\in\A_{1n}}V_n(k)^{-1} =o_p(1)$. Let $\A_{2n}^1 = \A_{2n}\cap \{k\leq n\frac{r^\ast_{p+1}{+}r^\ast_{p+2}}{2}\}$ and $\A_{2n}^2 = \A_{2n}\cap \{k> n\frac{r^\ast_{p+1}{+}r^\ast_{p+2}}{2}\}$, then we have $\min_{k\in\A_{2n}^1}R_{1n}(k) = O_p^s(n^{1-2\delta_{p+2}})$ and $\min_{k\in\A_{2n}^2}L_{1n}(k) = O_p^s(n^{1-2\delta_{p+1}})$, which implies $\max_{k\in\A_{1n}}V_n(k)^{-1} =o_p(1)$.
	
	\noindent
	\textbf{(2) $\cP_n\neq 0$ and $\cP_n\in\{f_{p+1},f_{p+2}\}^\perp$.}	
	
	Similar to part (2) of Section \ref{sec_2cp_1}, we can write $\cP_n = \kappa_{p+1} f'_{p+1}{+}f^n_{p+1}$ for some $f'_{p+1},f^n_{p+1}\in \cH$ such that $\kappa_{p+1} \sim n^{-\zeta_{p+1}'}$, $\zeta_{p+1}'\in [0,1/2)$ and $\|f^n_{p+1}\|_\cH =o(\kappa_{p+1})$.
	In this scenario we have $\li\cP_n,\Delta^\ast_{p+1}\ri_\cH = \li\cP_n,\Delta^\ast_{p+2}\ri_\cH=0$.	If $\zeta_{p+1}'<\delta_{p+1}$, $\max_{k\in\A_{1n}}\frac{T_n(k)}{\sqrt{V_n(k)}} =  O_p(1)$ can be shown in the same way as in part (2) of Section \ref{sec_2cp_1}. If $\zeta_{p+1}'>\delta_{p+1}$, it suffice to show $\max_{k\in\A_{1n}}V_n(k)^{-1} =o_p(1)$, which follows in the same way as in part (1). If $\zeta_{p+1}'=\delta_{p+1}$, we can show $\min_{k\in\A_{2n}}L_{1n}(k)=O_p^s(n^{1-2\delta_{p+1}})$ in the same way as in Equation (\ref{eq_2cp_4})
	
	\noindent
	\textbf{(3) $\cP_n\neq 0$.  $\li\cP_n,\Delta^\ast_{p+1}\ri_\cH = 0$, $\li\cP_n,\Delta^\ast_{p+2}\ri_\cH\neq0$ or $\li\cP_n,\Delta^\ast_{p+1}\ri_\cH \neq 0$, $\li\cP_n,\Delta^\ast_{p+2}\ri_\cH=0$.}		
	
	We only give the proof when $\li\cP_n,\Delta^\ast_{p+1}\ri_\cH = 0$ and $\li\cP_n,\Delta^\ast_{p+2}\ri_\cH\neq0$ since the proof for $\li\cP_n,\Delta^\ast_{p+1}\ri_\cH \neq 0$ and $\li\cP_n,\Delta^\ast_{p+2}\ri_\cH=0$ is similar.
	In this case, $\li\cP_n,\Delta^\ast_{p+2}\ri_\cH \sim n^{-\delta_{p+2}-\zeta_{p+2}}$ for some $\zeta_{p+2}\in[\zeta_{p+1}',1/2)$ and $\max_{k\in\A_{1n}}|T_n(k)| = O_p(n^{1/2-\delta_{p+1}}){+} O_p(n^{1/2-\zeta_{p+1}'}){+} O_p(n^{1-\delta_{p+2}-\zeta_{p+2}})$. Note that 
	\small\begin{align}\label{eq_2cp_6}
		\min_{k\in\A_{2n}^2}L_{1n}(k) = O_p^s(\max\{n^{1-2\zeta_{p+1}'},n^{1-2\delta_{p+1}}\}), \min_{k\in\A_{2n}^1}R_{1n}(k) = O_p^s(\max\{n^{1-2\zeta_{p+1}'},n^{2-2\delta_{p+2}{-}2\zeta_{p+2}}\}).
	\end{align}\normalsize
	If $n^{1/2-\delta_{p+1}}\geq \max\{n^{1/2-\zeta_{p+1}'},n^{1-\delta_{p+2}{-}\zeta_{p+2}}\}$, since $\max_{k\in\A_{1n}}V_n(k)^{-1} =o_p(1)$, we have $\max_{k\in\A_{1n}}\frac{T_n(k)}{\sqrt{V_n(k)}} =  o_p(n^{1/2-\delta_{p+1}})$. If $n^{1/2-\zeta_{p+1}'}\geq \max\{n^{1/2-\delta_{p+1}},n^{1-\delta_{p+2}{-}\zeta_{p+2}}\}$, it is easy to see $\max_{k\in\A_{1n}}\frac{T_n(k)}{\sqrt{V_n(k)}} =  O_p(1)$. If $ n^{1-\delta_{p+2}-\zeta_{p+2}}\geq\max\{n^{1/2-\delta_{p+1}},n^{1/2-\zeta_{p+1}'}\}$ and $\zeta_{p+1}'=\zeta_{p+2}$, similar to Equation (\ref{eq_2cp_5}), we can show $\min_{k\in\A_{2n}}R_n(k)$ diverges to infinity faster than $n^{1-2\zeta_{p+1}}$, which implies $\max_{k\in\A_{1n}}V_n(k)^{-1} =o_p(n^{-1+2\zeta_{p+1}})$ and $\max_{k\in\A_{1n}}\frac{T_n(k)}{\sqrt{V_n(k)}} =  o_p(n^{1/2-\delta_{p+2}})$. At last, if $ n^{1-\delta_{p+2}-\zeta_{p+2}}\geq\max\{n^{1/2-\delta_{p+1}},n^{1/2-\zeta_{p+1}'}\}$ and $\zeta_{p+1}'<\zeta_{p+2}$, by Equation (\ref{eq_2cp_6}), we have $\max_{k\in\A_{1n}}V_n(k)^{-1} =o_p(n^{-1+2\zeta_{p+1}})$ and $\max_{k\in\A_{1n}}\frac{T_n(k)}{\sqrt{V_n(k)}} =  o_p(n^{1/2-\delta_{p+2}})$.
	
	\noindent
	\textbf{(4) $\cP_n\neq 0$,  $\li\cP_n,\Delta^\ast_{p+1}\ri_\cH \neq 0$ and $\li\cP_n,\Delta^\ast_{p+2}\ri_\cH\neq0$.}			
	
	In this case, $\li\cP_n,\Delta^\ast_{p+2}\ri_\cH \sim n^{-\delta_{p+2}-\zeta_{p+2}}$, $\li\cP_n,\Delta^\ast_{p+1}\ri_\cH \sim n^{-\delta_{p+1}-\zeta_{p+1}}$ for some $\zeta_{p+1},\zeta_{p+2}\in[\zeta_{p+1}',1/2)$ and $\max_{k\in\A_{1n}}|T_n(k)| = O_p(n^{1-\delta_{p+2}-\zeta_{p+2}}){+} O_p(n^{1/2-\zeta_{p+1}'}){+} O_p(n^{1-\delta_{p+1}-\zeta_{p+1}})$. As in part (3), we can show $\max_{k\in\A_{1n}}\frac{T_n(k)}{\sqrt{V_n(k)}} =  O_p(1)$ if $n^{1/2-\zeta_{p+1}'}\geq \max\{n^{1-\delta_{p+1}-\delta_{p+1}},n^{1-\delta_{p+2}{-}\zeta_{p+2}}\}$. If $ n^{1-\delta_{p+2}-\zeta_{p+2}}\geq\max\{n^{1-\delta_{p+1}-\zeta_{p+1}},n^{1/2-\zeta_{p+1}'}\}$, we can show $\max_{k\in\A_{1n}}\frac{T_n(k)}{\sqrt{V_n(k)}} =  o_p(n^{1/2-\delta_{p+2}})$ in the same way as in part (3) when  $ n^{1-\delta_{p+2}-\zeta_{p+2}}\geq\max\{n^{1/2-\delta_{p+1}},n^{1/2-\zeta_{p+1}'}\}$. The proof for the $ n^{1-\delta_{p+1}-\zeta_{p+1}}\geq\max\{n^{1-\delta_{p+2}-\zeta_{p+2}},n^{1/2-\zeta_{p+1}'}\}$ is analogous and we omit the details.

	\subsection{Three or more change points setting}		\label{sec_3cp}	
	
	For nonnegative integers $q,p,s$ such that $0\leq p+q\leq m_0-s$ and $s\geq 3$, assume there are $p+q+s$ change points with relative location $0<r^\ast_1<\cdots<r^\ast_p<b<r^\ast_{p+1}<\cdots<r^\ast_{p+s}<1-b<r^\ast_{p+s+1}<\cdots<r^\ast_{p+s+q}<1$. In this case, 
	\begin{align}
		\cP_n = &a_1\Delta^\ast_1{+}(a_1{+}a_2)\Delta^\ast_2{+}\cdots{+}({a_1{+}\cdots{+}a_p})\Delta^\ast_p{+}\Delta^\ast_{p+1}{+}\cdots{+}\Delta^\ast_{p+s}\nonumber \\
		&{+}(1{-}b_1)\Delta^\ast_{p+s+1}{+}(1{-}b_1{-}b_2) \Delta^\ast_{p+s+2}{+}\cdots{+}(1{-}b_1{-}\cdots{-}b_q)\Delta^\ast_{p+s+q},   \label{def_pn_3cp}
	\end{align}		
	where $\{a_i\}$ are defined in section \ref{sec_0cp}, $b_1 = \frac{r^\ast_{p+s+1}-1+b}{1-b},b_{q+1} = \frac{1-r^\ast_{p+s+q}}{1-b}$ and $b_i = \frac{r^\ast_{p+s+i}-r^\ast_{p+s+i-1}}{1-b}$ for $i=2,3,\dots,q$.
	Let 
	\begin{align}
		\A_{1n}=&\{k\geq \fb+1|k\leq k^\ast_{p+1}-\epsilon_{(p+1)n},\,k\in \mathcal{Z}\}, \nonumber \\
		\A_{in}=&\{k^\ast_{p+i-1}{+}\epsilon_{(p+i-1)n} k\leq k^\ast_{p+i}-\epsilon_{(p+i)n},\,k\in \mathcal{Z}\} \text{ for }i=2,3,\dots,s, \nonumber \\
		\A_{(s+1)n}=&\{k\leq n-\fb|k\geq k^\ast_{p+s}+\epsilon_{(p+s)n},\,k\in \mathcal{Z}\}. \nonumber 
	\end{align}
	We show that no change point can be selected inside $\cup_{i=1}^{s+1}\A_{in}$. To be specific, we proof the following lemma.
	
	\begin{lemma}\label{lemma_3cp}
		Suppose Assumption \ref{assump_wbs} holds, then
		\begin{align}
			\max_{k\in\cup_{i=1}^{s+1}\A_{in}}\frac{T_n(k)}{\sqrt{V_n(k)}} = o_p(n^{1/2-\underline {\delta}_{s} }),
		\end{align}
		where $\underline {\delta}_{s} = \min\{\delta_{p+1},\dots,\delta_{p+s}\}$.
	\end{lemma}
	
	\subsubsection{No change point can be selected on $\A_{1n}$}\label{sec_3cp_1}	
	
	For $k\in\A_{1n}$, let $\tilde \Delta_{p+1} = (1{-}b{-}r^\ast_{p+2})\Delta^\ast_{p+2}{+}\cdots{+} (1{-}b{-}r^\ast_{p+s})\Delta^\ast_{p+s},$ we have
	\begin{align}
		T_n(k) = & {-}\frac{(k{-}nb)n}{(n-2nb)^{3/2}}\big\li \bar Y_1'{-}\bar Y_3'{-}b\sqrt{n}\cP_n,(1{-}b{-}r^\ast_{p+1})\Delta^\ast_{p+1}{+}\tilde \Delta_{p+1}\big\ri_\cH \nonumber \\
		& {+}\big\li \bar Y_1'{-}\bar Y_3'{-}b\sqrt{n}\cP_n, \frac{1}{\sqrt{n-2nb}}\big[S_{nb+1,k}^{Y'}-\frac{k{-}nb}{n{-}2nb}S_{nb+1,n-nb}^{Y'}\big]\big\ri_\cH. \label{eq_3cp_1}
	\end{align}
	Similar to Section \ref{sec_1cp}, define $R_{1n}(k) = \frac{1}{(n{-}2nb)^{2}}\sum_{t=k^\ast_{p+1}+1}^{k^\ast_{p+2}}\Big\{S_{t,n-nb}{-}\frac{n{-}nb{-}t{+}1}{n{-}nb{-}k}S_{k+1,n-nb}\Big\}^2$ and $L_{n}(k) =\frac{1}{(n{-}2nb)^{2}}\sum_{t=nb+1}^{k}\Big\{S_{nb+1,t}{-}\frac{t{-}nb}{k{-}nb}S_{nb+1,k}\Big\}^2$, then
	\small\begin{align}		
		R_{1n}(k) = & \frac{1}{(n{-}2nb)^{2}}\sum_{t=k^\ast_{p+1}+1}^{k^\ast_{p+2}}\Big\{\big\li \bar Y_1'{-}\bar Y_3'{-}b\sqrt{n}\cP_n,S^{Y'}_{t,n-nb} {-}\frac{n{-}nb{-}t{+}1}{n{-}nb{-}k}S^{Y'}_{k+1,n-nb}    \big\ri_\cH  \nonumber \\
		&{+}  \frac{(n{-}nb{-}t{+}1)(k^\ast_{p+1}{-}k)}{n{-}nb{-}k}\big\li \bar Y'_1{-}\bar Y'_3{-}b\sqrt{n}\cP_n,  \Delta^\ast_{p+1}\big\ri_\cH{+}\frac{n(t{-}1{-}k)}{n{-}nb{-}k} \big\li\bar Y'_1{-}\bar Y'_3{-}b\sqrt{n}\cP_n,  \tilde \Delta_{p+1}  \big\ri_\cH    \Big\}^2 \nonumber \\
		=& \frac{1}{(n{-}2nb)^{2}}\sum_{t=k^\ast_{p+1}+1}^{k^\ast_{p+2}}\Big\{R_{11}(k,t){+}R_{12}(k,t){+}R_{13}(k,t) \Big\}^2,
	\end{align}	
	\begin{align}
		L_n(k) = & \frac{1}{(n{-}2nb)^{2}}\sum_{t=nb{+}1}^{k}\big\li \bar Y_1'{-}\bar Y_3'{-}b\sqrt{n}\cP_n,S^{Y'}_{nb+1,t} {-}\frac{t{-}nb}{k{-}nb}S^{Y'}_{nb+1,k}    \big\ri_\cH^2 \nonumber 
	\end{align}		\normalsize	
	$R_n(k) \geq R_{1n}(k)$ and $V_n(k) =R_n(k)+L_n(k)$. Note that it is either $\|\tilde \Delta_{p+1} \|_\cH = 0$ or $\tilde \Delta_{p+1} = \xi_{p+1} g'_{p+1}{+}g^n_{p+1}$ for some $g'_{p+1},g^n_{p+1}\in \cH$ such that $\xi_{p+1} \sim n^{-\delta_{p+1}'}$ for some $\delta_{p+1}'\in \{\delta_{p+2},\dots,\delta_{p+s}\}$ and $\|g^n_{p+1}\|_\cH =o(\xi_{p+1})$. If $\|\tilde \Delta_{p+1} \|_\cH = 0$, following the same proof as in Section \ref{sec_1cp}, we have $	\max_{k\in\A_{1n}}\frac{T_n(k)}{\sqrt{V_n(k)}} = o_p(n^{1/2-\delta_{p+1}})$. 
	
	Assume $\|\tilde \Delta_{p+1} \|_\cH \neq 0$, and $\li\cP_n,\tilde \Delta_{p+1}\ri_\cH=0$, following the same method as in part (1), (2) and (4) in Section \ref{sec_2cp_1} with $(1{-}b{-}r^\ast_{p+2})\Delta^\ast_{p+2}$ replaced by $\tilde \Delta_{p+1}$, we can show $\max_{k\in\A_{1n}}\frac{T_n(k)}{\sqrt{V_n(k)}} = o_p(n^{1/2-\underline \delta_{p+1}})$. If $\li\cP_n,\tilde \Delta_{p+1}\ri_\cH\neq 0$, similar to part (2) of Section \ref{sec_2cp_1}, we can write $\cP_n = \kappa_{p+1} f'_{p+1}{+}f^n_{p+1}$ for some $f'_{p+1},f^n_{p+1}\in \cH$ such that $\kappa_{p+1} \sim n^{-\zeta_{p+1}'}$, $\zeta_{p+1}'\in [0,1/2)$ and $\|f^n_{p+1}\|_\cH =o(\kappa_{p+1})$. Then we have $\li\cP_n,\tilde \Delta_{p+1}\ri_\cH \sim n^{-\delta^\ast_{p+1}-\zeta_{p+2}}$ for some $\zeta_{p+2}\in[\zeta_{p+1}',1/2)$ and $\delta^\ast_{p+1}\in[\delta_{p+1}',1/2)$. We consider the following different scenarios.		
	
	\noindent
	\textbf{(1) $\li\cP_n,\Delta^\ast_{p+1}\ri_\cH = 0$ and $\li\cP_n,\tilde \Delta_{p+1}\ri_\cH\neq0$.}		
	
	In this case, $\max_{k\in\A_{1n}}|T_n(k)| = O_p(n^{1/2-\delta_{p+1}}){+}O_p(n^{1/2-\delta_{p+1}'}){+} O_p(n^{1/2-\zeta_{p+1}'}){+} O_p(n^{1-\delta^\ast_{p+1}-\zeta_{p+2}})$. If $\max\{n^{1/2-\delta_{p+1}'},n^{1/2-\zeta_{p+1}'},n^{1-\delta^\ast_{p+1}-\zeta_{p+2}}\}> n^{1/2-\delta_{p+1}}$, $\min_{k\in\A_{1n}}R_{1n}(k)$ is of the same order as 
	$ \min_{k\in\A_{1n}} \frac{1}{(n{-}2nb)^{2}}\sum_{t=k^\ast_{p+1}+1}^{k^\ast_{p+2}}\Big\{R_{11}(k,t){+}R_{13}(k,t) \Big\}^2 = O_p^s(\max\{n^{1-2\delta_{p+1}'},n^{1-2\zeta_{p+1}'},n^{2-2\delta^\ast_{p+1}-2\zeta_{p+2}} \})$, so we have $\max_{k\in\A_{1n}}V_n(k)^{-1} = O_p(\min\{n^{-1+2\delta_{p+1}'},n^{-1+2\zeta_{p+1}'},n^{-2+2\delta^\ast_{p+1}+2\zeta_{p+2}} \})$, which implies $$\max_{k\in\A_{1n}}\frac{T_n(k)}{\sqrt{V_n(k)}} =  O_p(1).$$ If $\max\{n^{1/2-\delta_{p+1}'},n^{1/2-\zeta_{p+1}'},n^{1-\delta^\ast_{p+1}-\zeta_{p+2}}\}\leq n^{1/2-\delta_{p+1}}$, then $\min_{k\in\A_{1n}^1}R_{1n}(k) = O_p^s(n^{1/2-\delta_{p+1}})$ and $\min_{k\in\A_{1n}^2}L_n(k)  =O_p^s(n^{1-2\zeta_{p+1}'})$, where $\A_{1n}^1 = \A_{1n}\cap \{k\leq n\frac{b+r^\ast_{p+1}}{2}\}$ and $\A_{1n}^2 = \A_{1n}\cap \{k> n\frac{b+r^\ast_{p+1}}{2}\}$. Then we have $\max_{k\in\A_{1n}}V_n(k)^{-1} = o_p(1)$, which implies $\max_{k\in\A_{1n}}\frac{T_n(k)}{\sqrt{V_n(k)}} =  o_p(n^{1/2-\delta_{p+1}})$.	
	
	\noindent
	\textbf{(2) $\li\cP_n,\Delta^\ast_{p+1}\ri_\cH \neq 0$ and $\li\cP_n,\tilde \Delta_{p+1}\ri_\cH\neq0$.}		
	
	In this case, $\li\cP_n,\Delta^\ast_{p+1}\ri_\cH \sim n^{-\delta_{p+1}-\zeta_{p+1}}$ for some $\zeta_{p+1}\in[\zeta_{p+1}',1/2)$ and $\max_{k\in\A_{1n}}|T_n(k)| =  O_p(n^{1-\delta_{p+1}-\zeta_{p+1}}){+}O_p(n^{1/2-\delta_{p+1}'}){+} O_p(n^{1/2-\zeta_{p+1}'}){+} O_p(n^{1-\delta^\ast_{p+1}-\zeta_{p+2}})$. Following similar method as in part (4) of section \ref{sec_2cp_1} (replace $n^{1/2-\delta_{p+2}}$ with $\max\{n^{1/2-\delta_{p+1}'}, n^{1-\delta^\ast_{p+1}-\zeta_{p+2}}\}$), we can show $$\max_{k\in\A_{1n}}\frac{T_n(k)}{\sqrt{V_n(k)}} =  o_p(n^{1/2-\delta_{p+1}}).$$

	\subsubsection{No change point can be selected on $\A_{in}$ for $i=2,3,\dots,s$.}\label{sec_3cp_2}		
	
	Fix $i\in\{2,3,\dots,s\}$, denote $\tilde \Delta_{p+i-1} = (r^\ast_{p+1}{-}b)\Delta^\ast_{p+1}{+}\cdots{+}(r^\ast_{p+i-2}{-}b)\Delta^\ast_{p+i-2}$ (we set $\tilde \Delta_{p+i-1}=0$ if $i=2$) and $\tilde \Delta_{p+i} = (1{-}b{-}r^\ast_{p+i+1})\Delta^\ast_{p+i+1}{+}\cdots{+}(1{-}b{-}r^\ast_{p+s})\Delta^\ast_{p+s}$ (we set $\tilde \Delta_{p+i}=0$ if $i=s$) for $k\in\A_{in}$,
	\small\begin{align}
		T_n(k) = & {-}\frac{(n{-}nb{-}k)n}{(n-2nb)^{3/2}}\big\li \bar Y_1'{-}\bar Y_3'{-}b\sqrt{n}\cP_n,(r^\ast_{p+i-1}{-}b)\Delta^\ast_{p+i-1}{+}\tilde \Delta_{p+i-1} \big\ri_\cH\nonumber \\
		&{-}\frac{(k{-}nb)n}{(n-2nb)^{3/2}}\big\li \bar Y_1'{-}\bar Y_3'{-}b\sqrt{n}\cP_n, (1{-}b{-}r^\ast_{p+i})\Delta^\ast_{p+i}{+}\tilde \Delta_{p+i} \big\ri_\cH \nonumber \\
		& {+}\big\li \bar Y_1'{-}\bar Y_3'{-}b\sqrt{n}\cP_n, \frac{1}{\sqrt{n-2nb}}\big[S_{nb+1,k}^{Y'}-\frac{k{-}nb}{n{-}2nb}S_{nb+1,n-nb}^{Y'}\big]\big\ri_\cH. 
	\end{align}\normalsize
	Similar to Section \ref{sec_2cp_2}, define $R_{1n}(k) = \frac{1}{(n{-}2nb)^{2}}\sum_{t=k^\ast_{p+i}{+}1}^{\bar b_i}\Big\{S_{t,n-nb}{-}\frac{n{-}nb{-}t{+}1}{n{-}nb{-}k}S_{k+1,n-nb}\Big\}^2$ and $L_{1n}(k) = \frac{1}{(n{-}2nb)^{2}}\sum_{t=\underline b_i}^{k^\ast_{p+i-1}}\Big\{S_{nb+1,t}{-}\frac{t{-}nb}{k{-}nb}S_{nb+1,k}\Big\}^2$, where $\bar b_i=k^\ast_{p+i+1}$ if $i<s$ or $n-nb$ if $i=s$; $\underline b_i=k^\ast_{p+i-2}{+}1$ if $i>2$ or $nb{+}1$ if $i=2$. Then we have
	\small\begin{align}
		R_{1n}(k) = & \frac{1}{(n{-}2nb)^{2}}\sum_{t=k^\ast_{p+i}+1}^{\bar b_i}\Big\{\big\li \bar Y_1'{-}\bar Y_3'{-}b\sqrt{n}\cP_n,S^{Y'}_{t,n-nb} {-}\frac{n{-}nb{-}t{+}1}{n{-}nb{-}k}S^{Y'}_{k+1,n-nb}    \big\ri_\cH  \nonumber \\
		&\qquad\qquad{+} \frac{(n{-}nb{-}t{+}1)(k^\ast_{p+i}{-}k)}{n{-}nb{-}k}\big\li \bar Y'_1{-}\bar Y'_3{-}b\sqrt{n}\cP_n,  \Delta^\ast_{p+i}\big\ri_\cH{+}\frac{n(t{-}1{-}k)}{n{-}nb{-}k} \big\li\bar Y'_1{-}\bar Y'_3{-}b\sqrt{n}\cP_n,  \tilde \Delta_{p+i}  \big\ri_\cH    \Big\}^2 \nonumber \\
		=& \frac{1}{(n{-}2nb)^{2}}\sum_{t=k^\ast_{p+i}{+}1}^{\bar b_i}\Big\{R_{11}(k,t){+}R_{12}(k,t){+}R_{13}(k,t) \Big\}^2,
	\end{align}	
	\begin{align}
		L_{1n}(k) = & \frac{1}{(n{-}2nb)^{2}}\sum_{t=\underline b_i}^{k^\ast_{p+i-1}}\Big\{\big\li \bar Y_1'{-}\bar Y_3'{-}b\sqrt{n}\cP_n,S^{Y'}_{nb+1,t} {-}\frac{t{-}nb}{k{-}nb}S^{Y'}_{nb+1,k}    \big\ri_\cH  \nonumber \\
		&\qquad\qquad{-}\frac{(t{-}nb)(k{-}k^\ast_{p+i-1})}{k{-}nb}    \big\li \bar Y'_1{-}\bar Y'_3{-}b\sqrt{n}\cP_n, \Delta^\ast_{p+i-1} \big\ri_\cH   {-}\frac{n(k{-}t)}{k{-}nb}    \big\li \bar Y'_1{-}\bar Y'_3{-}b\sqrt{n}\cP_n, \tilde \Delta_{p+i-1} \big\ri_\cH  \Big\}^2 \nonumber \\
		=& \frac{1}{(n{-}2nb)^{2}}\sum_{t=\underline b_i}^{k^\ast_{p+i-1}}\Big\{L_{11}(k,t){+}L_{12}(k,t){+}L_{13}(k,t) \Big\}^2,
	\end{align}\normalsize	
	$R_n(k) \geq R_{1n}(k)$, $L_n(k) \geq L_{1n}(k)$ and $V_n(k) =R_n(k)+L_n(k)$. If $\tilde \Delta_{p+i} = \tilde \Delta_{p+i-1}=0$, then according to the proof in Section \ref{sec_2cp_2}, we can show $\max_{k\in\A_{in}}\frac{T_n(k)}{\sqrt{V_n(k)}} =  o_p(n^{1/2-\min\{\delta_{p+i-1},\delta_{p+i}\}})$.  We consider the following different scenarios.			
	
	\noindent
	\textbf{(1) $\tilde \Delta_{p+i} \neq \tilde \Delta_{p+i-1}=0$ or $\tilde \Delta_{p+i-1} \neq \tilde \Delta_{p+i}=0$.}			
	
	We only give the proof for	$\tilde \Delta_{p+i} \neq \tilde \Delta_{p+i-1}=0$ since the other case is similar. Assume $\tilde \Delta_{p+i} = \xi_{p+i} g'_{p+i}{+}g^n_{p+i}$ for some $g'_{p+i},g^n_{p+i}\in \cH$ such that $\xi_{p+i} \sim n^{-\delta_{p+i}'}$, $\delta_{p+i}'\in \{\delta_{p+i+1},\dots,\delta_{p+s}\}$ and $\|g^n_{p+i}\|_\cH =o(\xi_{p+i})$. 
	
	First, if $\cP_n=0$ and $\delta_{p+i}'<\min\{\delta_{p+i},\delta_{p+i-1}\}$, it is easy to see $\max_{k\in\A_{in}}\frac{T_n(k)}{\sqrt{V_n(k)}} =  O_p(1)$. If $\delta_{p+i}'\geq \min\{\delta_{p+i},\delta_{p+i-1}\}$, it suffice to show $\max_{k\in\A_{1n}}V_n(k)^{-1} = o_p(1)$, which can be done by noting that $\min_{k\in\A_{in}}\frac{1}{(n{-}2nb)^{2}}\sum_{t=\underline b_i}^{k^\ast_{p+i-1}}L_{12}(k,t)\stackrel{p}{\to}\infty$.
	
	Second, if $\cP_n\neq 0$, according to whether $\li\cP_n,\Delta\ri_\cH=0$ or not for $\Delta\in\{\Delta^\ast_{p+i-1},\Delta^\ast_{p+i},\tilde \Delta_{p+i}\}$, we can divide the proof into eight different scenarios and show $\max_{k\in\A_{in}}\frac{T_n(k)}{\sqrt{V_n(k)}} = o_p(n^{1/2-\underline {\delta}_{s} })$ in each scenario using similar methods as in Sections \ref{sec_1cp}, \ref{sec_2cp} and \ref{sec_3cp_1}. The details are omitted.
	
	\noindent
	\textbf{(2) $\tilde \Delta_{p+i} \neq 0$ and $ \tilde \Delta_{p+i-1}\neq 0$.}	
	
	Assume $\tilde \Delta_{p+i-1} = \xi_{p+i-1} g'_{p+i-1}{+}g^n_{p+i-1}$ for some $g'_{p+i-1},g^n_{p+i-1}\in \cH$ such that $\xi_{p+i-1} \sim n^{-\delta_{p+i-1}'}$, $\delta_{p+i-1}'\in \{\delta_{p+1},\dots,\delta_{p+i-2}\}$ and $\|g^n_{p+i-1}\|_\cH =o(\xi_{p+i-1})$. 	
	
	First, if $\cP_n=0$ and $\min\{\delta_{p+i-1}',\delta_{p+i}'\}<\min\{\delta_{p+i},\delta_{p+i-1}\}$, it is easy to see $\max_{k\in\A_{in}}\frac{T_n(k)}{\sqrt{V_n(k)}} =  O_p(1)$. If $\min\{\delta_{p+i-1}',\delta_{p+i}'\}\geq \min\{\delta_{p+i},\delta_{p+i-1}\}$, it suffice to show $\max_{k\in\A_{1n}}V_n(k)^{-1} = o_p(1)$. Assume $\delta_{p+i}\leq \delta_{p+i-1}$, then we have $\delta_{p+i}'\geq \delta_{p+i}$ and following similar method as in part (1) of Section \ref{sec_2cp_1} we can show $\max_{k\in\A_{1n}}V_n(k)^{-1} = o_p(1)$.
	
	If $\cP_n\neq 0$, the proof is divided into sixteen scenarios according to whether $\li\cP_n,\Delta\ri_\cH=0$ or not for $\Delta\in\{\Delta^\ast_{p+i-1},\Delta^\ast_{p+i},\tilde \Delta_{p+i-1},\tilde \Delta_{p+i}\}$. Since the method used in showing $\max_{k\in\A_{in}}\frac{T_n(k)}{\sqrt{V_n(k)}} = o_p(n^{1/2-\underline {\delta}_{s} })$ in each scenario is similar to those used in part (1) above, we omit the details. 
%\end{comment}
\end{document}